%% file: stpd_ICALP.tex
\documentclass[a4paper,UKenglish,cleveref,thm-restate]{lipics-v2021}

\usepackage{float}
\usepackage{graphicx}
\usepackage{xcolor}
\usepackage{amsmath}
\usepackage{amssymb}
\usepackage[sort,nocompress]{cite}
\usepackage[ruled,vlined,lined,boxed,commentsnumbered]{algorithm2e}
\usepackage{comment}
\usepackage[normalem]{ulem}
\usepackage{paralist}
\usepackage{hyperref}

\usepackage{thmtools}
\usepackage{thm-restate}

\renewcommand{\paragraph}[1]{\vspace{-1ex}\subparagraph*{#1}}

\newcommand{\stlex}{\mathtt{st\mbox{-}lex}}
\newcommand{\stlexp}{\mathtt{st\mbox{-}lex}^+}
\newcommand{\stlexm}{\mathtt{st\mbox{-}lex}^-}
\newcommand{\stcolex}{\mathtt{st\mbox{-}colex}}
\newcommand{\stcolexp}{\mathtt{st\mbox{-}colex}^+}
\newcommand{\stcolexm}{\mathtt{st\mbox{-}colex}^-}
\newcommand{\stpos}{\mathtt{st\mbox{-}pos}}
\newcommand{\stposp}{\mathtt{st\mbox{-}pos}^+}
\newcommand{\stposm}{\mathtt{st\mbox{-}pos}^-}

\newcommand{\sA}{\ensuremath{\mathrm{sA}}}
\newcommand{\SA}{\ensuremath{\mathrm{SA}}}
\newcommand{\PA}{\ensuremath{\mathrm{PA}}}
\newcommand{\ISA}{\ensuremath{\mathrm{ISA}}}
\newcommand{\IPA}{\ensuremath{\mathrm{IPA}}}
\newcommand{\LPF}{\ensuremath{\mathrm{LPF}}}
\newcommand{\BWT}{\ensuremath{\mathrm{BWT}}}
\newcommand{\coBWT}{\ensuremath{\mathrm{coBWT}}}

\newcommand{\SRC}{\ensuremath{\mathrm{SRC}}}

\newcommand{\PDA}{\ensuremath{\mathrm{PDA}}}
\newcommand{\sd}{\ensuremath{\mathrm{sd}}}
\newcommand{\lce}{\ensuremath{\mathrm{lce}}}
\newcommand{\llce}{\ensuremath{\mathrm{llce}}}
\newcommand{\rlce}{\ensuremath{\mathrm{rlce}}}
\newcommand{\lcs}{\ensuremath{\mathrm{lcs}}}
\newcommand{\lcp}{\ensuremath{\mathrm{lcp}}}

\renewcommand{\root}{\ensuremath{\mathit{root}}}
\newcommand{\child}{\ensuremath{\mathit{child}}}
\newcommand{\first}{\ensuremath{\mathit{first}}}
\newcommand{\succn}{\ensuremath{\mathit{succ}}}
\newcommand{\labeln}{\ensuremath{\mathit{label}}}
\newcommand{\lleaf}{\ensuremath{\mathit{lleaf}}}
\newcommand{\rleaf}{\ensuremath{\mathit{rleaf}}}
\renewcommand{\next}{\ensuremath{\mathit{next}}}
\newcommand{\locate}{\ensuremath{\mathit{locate}}}
\newcommand{\sdepth}{\ensuremath{\mathit{sdepth}}}
\newcommand{\ancestor}{\ensuremath{\mathit{ancestor}}}
\newcommand{\isleaf}{\ensuremath{\mathit{isleaf}}}

\newcommand{\sufsearch}{\ensuremath{\mathrm{sufsearch}}}

\newcommand{\LCP}{\ensuremath{\mathrm{LCP}}}
\newcommand{\PLCP}{\ensuremath{\mathrm{PLCP}}}

\newcommand{\locus}{\ensuremath{\mathrm{locus}}}
\newcommand{\out}{\ensuremath{\mathrm{out}}}
\newcommand{\argmin}{\mathrm{argmin}}
\newcommand{\argmax}{\mathrm{argmax}}

\title{Compressing Suffix Trees by Path Decompositions}

\author{Ruben Becker}{Ca' Foscari University of Venice, Italy}{rubensimon.becker@unive.it}{https://orcid.org/0000-0002-3495-3753}{}

\author{Davide Cenzato}{Ca' Foscari University of Venice, Italy}{davide.cenzato@unive.it}{https://orcid.org/0000-0002-0098-3620}{}

\author{Travis Gagie}{Dalhousie University, Halifax, Nova Scotia, Canada}{Travis.Gagie@dal.ca}{https://orcid.org/0000-0003-3689-327X}%
  {Funded by NSERC (Discovery Grant RGPIN-07185-2020).}

\author{Ragnar {Groot Koerkamp}}{ETH Zurich, Switzerland \and Karlsruhe Institute of Technology, Germany}{ragnar.grootkoerkamp@gmail.com}{https://orcid.org/0000-0002-2091-1237}{}

\author{Sung-Hwan Kim}{Ca' Foscari University of Venice, Italy}{sunghwan.kim@unive.it}{https://orcid.org/0000-0002-1117-5020}{}

\author{Giovanni Manzini}{University of Pisa, Italy}{giovanni.manzini@unipi.it}{https://orcid.org/0000-0002-5047-0196}%
  {Funded by the Next Generation EU PNRR MUR M4 C2 Inv 1.5 project ECS00000017 Tuscany Health Ecosystem Spoke 6 CUP I53C22000780001.}

\author{Nicola Prezza}{Ca' Foscari University of Venice, Italy}{nicola.prezza@unive.it}{https://orcid.org/0000-0003-3553-4953}{}

\authorrunning{R. Becker et al.}

\Copyright{Ruben Becker, Davide Cenzato, Travis Gagie, Ragnar Groot Koerkamp, Sung-Hwan Kim, Giovanni Manzini, and Nicola Prezza}



\relatedversiondetails[cite={becker2025stpd}]{Full Version}{https://arxiv.org/abs/2506.14734}

\supplementdetails[subcategory={Source Code}]{Software}{https://github.com/regindex/STPD-index}

\funding{\textit{Ruben Becker, Davide Cenzato, Sung-Hwan Kim, and Nicola Prezza}: Funded by the European Union (ERC, REGINDEX, 101039208). Views and opinions expressed are however those of the author(s) only and do not necessarily reflect those of the European Union or the European Research Council Executive Agency. Neither the European Union nor the granting authority can be held responsible for them.}

\nolinenumbers

\EventEditors{Sayan Bhattacharya, Danupon Nanongkai, Michael Benedikt, and Gabriele Puppis}
\EventNoEds{4}
\EventLongTitle{53rd International Colloquium on Automata, Languages, and Programming (ICALP 2026)}
\EventShortTitle{ICALP 2026}
\EventAcronym{ICALP}
\EventYear{2026}
\EventDate{July 7--10, 2026}
\EventLocation{Royal Holloway, University of London, Egham, United Kingdom}
\EventLogo{}
\SeriesVolume{374}
\ArticleNo{119}

\begin{document}
	
	\maketitle

\begin{abstract}
    The suffix tree is arguably the most fundamental data structure on strings: introduced by Weiner (SWAT 1973) and McCreight (JACM 1976), it allows solving a myriad of computational problems on strings in linear time. Motivated by its large space usage, subsequent research focused first on reducing its size by a constant factor via Suffix Arrays, and later on reaching space proportional to the size of the compressed string. Modern compressed indexes, such as the $r$-index (Gagie et al., SODA 2018), fit in space proportional to $r$, the number of runs in the Burrows-Wheeler transform (a strong and universal repetitiveness measure). These advances, however, came with a price: while modern compressed indexes boast optimal bounds in the RAM model, they are often orders of magnitude slower than uncompressed counterparts in practice due to catastrophic cache locality. This reality gap highlights that Big-O complexity in the RAM model has become a misleading predictor of real-world performance, leaving a critical question unanswered: can we design compressed indexes that are efficient in the I/O model of computation? 
    
    We answer this in the affirmative by introducing a new Suffix Array sampling technique based on particular path decompositions of the suffix tree. We prove that sorting the suffix tree leaves by specific priority functions induces a decomposition where the number of distinct paths (each corresponding to a string suffix) is bounded by $r$. This allows us to solve indexed pattern matching efficiently in the I/O model using a Suffix Array sample of size at most $r$, strictly improving upon the (tight) $2r$ bound of Suffixient Arrays, another recent compressed Suffix Array sampling technique.
    
    Experiments confirm that this theoretical I/O efficiency translates to practice in pangenomic applications: our index locates pattern occurrences using less space and orders of magnitude less time than the $r$-index when performing pattern matching on repetitive DNA collections. Beyond this, our contributions are twofold: (i) unlike Suffixient Arrays, our technique supports most standard suffix tree operations in $O(r)$ space on top of the text while matching the I/O complexity of uncompressed suffix trees; and (ii) we establish a general framework where any valid path decomposition induces a Suffix Array sampling whose size is a new strong repetitiveness measure; we provide a universal mechanism for locating all pattern occurrences for each such path decomposition. 
\end{abstract}
    
	\section{Introduction}\label{sec:intro}
	
	In this paper, we describe a new elegant and very efficient paradigm to solve the well-studied problem of compressing suffix trees. Suffix trees were introduced in 1973 by Weiner \cite{weiner1973linear} and revisited (in their modern form) in 1976 by McCreight \cite{McCreight76} to solve string processing problems such as finding longest common substrings and matching patterns on indexed text in linear time. 
	The latter problem (indexed string matching) asks to build a data structure on a text $\mathcal T$ of length $n$ over alphabet of size $\sigma$ so that later (at query time), given a string $P$ (the pattern) of length $m\le n$, the following can be returned: (1) one exact occurrence $\mathcal T[i,j] = P$ of $P$ in $\mathcal T$ if any exists (\emph{find} queries), (2) all exact occurrences  of $P$ in $\mathcal T$ (\emph{locate} queries), or (3) the number of exact occurrences of $P$ in $\mathcal T$ (\emph{count} queries).
	While being extremely fast due to excellent query-time cache locality, suffix trees require a linear number of words to be stored in memory regardless of the input compressibility and are therefore not suitable to nowadays massive-data scenarios such as pan-genome indexing (where one aims at indexing terabytes of data in the form of repetitive collections of thousands of genomes). This problem was later mitigated by Suffix Arrays \cite{ManberM90,manber1993suffix,gonnet1992new}, which use only a constant fraction of the space of suffix trees while supporting (cache-efficiently) a subset of their functionality, still sufficient to support pattern matching queries. 
	
    Subsequent research succeeded (spectacularly) in reducing the space usage of suffix trees and Suffix Arrays to the bare minimum needed to store the \emph{compressed text}. Notable contributions in this direction include the compressed Suffix Array (CSA) \cite{GV00}, FM-index \cite{FM00}, run-length compressed Suffix Array \cite{makinen2010storage}, run-length FM-index \cite{MN05}, $r$-index \cite{GNP20}, Lempel-Ziv-based \cite{KNtcs12} and indexes based on straight-line programs (SLPs) \cite{claude2012improved} (and their variants), and the more recent $\delta$-SA \cite{DBLP:conf/focs/KempaK23}. 
	See the survey of Navarro~\cite{Nav22b} for an extensive treatment of the subject. 
	While the first line of work (CSA and FM-index) focused on entropy compression, the subsequent works mentioned above switched to text compressors capable of exploiting the \emph{repetitiveness} of the underlying text sequence (a source of redundancy that entropy compression is not able to exploit).  Among those, the $r$-index \cite{rindex-soda,GNP20} and its improvements \cite{Nishimoto21Move,BertramMoveR24,zakeri2024movi} stood out for its optimal linear-time pattern matching query time and its size --- linear in the number $r$ of equal-letter runs in the Burrows-Wheeler transform (BWT) of the text. 
    These works on repetition-aware compressed text indexes spurred a very fruitful line of research on compressibility measures (the survey of Navarro \cite{Nav22a} covers the subject in detail), culminating in recent breakthroughs \cite{kempa2018roots,kempa2020resolution} which showed that $r$ is a strong and universal repetitiveness measure, being equivalent to all other known compressibility measures (such as the size of the Lempel-Ziv factorization \cite{LempelZ76}, normalized substring complexity~\cite{kociumaka2022toward}, and straight-line programs) up to a multiplicative polylogarithmic factor. Altogether, these results laid the theoretical foundations for subsequent works in computational pan-genomics that showed how the run-length encoded BWT can successfully be used to index very large collections of related genomes in compressed space~\cite{Moni22,Cozzi23,Rossi22,ahmed2023spumoni,shivakumar2024sigmoni, ahmed2021pan}. 
    
	\subsection{Are compressed data structures just a theoretical tool?} 
	
Modern compressed indexes such as the $\delta$-SA of Kempa and Kociumaka \cite{DBLP:conf/focs/KempaK23} support random access and pattern matching queries, but their time complexities depend on a high polynomial of the logarithm of the text's length, which makes them hardly practical.
    Indexes based on the Lempel-Ziv factorization or on grammar compression mitigate this problem (even reaching optimal search time \cite{ChristiansenEKN21}), but rely on complex and cache-inefficient data structures that, again, make them orders of magnitude slower than simple suffix trees in practice.
    The $r$-index (in its modern version \cite{Nishimoto21Move}) uses $O(r)$ words of space and solves \emph{find}, \emph{locate}, and \emph{count} queries in $O(m)$ time (assuming constant alphabet for simplicity), plus the number of occurrences to be reported (if any). While this is essentially the end of the story in the word RAM model,
    it does not take into account caching effects. As a matter of fact, each of the $O(m)$ steps of the \emph{backward search} algorithm of the $r$-index and of its predecessor (the FM-index \cite{FM00}) triggers I/O operations (that is, likely cache misses). 
    While this issue was later partially addressed by the \emph{move} structure of Nishimoto and Tabei \cite{Nishimoto21Move},
    that solution still triggers $O(m)$ cache misses. 
    This does not happen with the suffix tree:
    
    \begin{remark}\label{remark:ST cache}
        The suffix tree of Weiner~\cite{weiner1973linear} and McCreight~\cite{McCreight76} uses $O(n)$ words on top of the plain text ($n\log\sigma$ bits) and allows locating all the $occ$ occurrences of any pattern $P$ of length $m$ with $O(d + m/B + occ)$ I/O complexity, where $d$ is the node depth of $P$ in the suffix tree and $B$ is the number of integers fitting in an I/O block.
    \end{remark}

    To see this, observe that path compression (i.e. each suffix tree edge is represented as a pair of pointers to the text) makes it possible to compare (a substring of) the pattern with the label of an edge of length $\ell$ with $O(\ell/B + 1)$ I/O complexity. In particular, each of the $d$ edge traversals triggers at least one I/O operation in the worst case.
    The additive term $d$ is negligible in practice as in most interesting scenarios the suffix tree tends to branch mostly in the highest levels (we investigate this effect in Section \ref{sec: experiment ST vs rindex}). For instance, if the text is uniform then $d \in \Theta(\log n)$ with high probability since the longest repeated substring's length is $\Theta(\log n)$ w.h.p.
    While Remark \ref{remark:ST cache} reflects the original suffix tree design \cite{weiner1973linear,McCreight76} (locating pattern occurrences by subtree navigation), it is worth noticing that augmenting the suffix tree with the Suffix Array $\SA$ \cite{manber1993suffix,gonnet1992new}, the $occ$ term gets reduced to $occ/B$ since the $occ$ pattern occurrences occur contiguously in $\SA$. 
    In this paper, we 
    stick with the original suffix tree design of Remark \ref{remark:ST cache} 
    since we are mostly interested in matching the $m/B$ term, which becomes dominant as the pattern length increases asymptotically ($occ$, on the other hand, does not increase when appending new characters to a given query pattern).
    
    The performance gap in the I/O model between compressed indexes and the suffix tree shows up dramatically in practice: a simple experiment (see Section \ref{sec: experiment ST vs rindex}) shows that, while the $r$-index is orders of magnitude smaller than the suffix tree on very repetitive inputs, it also solves queries \emph{orders of magnitude slower}. 
    The same holds for all the existing compressed indexes using a space close to that of the $r$-index.
    This reality gap highlights that Big-O complexity in the RAM model has become a misleading predictor of real-world performance in the context of compressed data structures, leaving a critical question unanswered:  
	
	\begin{center}
		\emph{
		Can compressed indexes be efficient in the I/O model of computation?
		}
	\end{center}

\subsection{Our contributions}\label{sec:overview}

    We answer this question in the affirmative by introducing a novel compressed sampling of the Prefix Array (the colexicographically sorted array of all text prefixes). In one configuration, our sampling requires at most $r$ samples to resolve pattern matching queries via a simple binary search strategy, assuming random access to the text is available. 
    Since random access in compressed space is a well-studied problem \cite{BilleLRSSW15,KNtcs12,BCGGKNOPTjcss20,kuruppu2010relative} (see also \cite{Nav22a} for a survey on the topic), this requirement does not limit our strategy. On the contrary, unlike existing compressed indexes, our approach is flexible enough to utilize any random access compressed data structure.
    
    By incorporating additional data structures, we achieve our main result: a compressed suffix tree topology occupying only $O(r)$ words of space on top of a (potentially compressed) text oracle that still efficiently supports most standard navigation queries:

    \begin{theorem}\label{thm:st operations}
    Let $\mathcal T$ be a text of length $n$ over an integer alphabet of size $\sigma$.
        Assume we have access to an oracle supporting longest common extension ($\lce$) and random access queries (extraction of one character) on $\mathcal T$ in $O(t)$ time. Then, there is a representation of $\mathcal T$'s suffix tree using $O(r)$ words on top of the text oracle and supporting these queries:
        \begin{itemize}
            \item $\root()$ in $O(1)$ time: the suffix tree root. 
            \item $\child(u,a)$ in $O(t\log r + \log\sigma)$ time: the child of node $u$ by letter $a$.
            \item $\first(u)$ in $O(\log\sigma)$ time: the alphabetically-smallest label among the outgoing edges of $u$. 
            \item $\succn(u,a)$ in $O(\log\sigma)$ time: given node $u$ and a character $a$ labeling one of the outgoing edges of $u$, return the successor of $a$ (in alphabetic order) among the characters labeling outgoing edges of $u$ (return $\bot$ if no such label exists).
            \item $\labeln(u,v)$ in $O(1)$ time: given an edge $(u,v)$, return $(i,j) \in [n]^2$ such that the edge's label is $\mathcal T[i,j]$.
            \item $\lleaf(u)$, $\rleaf(u)$ in $O(1)$ time: the leftmost/rightmost leaves of the subtree rooted in a given node $u$.
            \item $\next(u)$: if $u$ is a leaf, return the next leaf in lexicographic order; we support following a sequence of $k$ leaf pointers in $O(\log\log (n/r) + k)$ time.
            \item $\locate(u)$ in $O(1)$ time: the text position $i$ of a given leaf (representing suffix $\mathcal T[i,n]$).
            \item $\sdepth(u)$ in $O(1)$ time: the string depth of node $u$.
            \item $\ancestor(u,v)$ in $O(t)$ time: whether $u$ is an ancestor of $v$.
        \end{itemize}
        If the text oracle also supports computing a collision-free (on text substrings) hash $\kappa(\mathcal T[i,j])$ of any text substring in $O(h)$ time (an operation we call \emph{fingerprinting}), then $\child(u,a)$ can be supported in $O(t + h\log n + \log\sigma)$ time within the same asymptotic space. 

        The same asymptotic bounds apply to I/O complexity if $t$ and $h$ represent I/O costs rather than time.
    \end{theorem}

    A subset of the above queries suffices to navigate the suffix tree (a task at the core of several string-processing algorithms) and to solve pattern matching queries.
    For example, using the cache-efficient text representation of Prezza~\cite{prezzaTalg21}, supporting $\lce$ with $O(t) = O(\log n)$ I/O complexity, fingerprinting with $O(h)=O(1)$ I/O complexity, and extraction of $\ell$ contiguous characters with $O(1+\ell/B)$ I/O complexity on polynomial alphabets, we obtain the following corollary:

    \begin{restatable}{corollary}{locatingOurST} \label{cor:locating on our ST}
        Let $\mathcal T$ be a text of length $n$ over an integer alphabet of size $\sigma$.
        The topology of $\mathcal T$'s suffix tree can be compressed in $O(r)$ words on top of a text representation~\cite{prezzaTalg21} of $n\log\sigma + O(\log n)$ bits so that all the $occ$ occurrences of any pattern $P$ of length $m$ can be located with $O(d\log n + m/B + occ)$ I/O complexity, where $d$ is the node depth of $P$ in the suffix tree and $B$ is the number of integers fitting in an I/O block.
    \end{restatable}

    That is, the same I/O complexity of Weiner and McCreight's suffix tree up to a logarithmic factor multiplying $d$ (see Remark~\ref{remark:ST cache}).
    At the same time, the space usage on top of the text is reduced from $O(n)$ to $O(r)$ words ($r$ is orders of magnitude smaller than $n$ on very repetitive inputs~\cite{GNP20}).
    Furthermore, we show that the term $d\log n$ can be replaced by $d\log m$ with a different technique (Theorem \ref{thm:locate all general}; this is important since in typical applications, $m\ll n$ holds).

    Using up-to-date cache-efficient compressed data structures, we show experimentally that an optimized implementation of our fully-compressed index is simultaneously \emph{smaller} and \emph{orders of magnitude faster} than the $r$-index on the task of locating \emph{all} pattern occurrences on a highly repetitive collection of genomes.
    

    \subsection{Improvements over related techniques}

\subsubsection{Suffixient Array}
Our approach is most closely related to the \emph{Suffixient Array}~\cite{cenzato2025suffixientarraysnewefficient}, another recent compressed sampling of the Prefix Array. While Suffixient Arrays pioneered this direction by showing that $O(r)$ samples are sufficient for pattern matching queries, our work reveals their theoretical and practical sub-optimality. We strictly improve upon them in three dimensions:

\begin{itemize}
    \item \textbf{Space Bounds.} The Suffixient Array requires up to $2r$ samples~\cite{navarro2025smallestsuffixientsetsrepetitiveness}, a bound recently proven to be worst-case tight up to a small additive constant~\cite{date2025neartightnesschileq2r}. We prove this is unnecessary: our technique demonstrates that $r$ samples are sufficient for pattern matching, effectively halving the worst-case space requirement.
    
    \item \textbf{Functionality.} Unlike Suffixient Arrays, which are limited to pattern matching, our sampling retains the topological structure of the suffix tree. Consequently, we support a much broader range of suffix tree operations efficiently.
    
    \item \textbf{Locating Efficiency.} Suffixient Arrays return an \emph{arbitrary} pattern occurrence, offering no control over which one is found. To locate \emph{all} occurrences starting from an arbitrary position $SA[i]$, one requires bidirectional navigation (both $SA[i-1]$ and $SA[i+1]$), costing about $4r$ words of auxiliary space~\cite{GNP20}. In contrast, our technique returns the occurrence minimizing a user-defined priority function $\pi$. By setting $\pi$ to the lexicographic order, our index guarantees returning the \emph{first} pattern occurrence in Suffix Array order. This allows us to locate all remaining occurrences using only unidirectional successor Suffix Array queries (i.e., retrieving $SA[i+1]$ given $SA[i]$), which requires only about $2r$ words of auxiliary space~\cite{GNP20}.
\end{itemize}

\subsubsection{Compressed Suffix Trees.} Gagie et al.~\cite{GNP20} previously addressed compressed suffix trees in space bounded by a function of $r$, showing that full functionality is possible in $O(r\log(n/r))$ space with logarithmic operation time. When augmented with a I/O efficient text oracle, their theoretical I/O complexity for pattern matching matches our Corollary~\ref{cor:locating on our ST}; however, their approach consumes significantly more space on top of the oracle ($O(r\log(n/r))$ vs $O(r)$). Furthermore, their structure relies on complex machinery atop a grammar-compressed Suffix Array; to our knowledge, it has never been implemented, likely due to the practical reality that Suffix Arrays do not compress as effectively as the text itself via grammar compression.

\subsubsection{Other I/O Efficient Solutions.} 

Other works in the literature addressed the I/O-efficiency of full-text indexes.  Chien et al.~\cite{chien2008geometric} introduced the Geometric Burrows-Wheeler Transform, which connects range searching with text indexing to provide compressed indexing techniques and theoretical lower bounds within the I/O model. Their solution, however, achieves I/O-efficient performance only in the uncompressed setting. Similarly, Ferragina and Venturini~\cite{ferragina2016compressed} proposed a compressed version of the string B-tree, adapting classical external-memory string structures to use compressed space while remaining efficient across memory hierarchies. While their compressed cache-oblivious string B-Tree supports I/O-efficient prefix search and enumeration, its space for \emph{locate} queries remains proportional to the number of suffixes.

As far as fully-compressed space is concerned, several recent works have targeted the I/O bottleneck of the $r$-index. The \emph{Move} structure of Nishimoto and Tabei~\cite{Nishimoto21Move} reduces the I/O operations of the $r$-index by a $\log\log n$ factor by replacing predecessor queries with pointers. However, it retains a worst-case complexity of $O(m+occ)$ I/O operations. While practical implementations~\cite{zakeri2024movi,BertramMoveR24} are roughly an order of magnitude faster than the standard $r$-index, they require significantly more space. Similarly, Puglisi and Zhukova~\cite{puglisi2020relative} applied relative Lempel-Ziv compression to a transformed Suffix Array. While this yields good memory locality and query speeds in practice, it lacks formal guarantees and increases the space usage of the $r$-index by an order of magnitude, again suffering from the poor compressibility of the Suffix Array relative to the text.

    \subsection{Overview of our techniques.}

    Our work can be interpreted as a generalization of suffix sorting: by sorting the text's suffixes according to any priority function (permutation) $\pi: [n] \rightarrow [n]$ satisfying a natural and desirable \emph{order-preserving} property (see Definition \ref{def:order-preserving pi}), 
    we obtain a compressed index that returns the pattern occurrence $\mathcal T[i,j]$ minimizing $\pi(i)$ among all pattern occurrences. Figure~\ref{fig:STLEX} broadly introduces our solution. Intuitively, we (i) sort the suffix tree's leaves according to $\pi$, (ii) we build a \emph{suffix tree path decomposition} (STPD for brevity) prioritizing the leftmost paths (i.e., smaller $\pi$) in such an order, and (iii) we path-compress the STPD paths by just recording their starting position in the text. At this point, we show that the colexicographically-sorted \emph{Path Decomposition Array} $\PDA$ of those positions (a sample of the Prefix Array) can be used to obtain a new elegant, simple, and remarkably efficient compressed suffix tree.
    
    This procedure can be formalized precisely as follows. In Section \ref{sec:order preserving STPD} we observe that each order-preserving STPD (i.e., an STPD using an order-preserving $\pi$) is associated with a Longest Previous Factor array $\LPF_\pi$ storing in entry $\LPF_\pi[i]$ the length $k$ of the longest factor $\mathcal T[j,j+k-1] = \mathcal T[i,i+k-1]$ occurring in a position $j$ with $\pi(j)<\pi(i)$. This array generalizes the well-known Permuted Longest Common Prefix Array $\PLCP$ (obtained by taking $\pi = \ISA$, the Inverse Suffix Array), and the Longest Previous Factor array $\LPF$ (obtained by taking $\pi = id$, the identity permutation). 
    Then, the positions at the beginning of each STPD path are $\{i+\LPF_\pi[i]\ :\ i\in [n]\}$. Array $\PDA$ is obtained by colexicographically-sorting such set, whose size is the number of \emph{irreducible} values in $\LPF_\pi$, that is, positions $i$ such that either $i=1$ or $\LPF_\pi[i] \neq \LPF_\pi[i-1]-1$ (a terminology borrowed from the literature on the $\PLCP$ array). This is interesting because the number of irreducible values in $\LPF_\pi$ is known to be at most $r$ when $\pi = \ISA$ (and in practice consistently smaller than that, as we show experimentally). More in general, by taking $\pi$ to be any order-preserving permutation, we obtain a new compressibility measure ($|\PDA_\pi|$) generalizing the number of irreducible $\PLCP$ values.
    In Theorem \ref{thm: compressing PDA} we show that such repetitiveness measure is reachable, meaning that $O(|\PDA_\pi|)$ words of space are sufficient to compress the text. 
    We provide a general mechanism of space $O(|\PDA_\pi|)$ (on top of the text oracle) for locating all pattern occurrences on any order-preserving $\pi$. When $\pi$ is the lexicographic rank of suffixes ($\pi=\ISA$) or the colexicographic rank of prefixes ($\pi = \IPA$, the Inverse Prefix Array), we describe more efficient (both in theory and practice) strategies for locating all pattern occurrences. We give more details below.

    \begin{figure}[h!]
		\centering
        \includegraphics[width=0.6\linewidth,page=2]{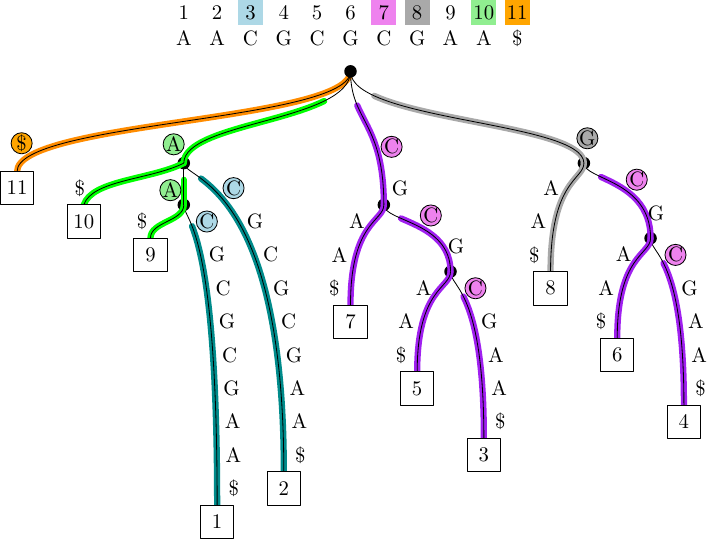}
        
		\caption{
            Overview of our technique. 
		We sort $\mathcal T$'s suffixes $\mathcal T[i,n]$ (equivalently, suffix tree leaves) by increasing $\pi(i)$. 
		In this example, $\pi$ corresponds to the standard lexicographic order of the text's suffixes (but $\pi$ can be more general).
		This induces a \emph{suffix tree path decomposition} (an edge-disjoint set of node-to-leaf paths covering all edges) obtained by always following the leftmost path. 
		At this point, we associate each path with the integer $i$ such the path's label is $\mathcal T[i,n]$ (in the figure, we also color each path according to the color of the associated position $i$). 
        In particular, the label from the root to the first path's edge is $\alpha_{i,k} = \mathcal T[i-k+1,i]$, where $k$ is the string depth of the first path's edge. Our indexing strategy essentially consists in implicitly colexicographically-sorting strings $\alpha_{i,k}$ in a subset of the Prefix Array called the \emph{Path Decomposition Array} $\PDA$: the colexicographically-sorted array containing (without duplicates) the first position of each path (in our example: $\PDA = [11,10,3,7,8]$). 
        Observe that, in this example, there are few (five) \emph{distinct} path labels: indeed, we show that $|\PDA|$ is bounded by universal compressibility measures and that it can be used to support basic suffix tree navigation and pattern matching operations. 
		}
		\label{fig:STLEX}
	\end{figure}

    \paragraph{Lexicographic rank.}
    We start in Section \ref{sec:stlex} with perhaps the most natural order-preserving permutation $\pi$: the lexicographic rank among the text's suffixes (see Figure \ref{fig:STLEX}). We prove that the corresponding $\PDA$ array --- called $\stlexm$ --- has size $|\stlexm| \le r$. In contrast, the Suffixient Array ~\cite{cenzato2025suffixientarraysnewefficient} has at most $2r$ samples~\cite{navarro2025smallestsuffixientsetsrepetitiveness}, a bound recently proven to be tight up to a small additive constant~\cite{date2025neartightnesschileq2r}.
    At this point, we consider the dual order-preserving permutation $\bar\pi(i) = n-\pi(i)+1$, yielding an STPD that always chooses the lexicographically-largest suffix tree paths. We prove that the same bound holds on the size of the corresponding array: $|\stlexp| \le r$. 
    We then show that the colexicographically-sorted array $\stlex = \{i-1\ :\ i\in \stlexm \cup \stlexp \cup \{n+1\}\}$, in addition to longest common extension queries on the text, naturally supports several suffix tree operations, including descending to children.
    Intuitively, we design a constant-size representation $R_u$ of any suffix tree node $u = \locus(\alpha)$ (i.e. $u$ is the node reached by reading $\alpha\in \Sigma^*$ from the suffix tree root) formed by the colexicographic range of $\alpha$ in $\stlex$, the lexicographically smallest and largest text suffixes being prefixed by $\alpha$, and the node's string depth $|\alpha|$. Given a character $c\in \Sigma$, the representation of $u$ combined with $\stlex$ and Range Minimum/Maximum Query (RMQ) data structures of $O(|\stlex|)$ bits allow us to find the lexicographically smallest and largest text suffixes being prefixed by $\alpha\cdot c$. Then, the longest common prefix between those two suffixes (found by a longest common extension query) is precisely the length of the suffix tree edge reached by reading $\alpha\cdot c$ from the root. This allows us to find the remaining characters $\beta \in \Sigma^*$ labeling that edge, meaning that the end of the edges is reached by reading $\alpha\cdot c\cdot \beta$ from the root. In turn, this allows us to reconstruct the representation of node $\child(x,c)$ via binary search on $\stlex$ and longest common extension queries.
    We note that our representation works only for \emph{explicit} suffix tree nodes, since strings $\alpha$ corresponding to implicit nodes do not suffix any string in $\stlex$.
        
    The procedure outlined above suffices to identify the suffix tree locus of pattern $P$. By traversing the subtree rooted at this node, we can report all $occ$ occurrences of $P$, taking $O(\log n)$ time per occurrence for a total of $O(occ \cdot \log n)$.
    Notably, this strategy operates independently of the $r$-index's $\phi$ function \cite{GNP20}, which was previously the only known technique for supporting locate queries within $r$-bounded space.
    We reduce this time to $O(\log\log (n/r) + occ)$ as follows.  
    Getting the leftmost/rightmost leaves of the subtree rooted in a given node and returning the beginning $i$ of the suffix $\mathcal T[i,n]$ corresponding to a given leaf are easily supported in constant time each thanks to the particular node representation that we adopt. Finally, leaf pointers (connecting leaves in lexicographic order) are supported in the claimed running time using techniques borrowed from the $r$-index \cite{GNP20} (the above-mentioned $\phi$-function, optimized with the \emph{move} technique of Nishimoto and Tabei \cite{Nishimoto21Move}), adding further $O(r)$ memory words of space. 

    \paragraph{General locating mechanism.} 

    We proceed in Section \ref{sec: general STPD index}
    by presenting a general technique to locate all pattern occurrences on \emph{any} order-preserving STPD. 
    Our technique is based on the observation that array $\LPF_\pi$ induces 
    an \emph{overlapping bidirectional parse} (i.e., a text factorization generalizing Lempel-Ziv '77 and allowing phrases to overlap) with $|\PDA_\pi|$ phrases. We prove that these factorizations have the following remarkable property: any given pattern $P$ that is a substring of $\mathcal T$ has exactly one text occurrence $\mathcal T[i,j]$ not entirely contained in a single phrase --- that is, a \emph{primary occurrence}. All the remaining occurrences of $P$ are entirely contained inside a phrase --- we call these \emph{secondary occurrences} --- and can be located from the primary occurrence by resorting to two-dimensional orthogonal point enclosure in $O(\log n)$ time each. To locate the primary occurrence, we present an algorithm working on any order-preserving STPD and returning the pattern occurrence $P[i,j]$ minimizing $\pi(i)$.

    \paragraph{Colexicographic rank.}

    In Section \ref{sec:stcolex} we then observe that on a particular order-preserving permutation $\pi$, the above general locating mechanism gets simplified: this happens when $\pi$ is the colexicographic rank of the text's prefixes (that is, $\pi = \IPA$). We show that the corresponding $\PDA$ array --- deemed $\stcolexm$ --- has size $|\stcolexm| \le \bar r$, where $\bar r$ is the number of runs in the Burrows-Wheeler transform of $\mathcal T$ reversed (note the symmetry with $|\stlexm| \leq r$). 
    This improves upon the known bound for Suffixient Arrays, whose size $\chi$ is only known to be at most $\chi \leq 2\bar r$ ~\cite{cenzato2025suffixientarraysnewefficient}.
    This STPD possesses a feature that makes it appealing for a practical implementation: the image  $\pi(\stcolexm)$ of $\stcolexm$ through $\pi$ is increasing. Ultimately, this implies that we do not need Range Minimum queries (constant-time in theory, but slow in practice) to identify samples minimizing $\pi$.
    This simplified version of the algorithm of Section \ref{sec: general STPD index} allows us to locate the colexicographically-smallest text prefix being suffixed by the query pattern $P$. 
    At this point, the locate mechanism of the $r$-index allows us to locate the remaining $occ$ occurrences. 
    Due to a smaller space usage with respect to the compressed suffix tree of Section \ref{sec:stlex} (the constant hidden in the $O(\bar r)$ space usage is close to $3$ in practice), our experimental results on \emph{locate} queries use an optimized implementation of this index.
    Experimentally, we show that this index solves the long-standing locality problem of compressed indexes supporting the powerful \emph{locate} query, being at the same time \emph{smaller} than the $r$-index and \emph{one to two orders of magnitude faster}. 

    \paragraph{Text position order.}
    Of interest is also the order-preserving identity permutation $\pi(i) = i$. The corresponding path decomposition array $\stposm$ allows computing the \emph{leftmost} pattern occurrence in the text. 
    The size $|\stposm|$ of this STPD array is the number of irreducible values in the Longest Previous Factor Array ($\LPF$), a new interesting repetitiveness measure that we analyze both in theory and experimentally.  
    We show that $p$ words of space are close to optimal in the worst case to compress strings of length $n$ with $|\stposm|=p$. 
    Together with the fact that $O(p)$ words of space are sufficient to compress the text (Theorem \ref{thm: compressing PDA}), this indicates that $|\stposm|$ is a strong compressibility measure, a fact that we confirm in section \ref{sec:experiments} showing that $|\stposm|$ is consistently smaller than $r$ in practice.
    We then briefly discuss applications of this STPD and observe that it is tightly connected with (1) the celebrated Ukkonen's suffix tree construction algorithm, and (2) the PPM$^*$ (prediction by partial matching) compression algorithm.

\section{Preliminaries}
Let $\Sigma = \{0,\dots, \sigma-1\}$ be a finite integer alphabet of size $\sigma$ endorsed with a total order $<$, which we call the \emph{alphabetic order}.
A \emph{string} of length $n$ over $\Sigma$ is a sequence $\mathcal S=\mathcal S[1]\mathcal S[2]\cdots\mathcal S[n]\in \Sigma^n$. In particular, note that $\sigma \le n$.
The \emph{reverse} $\mathcal S^{rev}$ of $\mathcal S$ is $\mathcal S^{rev} = \mathcal S[n]\mathcal S[n-1] \cdots \mathcal S[1]$.
With $\Sigma^*$ we denote the set of strings of arbitrary length, i.e., $\Sigma^*=\bigcup_{n\ge 0} \Sigma^n$. 
A \emph{text} is a string $\mathcal T\in \Sigma^n$ such that symbol $\mathcal T[n] = \$ \in \Sigma$ appears only in $\mathcal T[n]$ and is alphabetically-smaller than all other alphabet symbols, i.e., $\$ < c$ for all $c\in \Sigma \setminus \{\$\}$. 

With $[i, j]$ we denote the integer interval $\{i, i + 1, \ldots, j\}$ and with $[i]$ the interval $[1,i]$. For an arbitrary string $\mathcal S\in\Sigma^*$ and an interval $[i,j]$, we let $\mathcal S[i,j] = \mathcal S[i]\mathcal S[i + 1]\ldots \mathcal S[j]$. A string $\alpha$ is a \emph{substring} of a string $\mathcal S$, if there exists an interval $[i,j]$ such that $\alpha = \mathcal S[i, j]$. If $i=1$ ($j=n$), we call $\mathcal S[i, j]$ a \emph{prefix} (\emph{suffix}) of $\mathcal S$. The string $\mathcal S[i, j]$ is called a \emph{proper prefix} (\emph{proper suffix}) if it is a prefix (suffix) and in addition $\mathcal S[i, j] \neq \mathcal S$. To improve readability, we use symbols $\mathcal S$ and $\mathcal T$ to indicate the main string/text subject of our lemmas and theorems (usually, this is the string/text being indexed), and Greek letters $\alpha, \beta, \dots$ for their substrings.

We use the same notation that is used for strings for indexing arrays, i.e., if $A\in U^n$ is an array of $n$ elements over a universe $U$, then $A[i,j]= A[i]\ldots A[j]$. For a function $f$ with domain $U$, we use $f(A)$ to denote the array $f(A[1]) \ldots f(A[n])$.

\begin{definition}
    A substring $\alpha$ of $\mathcal S$ is said to be \emph{right-maximal} if (1)~it is a suffix of $\mathcal S$ or (2) there exist distinct $a,b\in \Sigma$ such that $\alpha a$ and $\alpha b$ are substrings of $\mathcal S$.
\end{definition}

We extend the alphabetic order $<$ of $\Sigma$ to the \emph{lexicographic} order $<_{\text{lex}}$ of $\Sigma^*$ as follows. For two strings (substrings) $\alpha$ and $\beta$, it holds that $\alpha<_{\text{lex}} \beta$ if $\alpha$ is a proper prefix of $\beta$, or if there exists $j$ such that $\alpha[i]=\beta[i]$ for all $i\in [j-1]$ and $\alpha[j]<\beta[j]$.
The \emph{colexicographic order} $<_{\text{colex}}$ is defined symmetrically: for two strings $\alpha$ and $\beta$, it holds that $\alpha<_{\text{colex}} \beta$ if $\alpha$ is a proper suffix of $\beta$, or if there exists $j$ such that $\alpha[|\alpha| - i + 1]=\beta[|\beta| - i + 1]$ for all $i\in [j-1]$ and $\alpha[|\alpha| - j + 1]<\beta[|\beta| - j + 1]$.

We proceed with the definition of the longest common prefix (suffix) functions.

\begin{definition}[$\lcp$, $\lcs$]
    The \emph{longest common prefix} function $\lcp$ (the \emph{longest common suffix} function $\lcs$) is defined as the function that, for two strings $\alpha\in \Sigma^n$ and $\beta\in\Sigma^m$, returns the maximum integer $k = \lcp(\alpha, \beta)$ ($k = \lcs(\alpha, \beta)$) such that $\alpha[1,k]=\beta[1,k]$ ($\alpha[n - k + 1, n]=\beta[m - k + 1, m]$).
\end{definition}
Using these functions we now define longest common extension oracles.
\begin{definition}[$\rlce$, $\llce$, and $\lce$]
    Let $\mathcal S$ be a string of length $n$.
    The \emph{right (left) longest common extension} function $\mathcal S.\rlce$ ($\mathcal S.\llce$) is the function that, for two distinct integers $i,j\in[n]$, returns $\mathcal S.\rlce(i, j) = \lcp(\mathcal S[i, n], \mathcal S[j, n])$ ($\mathcal S.\llce(i, j) = \lcs(\mathcal S[1, i], \mathcal S[1, j])$). 
    A \emph{longest common extension ($\lce$)} oracle is an oracle that supports both $\mathcal S.\rlce$ and $\mathcal S.\llce$ queries for the string $\mathcal S$. 
\end{definition}
If the string $\mathcal S$ is clear from the context, we simply write $\rlce$ and $\llce$ instead of $\mathcal S.\rlce$ and $\mathcal S.\llce$.

We review further basic concepts that we consider the expert reader to already be familiar with in Appendix~\ref{app:basic concepts}. This includes the Suffix Array, suffix tree and Burrows-Wheeler transform.

\input{order-preserving-stpds}

\section{Preliminary experimental results}\label{sec:experiments}

We show some preliminary experimental results on repetitive genomic collections. In Section \ref{sec: PDA in practice} we show that the number of irreducible $\LCP$ values is, in practice, much smaller than $r$. 
Section \ref{sec: experiment ST vs rindex} is dedicated to showing that (as expected from Remark \ref{remark:ST cache}) the suffix tree has excellent locate performance compared to the $r$-index.
In subsection \ref{sec: experiments index} we show that, in addition to being very small, our index is also orders of magnitude faster than the $r$-index (especially as the pattern's length $m$ increases).

\subsection{$|\PDA_\pi|$ in practice}\label{sec: PDA in practice}
We begin with a few numbers (Table \ref{tab:measures}) showing how $|\PDA_\pi|$ compares in practice with $n$ (text length), $r$ (number of runs in the Burrows-Wheeler transform of the text), $\bar r$ (number of runs in the Burrows-Wheeler transform of the reversed text), and $\chi$ (size of the smallest suffixient set) for $\pi = \ISA$ ($|\stlexm|$), $\pi = \IPA$ ($|\stcolexm|$), and $\pi = id$ ($|\stposm|$). For this experiment, we used the corpus \url{https://pizzachili.dcc.uchile.cl/repcorpus.html} of repetitive DNA collections  (this repository became a standard in the field of compressed data structures). While only being representative of a small set of data, we find it interesting to observe that on these datasets the number of irreducible $\LCP$ values ($|\stlexm|$ and $|\stcolexm|$) is consistently smaller than both the number of $\BWT$ runs ($r$ and $\bar r$) and the Suffixient Array samples ($\chi$). From a practical perspective, this means that it is preferable to use data structures whose space is proportional to $|\stcolexm|$ rather than $r$ or $\chi$. 
The same holds true for the number of irreducible $\LPF$ values ($|\stposm|$).

\begin{table}[]
    \centering
    \begin{tabular}{|l|r|r|r|}
    \hline
    & influenza & cere & escherichia\\
    \hline
    $n$ & 154808556 & 461286645 & 112689516 \\
    $|\stlexm|$ & 1815861 & 7455556 & 9817075 \\
    $|\stcolexm|$ & 1805730 & 7455980 & 9825972 \\
    $|\stposm|$ & 1928408 & 8954851 & 11677573 \\
    $\chi$ & 2225543 & 9921698 & 13119749 \\
    $r$ & 3022822 & 11574641 & 15044487 \\
    $\bar r$ & 3018825 & 11575583 & 15045278 \\
    \hline
    \end{tabular}
    \vspace{5pt}
    \caption{Experimental evaluation of Path Decomposition Array sizes on repetitive collections from \url{http://pizzachili.dcc.uchile.cl/repcorpus.html}.}
    \label{tab:measures}
\end{table}

\subsection{I/O complexity of classic versus compressed indexes}\label{sec: experiment ST vs rindex}

We compared custom implementations of the suffix tree and Suffix Array against the original $r$-index's implementation (\url{https://github.com/nicolaprezza/r-index}) on a small dataset in order to confirm the caching effects discussed in Section \ref{sec:intro}.
All the experiments in this and in the following section have been run on an Intel(R) Xeon(R) W-2245 CPU @ 3.90GHz workstation with 8 cores and 128 gigabytes of RAM
running Ubuntu 18.04 LTS 64-bit. 

Due to the large space usage of the suffix tree, in this experiment we tested a text of limited length $n = 10^8$ formed by variants of the SarsCoV2 genome downloaded from the Covid-19 data portal (\url{www.covid19dataportal.org/}). The goal of the experiment was to verify experimentally the size and I/O complexity of the suffix tree, Suffix Array, and $r$-index:

\begin{itemize}
    \item Suffix tree (Remark \ref{remark:ST cache}): $O(n)$ words of space, $O(d + m/B)$ I/O complexity for finding one occurrence, $O(occ)$ I/O  complexity to find the remaining $occ-1$ occurrences.
    \item Suffix array: $n$ words of space on top of the text, $O((1+m/B)\log n)$ I/O complexity for finding one occurrence, $O(occ/B)$ I/O  complexity to find the remaining $occ-1$ occurrences.
    \item $r$-index: $O(r)$ words of space, $\Omega(m)$ I/O complexity for finding one occurrence, $\Omega(occ)$ I/O  complexity to find the remaining $occ-1$ occurrences.
\end{itemize}

The experiment consisted in running locate queries on $10^5$ patterns of length $m=100$, all extracted from random text positions (in particular, $occ\ge 1$ for every pattern). 
The average number of occurrences per pattern was 2896 (the dataset is very repetitive).
For each pattern, we measured the running time of locating the first occurrence and that of locating all occurrences. 

To no surprise, due to the highly repetitive nature of the dataset, the $ r$-index's size was of just $0.9$ MiB, orders of magnitude smaller than the Suffix Array (0.3 GiB) and of the suffix tree (9.9 GiB). As expected, however, the poor cache locality of pattern matching queries on the $r$-index makes it orders of magnitude slower than those two classic data structures. 
As far as locating the primary occurrence (finding the pattern's locus in the suffix tree) was concerned, the $r$-index took 572 $ns$ per pattern character on average. This is not far from the 200 $ns$ per character of the Suffix Array (where each binary search step causes a cache miss --- an effect that gets diluted as $m$ increases), but one order of magnitude larger than the 60 $ns$ per character of the suffix tree. This also indicates that the term $d$ (node depth of the pattern's locus in the suffix tree) has a negligible effect in practice. 

As far as locating the remaining $occ-1$ occurrences was concerned, the $r$-index took 105 $ns$ per occurrence on average to perform this task. This is five times larger than the 21 $ns$ per character of the suffix tree; while both trigger at least one cache miss for each located occurrence, the $r$-index executes a predecessor query per occurrence, which in practice translates to several cache misses\footnote{This gap was recently closed by the move structure of Nishimoto and Tabei \cite{Nishimoto21Move}. On this dataset, the move-$r$ data structure \cite{BertramMoveR24} was able to locate each occurrence in 14 $ns$ on average.}. As expected, the Suffix Array solved this task in 3.4 $ns$ per occurrence. This confirms its $O(occ/B)$ I/O complexity for outputting the occurrences after locating the Suffix Array range.

\subsection{Comparing $\stcolexm$ with the state-of-the-art}\label{sec: experiments index}

Figure \ref{fig:results3} shows some preliminary results comparing our index based on $\stcolexm$ (Section \ref{sec:stcolex}) with the \texttt{$r$-index} \cite{GNP20}, \texttt{move-$r$} \cite{BertramMoveR24}, the \texttt{suffixient array} \cite{cenzato2025suffixientarraysnewefficient} and the \texttt{suffix array}.

\begin{figure}[pt]
\centering
    \begin{minipage}{0.45\linewidth}
        \includegraphics[width=\linewidth]{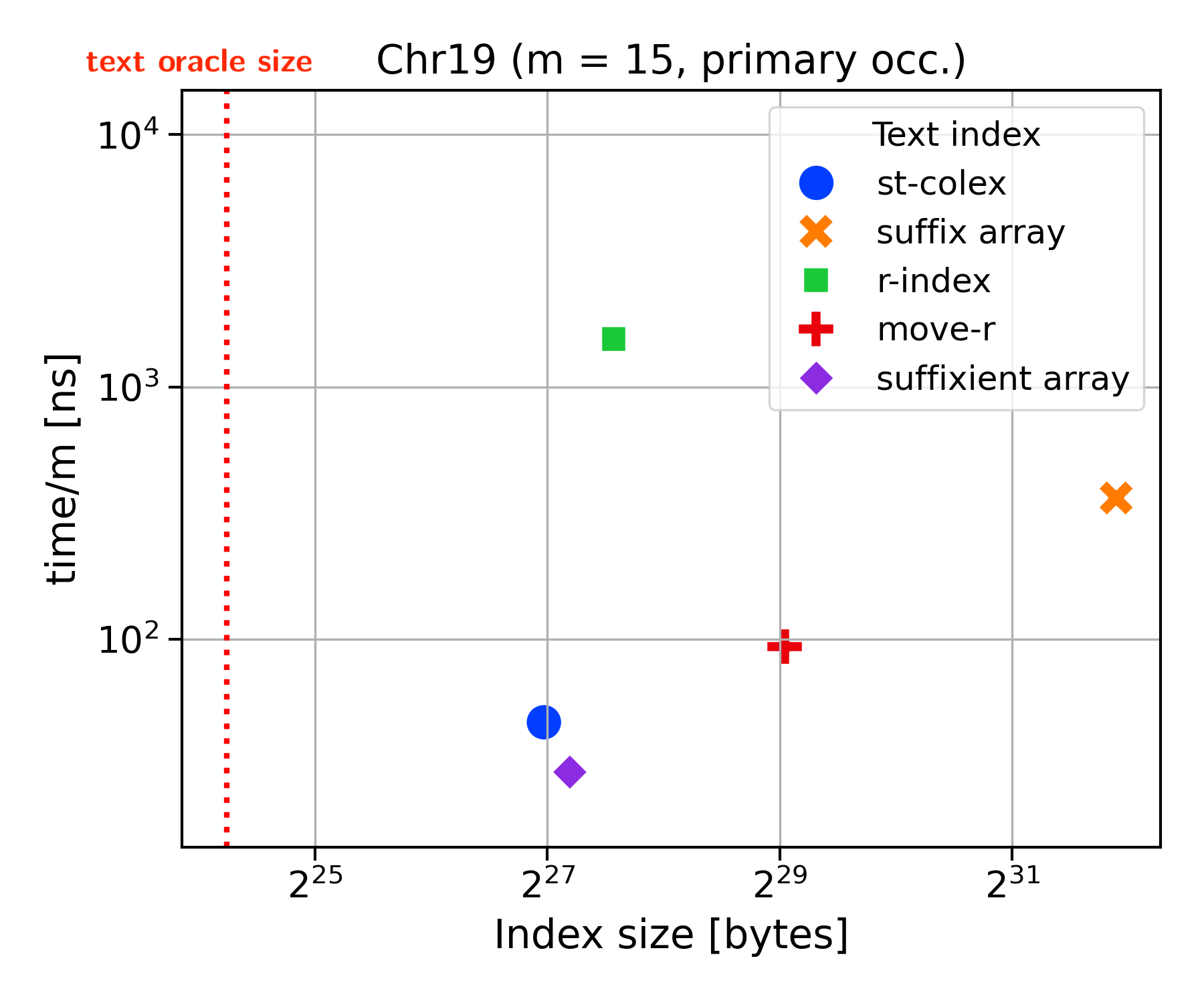}
    \end{minipage}
    \begin{minipage}{0.45\linewidth}
        \includegraphics[width=\linewidth]{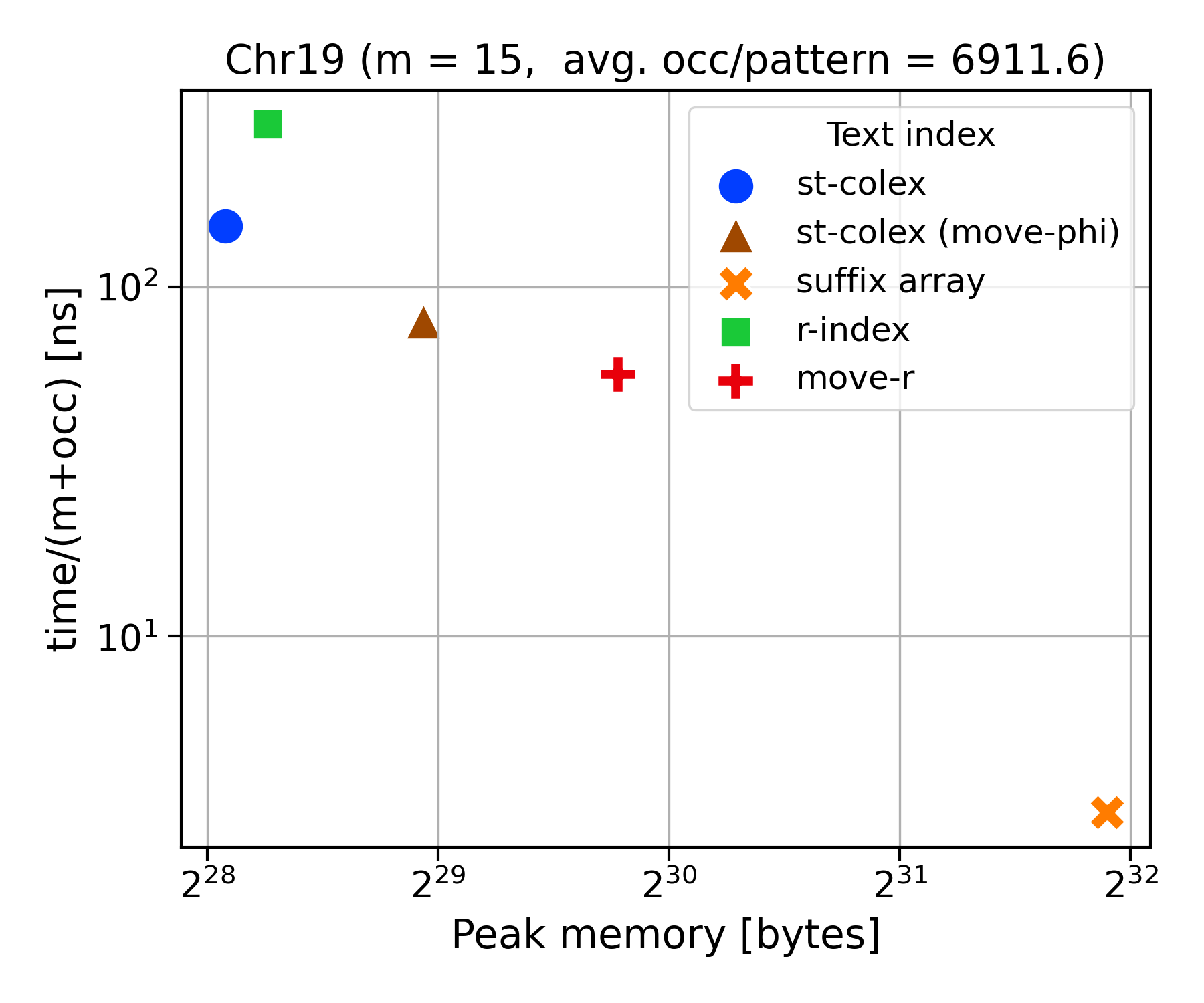}
    \end{minipage}
    \\
    \begin{minipage}{0.45\linewidth}
        \includegraphics[width=\linewidth]{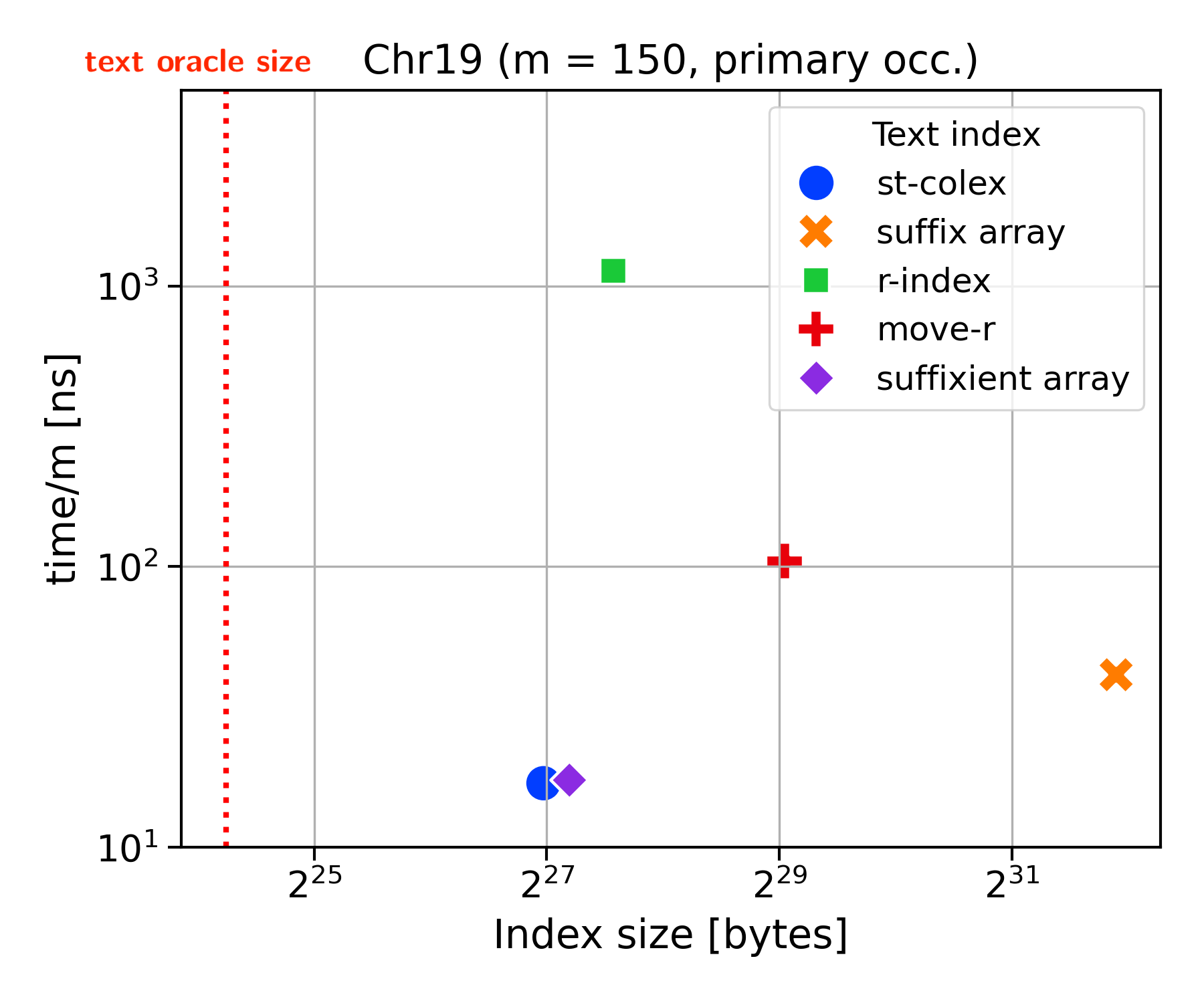}
    \end{minipage}
    \begin{minipage}{0.45\linewidth}
        \includegraphics[width=\linewidth]{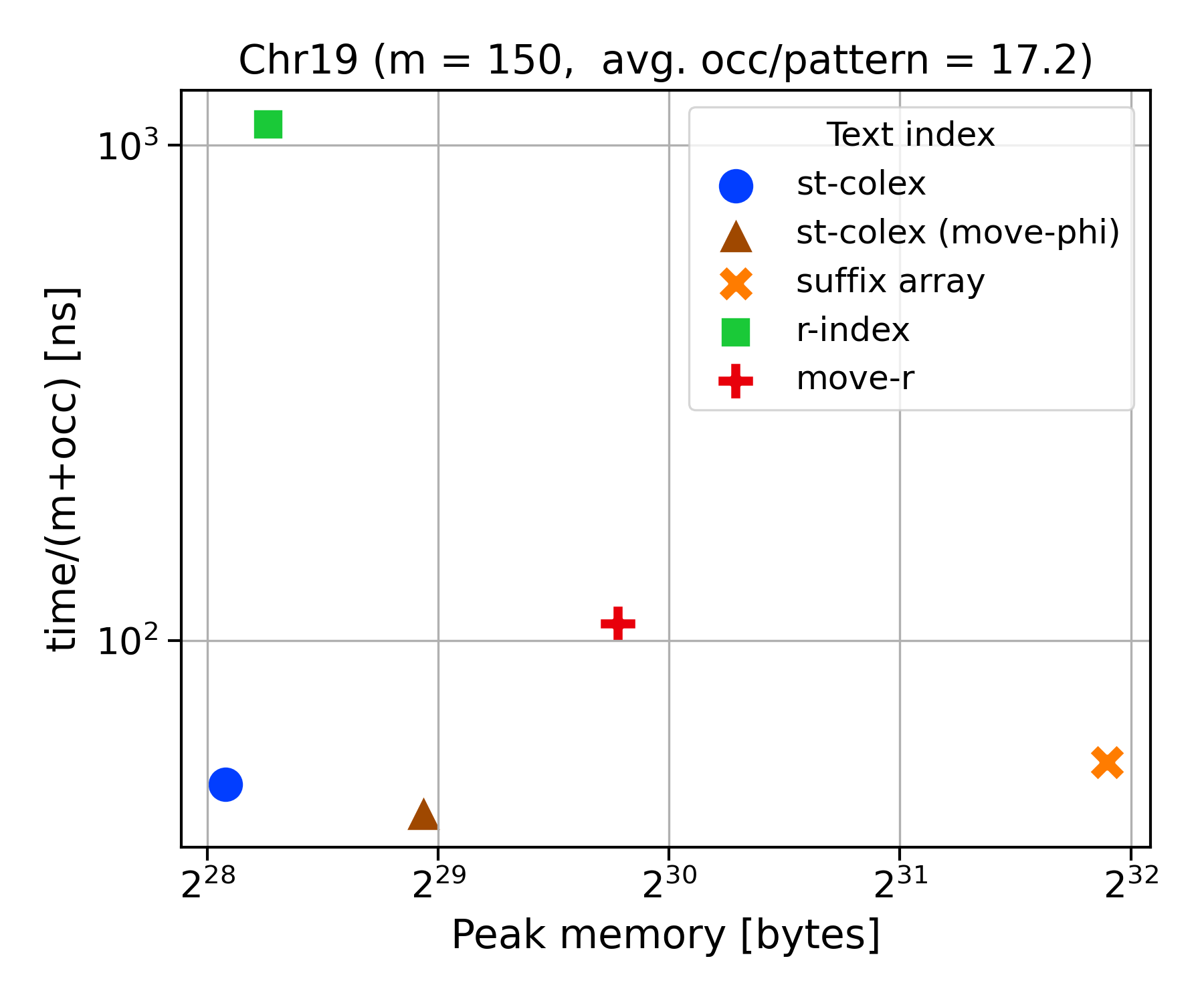}
    \end{minipage}    
    \\
    \begin{minipage}{0.45\linewidth}
        \includegraphics[width=\linewidth]{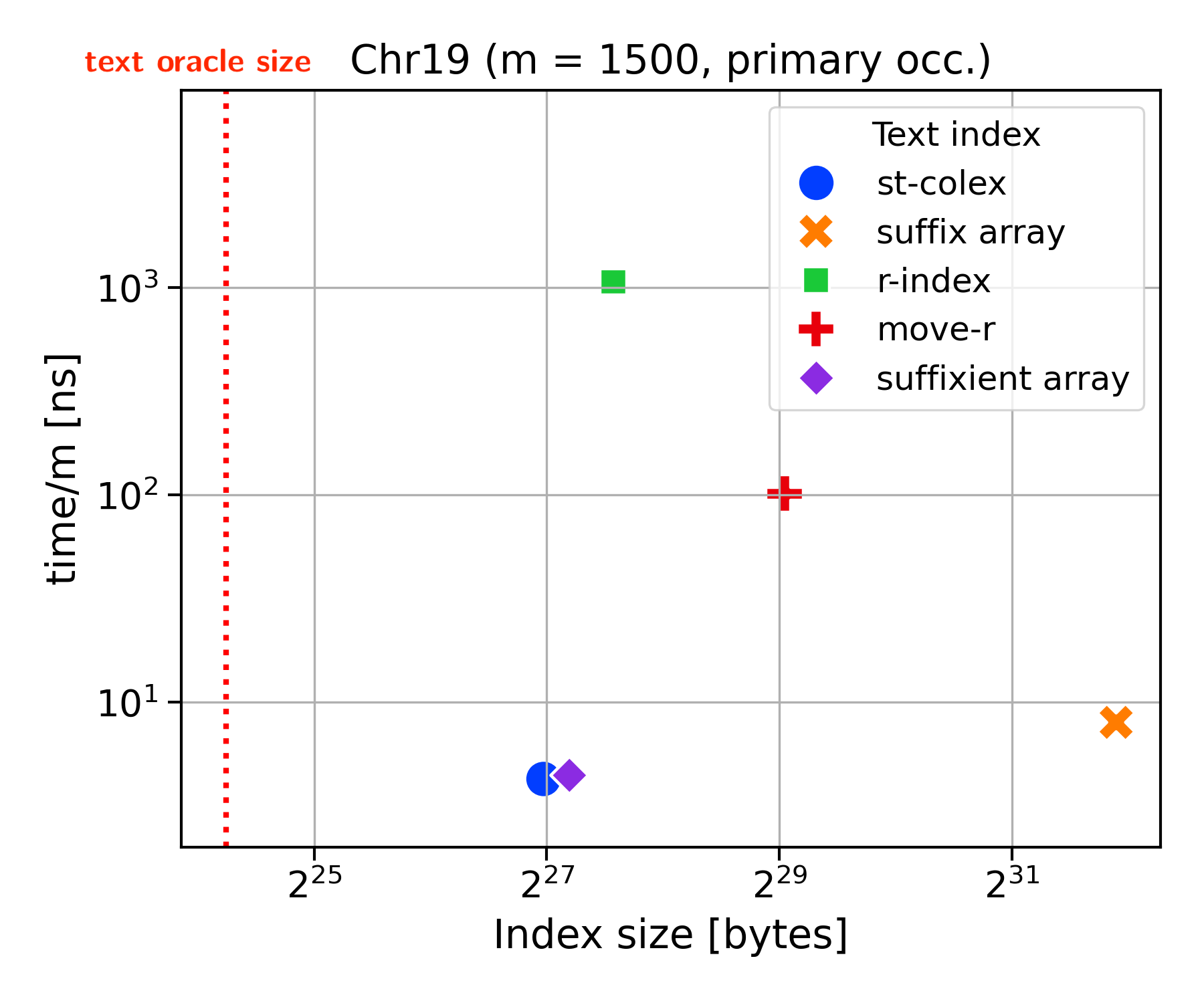}
    \end{minipage}
    \begin{minipage}{0.45\linewidth}
        \includegraphics[width=\linewidth]{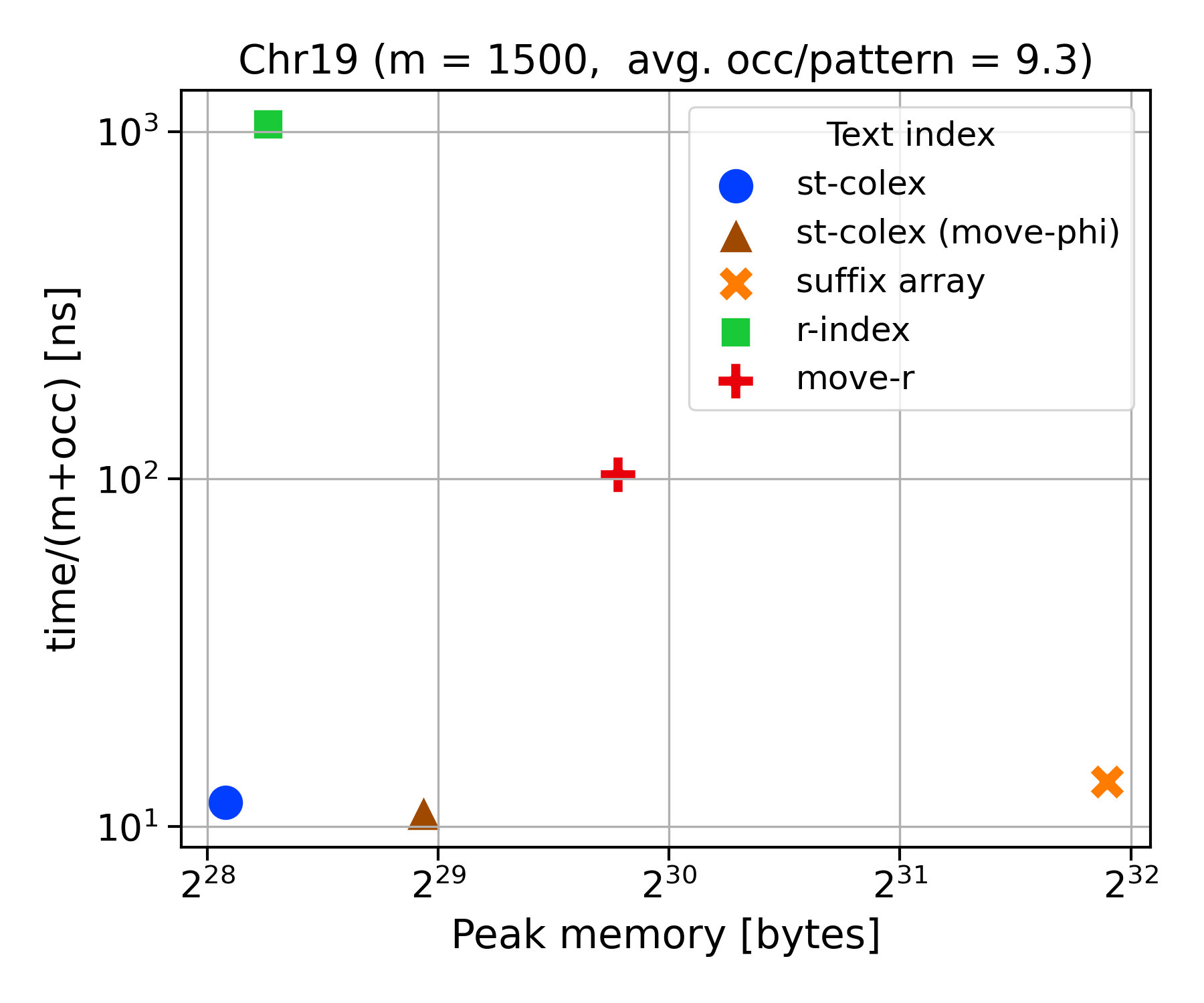}
    \end{minipage}    
    \caption{
    \small
    Time and space used by our index $\stcolex$ (on the right: using both the $r$-index and the move-$r$ implementation of the $\phi$-function to locate all occurrences; see description below for more details), \texttt{$r$-index}, \texttt{move-$r$}, \texttt{suffixient array}, and \texttt{suffix array} to find primary occurrences (left column) and to locate all occurrences (right column) for a set of $10^5$ patterns (lengths $m=$ 15, 150, 1500). 
    For primary occurrences, our index deviates from the usual Pareto curve, dominating all solutions (except the \texttt{suffixient array} for $m=15$) in both dimensions. 
    Conversely, being based on the $\phi$ function---incurring $O(1)$ I/O operations per reported occurrence---, our locate mechanism follows the Pareto curve when $occ$ is large (top-right; $\approx 6911.6$ occurrences/pattern). Escaping this curve is beyond the scope of this article and left to future research.}
    \label{fig:results3}
\end{figure}

We compare these indexes on the task of locating just one (here called ``primary'') occurrence and all occurrences of a set of patterns in a text of length $n \approx 10^9$. The suffix tree was excluded from this experiment due to its size (it did not fit in the RAM of our workstation).

\paragraph{Implementation details.} 

Our index follows the design of Section \ref{sec:stcolex}, except for a few details and optimizations that we describe below. 


First, when locating all occurrences (right plots in Figure \ref{fig:results3}), we integrated both the original $r$-index $\bar{\phi}$-function implementation, which locates each occurrence in $O(\log(n/r))$ time---in the plots, this version of our index is indicated with \texttt{st-colex}---, and the move structure by Nishimoto and Tabei \cite{Nishimoto21Move} (whose optimized version is implemented in move-$r$ \cite{BertramMoveR24})---in the plots, \texttt{st-colex (move-phi)}. This choice is motivated by the fact that, for large $occ$ (Figure~\ref{fig:results3}, pattern length 15), the $r$-index $\bar{\phi}$-function is slower than the move-$\bar{\phi}$ alternative, which, however, requires significantly more space. This results in a space-time trade-off that we discuss below.

As random access oracle, we used an ad-hoc implementation of Relative Lempel Ziv \cite{kuruppu2010relative}, yielding excellent compression and random access cache locality on repetitive collections of genomes. While relative Lempel-Ziv cannot formally guarantee $O(\ell/B)$ I/O operations for extracting a contiguous subsequence of $\ell$ characters, in practice it gets close to this performance on repetitive collections (as these experiments show).

Finally, we implemented two heuristics speeding up function $\sufsearch$ (see Algorithm~\ref{alg: locate primary}): (i) we tabulate the results of short suffixes of length at most $k_{max}$ (where the maximum $k_{max}$ is chosen so that the table uses at most $50\%$ of the space of the remaining part of the index), and (ii) we start Algorithm~\ref{alg: locate primary} from $j=k_{max}+1$, choosing the value of $i$ returned by the table on $P[1,k_{max}]$. 
Heuristics (i) and (ii) allow avoiding binary search in most cases and dramatically speed up Algorithm~\ref{alg: locate primary} in practice. 

\paragraph{Datasets and experiments.}

The input text consisted of 19 variants of Human chromosome 19 (total length $n\approx 10^9$) downloaded from \url{https://github.com/koeppl/phoni}.
We extracted $10^5$ patterns of variable length from the text. For each pattern, we measured the resources (peak memory and running time normalized by dividing by $m$ --- left --- and $m+occ$ --- right) used by all the indexes while finding one pattern occurrence and locating all pattern occurrences (Figure \ref{fig:results3}).

\paragraph{Comments.}

On this dataset, our index \texttt{st-colex} is smaller than \texttt{$r$-index}, about four times smaller than \texttt{move-$r$}, and orders of magnitude smaller than \texttt{suffix array}. 

\textbf{Locate one occurrence.} As expected, the left column of Figure \ref{fig:results3} shows that our index is always orders of magnitude faster than \texttt{$r$-index} on the task of locating one occurrence, and even faster than \texttt{suffix array}. 
As expected, I/O-inefficient indexes incurring $O(m)$ I/O operations (\texttt{$r$-index} and \texttt{move-$r$}) exhibit a constant time-per-character running time as a function of $m$ (note that times are normalized by dividing by $m$). The other indexes, incurring a number of I/O operations close to the optimal $m/B$, on the other hand, exhibit a time-per-character search time that gets smaller as the pattern length increases. Our index and \texttt{suffixient array} always dominate all other indexes by a wide margin in both coordinates. Our index is smaller than \texttt{suffixient array}: this reflects the fact that, as we proved in Lemma \ref{lem: stlex <= r}, we sample at most $r$ prefix array entries (in practice much less than that; see also Table \ref{tab:measures}), while \texttt{suffixient array} can sample up to $2r$ entries~\cite{navarro2025smallestsuffixientsetsrepetitiveness, date2025neartightnesschileq2r}.


\textbf{Locate all occurrences.} Figure \ref{fig:results3} right column compares all indexes on the task of locating all occurrences. \texttt{suffixient array} was excluded from this experiment since its available implementation does not support this query; in any case, as observed previously, note that locating all occurrences starting from the arbitrary occurrence returned by \texttt{suffixient array} (whose running time is shown in the left plots), requires storing \emph{twice} the phi function mechanism (since both backward and forward navigation of the suffix array are required). Hence, even if an implementation were available, this index would exhibit essentially the same running time as our index, but a much higher space usage. 

Note that running times are normalized by dividing by $m+occ$, that is, the size of the input plus the output. As expected, when $occ \gg m$ the time performance of our index gets close to that of \texttt{$r$-index} and worse than that of \texttt{move-$r$}. This is the case for $m=15$, where each pattern occurred on average 6911.6 times in the (repetitive) text. Longer patterns occurred less frequently and, again, in those regimes our index outperforms all the others in both query time and space usage. 

As proved in these experiments (\texttt{st-colex (move-phi)}), the locating mechanism of \texttt{move-$r$} can be seamlessly integrated in our index (since it is just a way to speed up computing the $\bar\phi$ function of Definition \ref{def:phi}, which stands at the core of our locating mechanism). This yields (at the same time) about the same \emph{locate} time per occurrence of \texttt{move-$r$} data structure, but a much better space usage (almost twice as small). 
Reducing the space of the move data structure of Nishimoto and Tabei \cite{Nishimoto21Move} without sacrificing query times goes out of the scope of this article and will be covered in a journal extension. In any case, this approach cannot break the $\Omega(occ)$ cache misses for locating the pattern's occurrences. Obtaining $O(occ/B)$ cache misses in a space comparable to that of our data structure is the ultimate (challenging) goal of this line of research.

\newpage
\bibliography{stpd}
\newpage

\begin{appendix}

\section{Basic Concepts}
\label{app:basic concepts}

In this section, we review some basic concepts to help readers not being familiar with the field.

\subsection{Suffixient Arrays}\label{app:suffixient arrays}

Like Suffixient Arrays \cite{cenzato2025suffixientarraysnewefficient}, ours is a technique for sampling the Prefix Array while still maintaining search functionalities. 
Suffixient Arrays are based on the concept of \emph{suffixient set}. See Figure \ref{fig:sA example} for a running example.
A suffixient set is a set $S\subseteq [n]$ of positions on the text $\mathcal T$ with the property that, for every string $\alpha$ labeling the path starting from the suffix tree root to the first character of every suffix tree edge, there exists a position $i\in S$ such that $\alpha$ is a suffix of $\mathcal T[1,i]$. In other words, for every one-character right-extension $\alpha$ of every right-maximal string $\alpha[1,|\alpha|-1]$, there exists $i\in S$ such that $\alpha$ is a suffix of $\mathcal T[1,i]$. The Suffixient Array $\sA$ is a (not necessarily unique) suffixient set $S$ of smallest cardinality, sorted according to the colexicographic order of the corresponding text prefixes  $\{\mathcal T[1,i]\ :\ i\in S\}$.
As discussed in the caption of Figure \ref{fig:sA example} with an example, binary search on $\sA$ and random access on $\mathcal T$ suffice to locate one pattern occurrence.

\begin{figure}
    \centering
    \includegraphics[width=0.6\linewidth,page=3]{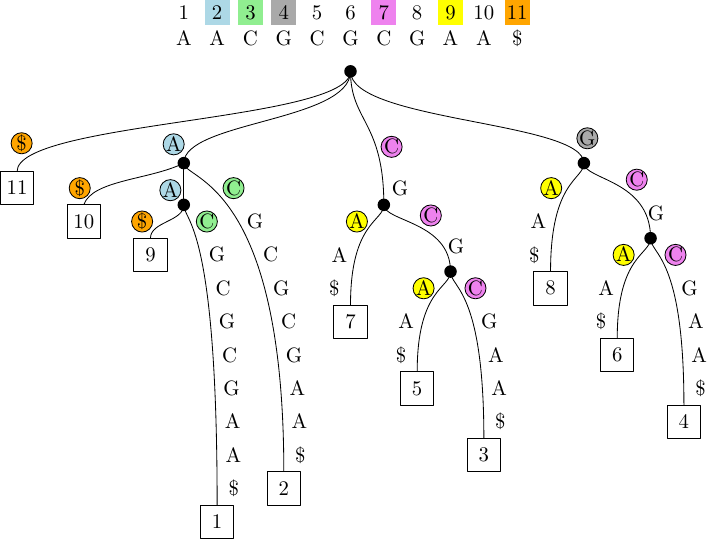}
    
    \caption{
    Visualization of a smallest suffixient set for string AACGCGCGAA\$.  
        The corresponding Suffixient Array is $\sA = [11,2,9,3,7,4]$ (assuming alphabet order $\$ < A < C < G$).
        We show how pattern matching works with an example.
        Imagine the task of matching pattern $P = CGCGA$ on $\mathcal T$, and assume that $P$ does occur in $\mathcal T$. Since the alphabet has cardinality at least 2, then the empty string $\epsilon$ is \emph{right-maximal}, hence the pattern prefix $C$ is a one-character extension of a right-maximal string. 
        By binary search on $\sA$ and random access on $\mathcal T$, we find all $i\in \sA$ such that $C$ suffixes $\mathcal T[1,i]$: in this case,  $\mathcal T[1,3]$ and $\mathcal T[1,7]$. Choose arbitrarily such a prefix, for instance $\mathcal T[1,3]$ (this arbitrary choice does not affect correctness of the procedure). From this point, continue matching (by random access on $\mathcal T$) the remaining suffix $P[2,5] = \underline{GCG}A$ with $\mathcal T$'s suffix following the match: $\mathcal T[4,7] = \underline{GCG}CGAA\$$. As highlighted, three characters ($GCG$) match. At this point, we know that both $\mathcal T[3,7] = CGCGC$ and $P[1,5] = CGCGA$ occur in the text. But then, this means again that $P[1,4] = CGCG$ is right-maximal, hence binary-searching $\sA$ with string $CGCGA$ will yield at least one prefix being suffixed by $CGCGA$: in this case, $\mathcal T[1,9]$. Since we reached the end of $P$, we found an occurrence of $P$, that is: $P = \mathcal T[9-m+1,9] = \mathcal T[5,9]$.
    }
    \label{fig:sA example}
\end{figure}

While Suffixient Arrays sample every edge of the suffix tree (see Figure \ref{fig:sA example}), our technique samples just a subset of the edges (that is, the first edge in every STPD path, see Figure \ref{fig:STLEX}). As we show in this paper, this has several benefits: it leads to a smaller sampling size (consistently smaller in practice), it allows us to locate the occurrence optimizing a user-defined function $\pi$, and ultimately it allows us to simulate suffix tree operations and locate all pattern occurrences in small space.

\subsection{Suffix sorting and basic succinct data structures}

\begin{definition}[Suffix Array (SA) etc.~\cite{manber1993suffix}]
    Let $\mathcal S$ be a string of length $n$.
    \begin{itemize}
        \item The \emph{Suffix Array} $\SA$ of $\mathcal S$ is the permutation of $[n]$ such that $\mathcal S[\SA[i], n] <_{\text{lex}} \mathcal S[\SA[j], n]$ holds for any $i,j\in [n]$ with $i<j$. 
        \item The \emph{Prefix Array} $\PA$ of $\mathcal S$ is the permutation of $[n]$ such that $\mathcal S[1,\PA[i]] <_{\text{colex}} \mathcal S[1,\PA[j]]$ holds for any $i,j\in [n]$ with $i<j$.
        \item The \emph{Inverse Suffix Array} $\ISA$ of $\mathcal S$ is the permutation of $[n]$ such that $\ISA[i] = j$ if and only if $\SA[j] = i$.
        \item The \emph{Inverse Prefix Array} $\IPA$ of $\mathcal S$ is the permutation of $[n]$ such that $\IPA[i] = j$ if and only if $\PA[j] = i$.
    \end{itemize}
\end{definition}
We proceed with the definition of the suffix tree. In what follows, we fix a text $\mathcal T$ of length $n$.
\begin{definition}[Suffix trie and Suffix tree (ST)~\cite{weiner1973linear}]
    The \emph{suffix tree} of $\mathcal T$ is an edge-labeled rooted tree with $n$ leaves numbered from $1$ to $n$ such that (i) each edge is labeled with a non-empty substring of $\mathcal T$, (ii) each internal node has at least two outgoing edges, (iii) the labels of outgoing edges from the same node start with different characters, and (iv) the string obtained by concatenating the edge labels on the path from the root to the leaf node numbered $\SA[i]$ is $\mathcal T[\SA[i],n]$ where $\SA$ is the Suffix Array of $\mathcal T$.
\end{definition}
The \emph{suffix trie} of $\mathcal T$ is the edge-labeled tree obtained from the suffix tree by replacing every edge labeled with a string $\alpha$ with a path of $|\alpha|$ edges labeled with $\alpha[1], \ldots, \alpha[|\alpha|]$. Hence the suffix trie is edge-labeled with characters, while the suffix tree is edge-labeled with strings. A node of degree $2$ in the suffix trie that is not the root is an \emph{implicit node}, all other nodes are \emph{explicit nodes}. Note that explicit nodes are the nodes that are both in the suffix tree and the suffix trie, while a node that is introduced by the procedure of replacing edges with paths above is an implicit node. For an internal node $u$ in the suffix tree, we denote with $\out(u)\subseteq \Sigma$ the set of first characters of strings labeling outgoing edges of $u$.

\begin{definition}[Path Label, String Depth, Locus]
    (i)~For a node $u$ in the suffix tree/trie of $\mathcal T$, we call $\alpha(u)$ the unique string obtained by concatenating the edge labels on the path from the root of the suffix tree to $u$ the \emph{path label} of $u$. (ii)~We call $\sd(u) = |\alpha(u)|$ the \emph{string depth} of $u$. (iii)~For a right-maximal substring $\alpha$ of $\mathcal T$, the \emph{locus} of $\alpha$, denoted by $\locus(\alpha)$, is the unique suffix tree node $u$ for which $\alpha(u) = \alpha$.
\end{definition}
Observe that, for any $i\neq j$, the longest common extension $\mathcal T[i,i + \rlce(i,j) - 1]$ at $i,j$ is a right maximal substring whose locus is at string depth $\rlce(i,j)$ in the suffix tree.  


        
        
        
        
        
        

We proceed with the definitions of the longest common Prefix Array, the permuted longest common Prefix Array, and the longest previous factor array. Note that $\LCP$ is the longest common Prefix Array for a text $\mathcal T$, while $\lcp$ was the function that for two strings returns the length of their longest common prefix.
\begin{definition}[LCP, PLCP, and LPF]
    Let $\mathcal S$ be a string of length $n$.
    \begin{itemize}
        \item The \emph{longest common prefix} array $\LCP$ of $\mathcal S$ is the length-$(n-1)$ integer array such that $\LCP[i]:=\mathcal S.\rlce(\SA[i], \SA[i-1])$, for all $i\in [2,n]$, where $\SA$ is the Suffix Array of $\mathcal S$.
        \item The \emph{permuted longest common prefix} array $\PLCP$ of $\mathcal S$ is the length-$(n-1)$ integer array such that $\PLCP[i]:=\LCP[\ISA[i]]$, for all $i\in [n - 1]$, where $\ISA$ is the Inverse Suffix Array of $\mathcal S$.
        \item An \emph{irreducible $\PLCP$ value} is a value $\PLCP[i]$ such that $i =1$ or $\PLCP[i]\neq \PLCP[i-1] - 1$.
        \item The \emph{longest previous factor} array $\LPF$ of $\mathcal S$ is the length-$n$ integer array such that $\LPF[i]:=\max\{\mathcal S.\rlce(i,j): j<i\}$ for all $i\in [n]$.
    \end{itemize}
\end{definition}
    
We now define the Burrows-Wheeler transform. In order to do so we need to define the rotations of a string. For a string $\mathcal S$ of length $n$ and an integer $i\in [n]$, the $i$'th rotation of $\mathcal S$ is the string $\mathcal S^{\leftarrow i}:=\mathcal S[i + 1] \ldots \mathcal S[n]\mathcal S[1] \ldots \mathcal S[i]$. 
    
\begin{definition}[Burrows-Wheeler Transform~\cite{BW94}]
    \label{def: BWT}
    \begin{itemize}
        \item The \emph{Burrows-Wheeler transform} (\emph{co-Burrows-Wheeler transform}) of a string $\mathcal S$, denoted by $\BWT(\mathcal S)$ ($\coBWT(\mathcal S)$), is the permutation of the characters of $\mathcal S$ that is obtained by lexicographically sorting (colexicographically sorting) all the rotations of $\mathcal S$ and concatenating the last (first) character of each string in this sorted list.
        \item We define $r$ ($\bar r$) as the number of equal-letter runs in $\BWT(\mathcal S)$ ($\coBWT(\mathcal S)$), i.e., the number of maximal substrings of $\BWT(\mathcal S)$ ($\coBWT(\mathcal S)$) containing a single character. 
    \end{itemize}
    We will omit $\mathcal S$ when clear from the context.
\end{definition}

\begin{remark}\label{remark: properties}
    The following properties concerning the $\SA$, $\ISA$, $\BWT$, and $\coBWT$ of a string $\mathcal S$ are immediate from their definition. For $i,j\in[n]$, 
    \begin{enumerate}
        \item $\BWT(\mathcal S)[i] = \mathcal S[\SA[i]-1]$, where $\mathcal{S}[0]:=\mathcal{S}[n]$.
        \item \label{remark: properties: ISA order} If $\ISA[i]<\ISA[j]$ and $\mathcal{S}[i-1]=\mathcal{S}[j-1]$, then $\ISA[i-1]<\ISA[j-1]$, where $\ISA[0]:=\ISA[n]$.
        \item \label{remark: properties: BWT ISA} $\BWT[\ISA[i]]=\mathcal{S}[i-1]$.
        \item \label{remark: properties: BWT rotation} $\BWT(\mathcal{S})=\BWT(\mathcal S^{\leftarrow i})$.
        \item \label{remark: properties: coBWT} $\coBWT[i]=\mathcal{S}[\PA[i]+1]$ where $\mathcal{S}[n+1]:=\mathcal{S}[1]$.
    \end{enumerate}
\end{remark}

We proceed with the definition of Range Minimum and Maximum queries. 
\begin{definition}
    Given a list $L$ of $n$ integers, the \emph{Range Minimum (Maximum) query} on $L$ with arguments $\ell, r\in [n]$, returns the index $\argmin_{k \in [\ell, r]} L[k]$ ($\argmax_{k \in [\ell, r]} L[k] $) of the minimum (maximum) element in $L[\ell,r]$.   
\end{definition}    
     
There exists a data structure that uses $2n+o(n)$ bits and supports Range Minimum (Maximum) queries in $O(1)$ time~\cite{FischerHeun11}.
		
\end{appendix}

\end{document}

%% file: order-preserving-stpds.tex
\section{Order-preserving STPDs}\label{sec:order preserving STPD}

 As introduced above, the starting point of our technique is the choice of a priority function (permutation) $\pi : [n] \rightarrow [n]$ that we use to generalize suffix sorting. 
    In this paper we focus on permutations possessing the following property (but the technique can be made to work on any permutation: we will treat the general case in an extension of this article). 

	\begin{definition}[Order-preserving permutation]\label{def:order-preserving pi}
	Let $\mathcal S \in\Sigma^n$ be a string and 	
        $\pi : [n] \rightarrow [n]$ be a permutation. The permutation $\pi$ is said to be \emph{order-preserving for $\mathcal S$} if and only if  ${\pi(i) < \pi(j)} \wedge \mathcal S[i,i+1] = \mathcal S[j,j+1]$ implies $\pi(i+1) < \pi(j+1)$ for all $i,j \in [n-1]$.
	\end{definition}

    The above property is sufficient and necessary to guarantee the following desirable universal minimization property: if $\mathcal S[i,j] = \mathcal S[i',j']$ are two pattern occurrences, then $\pi(i+k) < \pi(i'+k)$ either holds for all $k\in [0,j-i]$ or for none. 
    It is a simple exercise to show that the lexicographic rank of suffixes, the colexicographic rank of prefixes, and the identity function are order-preserving permutations (but not the only ones). We will prove this formally in Lemmas \ref{lem: stlex <= r} and \ref{lem: stcolex <= bar r}.

    \paragraph{Suffix tree path decomposition.}
    In this paper, a \emph{suffix tree path decomposition} (STPD) is an edge-disjoint collection of node-to-leaf paths covering all the suffix tree's edges built as follows. Since we will never start a path on an implicit suffix tree node, we can equivalently reason about path decompositions of the \emph{suffix trie}.
    We describe how to obtain the STPD associated with a given order-preserving permutation $\pi$, as we believe this will help the reader to better understand our technique. Then, we will make the construction fully formal.
    Let $\mathcal T\in \Sigma^n$ be a text.
    Imagine the process of inserting $\mathcal T$'s suffixes $\mathcal T[i,n]$ (for $i\in [n]$) in a trie in order of increasing $\pi(i)$. 
    The path associated with the first suffix $\mathcal T[\pi^{-1}(1),n]$ in this order is the one starting in the root and continuing with characters $\mathcal T[\pi^{-1}(1),n]$. 
    When inserting the $j$-th ($j>1$) suffix $\mathcal T[\pi^{-1}(j),n]$, let $k = \max_{j'<j}\rlce(\pi^{-1}(j'), \pi^{-1}(j))$ be the longest common prefix between the $j$-th suffix and all the previous suffixes in the order induced by $\pi$. The corresponding new path in the decomposition is the one starting in the suffix tree locus of string $\mathcal T [\pi^{-1}(j), \pi^{-1}(j)+k-1]$ and labeled with string $\mathcal T [\pi^{-1}(j)+k,n]$. In other words, the path associated with $\mathcal T[\pi^{-1}(j),n]$ is its suffix that ``diverges'' from the trie containing the previous suffixes in the order induced by $\pi$. See Figure \ref{fig:STLEX}, where leaves (suffixes) $\mathcal T[i,n]$ are sorted left-to-right in order of increasing $\pi(i)$ where $\pi = ISA$.
    
    The core of our indexing strategy is to store in a colexicographically-sorted array $\PDA$ all the \emph{distinct} integers $\pi^{-1}(j)+k$ obtained in this process (that is, the starting positions of paths in $\mathcal T$). We now formalize this intuition.
    
	

\paragraph{$\LPF$ and $\PDA$ arrays.}
We first introduce the concept of the Generalized Longest Previous Factor Array $\LPF_{\pi}$. Intuitively, this array stores the lengths of the longest common prefixes of a string's suffixes in the order induced by $\pi$. This array generalizes the well-known \emph{Permuted Longest Common Prefix} ($\PLCP$) array (obtained when taking $\pi = \ISA$ to be the lexicographic rank of the text's suffixes) and the $\LPF$ array (obtained when taking $\pi = id$ to be the identity function).
    
	\begin{definition}[Generalized Longest Previous Factor Array $\LPF_{\mathcal S, \pi}$]\label{def:LPF array}
    Let $\mathcal S \in\Sigma^n$ be a string and 	
        $\pi : [n] \rightarrow [n]$ be a permutation. The \emph{generalized Longest Previous Factor array} $\LPF_{\mathcal S, \pi}[1,n]$ associated with $\mathcal S$ and $\pi$ is the integer array that, for $i\in[n]$, is defined as:
		$$
        \LPF_{\mathcal S, \pi}[i] =
        \begin{cases}
            0 & \text{ if}\ \pi(i) = 1,  \\
            \max_{\pi(j)<\pi(i)} \rlce(j,i) & \text{ otherwise.}
        \end{cases}
		$$
		
	\end{definition}

	\begin{definition}[Path Decomposition Array $\PDA_{\mathcal S, \pi}$]\label{def: STPD array}
		Let $\mathcal S \in\Sigma^n$ be a string and 	
        $\pi : [n] \rightarrow [n]$ be an order-preserving permutation. The \emph{(suffix tree) Path Decomposition Array} $\PDA_{\mathcal S, \pi}$ associated with $\mathcal S$ and $\pi$ is the set $\{j = i+\LPF_{\mathcal S, \pi}[i]\ :\ i\in[n]\}$ sorted in colexicographic order of the corresponding string's prefixes $\mathcal S[1,j]$. 
	\end{definition}

        It will always be the case that the indexed string $\mathcal S$ is fixed in our discussion, so we will simply write $\LPF_{\pi}$ and $\PDA_{\pi}$ instead of $\LPF_{\mathcal S, \pi}$ and $\PDA_{\mathcal S, \pi}$, respectively. When also $\pi$ is clear from the context, we will write $\LPF$ and $\PDA$.
        \begin{example}\label{example:st-lex}
        Consider the STPD of Figure \ref{fig:STLEX}. In this example, $\pi(i)$ is the rank of leaf $i$ in lexicographic order (in other words, $\pi = \ISA$): $\pi = (4, 5, 8, 11, 7, 10, 6, 9, 3, 2, 1)$. Then, the corresponding $\LPF$ array corresponds to the $\PLCP$ array: $\LPF = \PLCP = [2,1,4,3,2,1,0,0,1,0,0]$. For instance, $\LPF[1] = 2$ because $\pi(1)=4$ and the suffixes $\mathcal T[j,n]$ with $\pi(j) < \pi(1)$ are those starting in positions $j\in \{11,10,9\}$. Among those, the one with the longest common prefix with $\mathcal T[1,n]$ is $\mathcal T[9,n]$, and their longest common prefix is $2 = \LPF[1]$.

        At this point, the sequence $i+\LPF[i]$ for $i=1, \dots, n$ is equal to $(3,3,7,7,7,7,7,8,10,10,11)$. The \emph{Path Decomposition Array} is the array containing the distinct values in such a sequence, sorted colexicographically: $\PDA = [11, 10, 3, 7, 8]$ (that is, $j$ precedes $j'$ in the order if and only if $\mathcal T[1,j]$ is colexicographically smaller than $\mathcal T[1,j']$). 
        \end{example}

        The STPD associated with $\pi$ and the corresponding array $\PDA_\pi$ are said to be \emph{order-preserving} if $\pi$ is order-preserving.

        The expert reader might have noticed that, for the permutation $\pi$ used in Example \ref{example:st-lex} (lexicographic rank), the values in $\PDA$ are in a one-to-one correspondence with the \emph{irreducible LCP values} (known to be at most $r$ in total, which gives a hint of how we will later prove $|\PDA_\pi| \le r$ for this particular $\pi$). 
        As a matter of fact, our technique generalizes the notion of \emph{irreducible values} to any array $\LPF_\pi$ such that $\pi$ is order-preserving. 
        First, note that the array $\LPF_\pi$ is almost nondecreasing: 

        \begin{lemma}\label{lem: LPF almost nondecreasing}
            For any string $\mathcal S\in\Sigma^n$, 
            if $\pi$ is order-preserving then for every $i>1$ it holds $\LPF_\pi[i] \ge \LPF_\pi[i-1]-1$. In particular, the sequence $i+\LPF_\pi[i]$, for $i=1, \dots, n$, is nondecreasing. 
        \end{lemma}
        \begin{proof}
        If, for a contradiction, it were $k = \LPF_\pi[i-1] \ge \LPF_\pi[i] + 2 \ge 2$, then the position $j-1$ with $\pi(j-1)<\pi(i-1)$ such that $\mathcal S[j-1,j+k-2] = \mathcal S[i-1,i+k-2]$ would satisfy 
        $\mathcal S[j-1,j] = \mathcal S[i-1,i]$ hence, by the order-preserving property of $\pi$,
        $\pi(j)<\pi(i)$ would hold. But then, since $\mathcal S[j,j+k-2] = \mathcal S[i,i+k-2]$, we would have $\LPF_\pi[i] \ge \rlce(i,j)\ge k-1 \ge \LPF_\pi[i] + 1$, a contradiction. 
        \end{proof}

        Values where the inequality of Lemma \ref{lem: LPF almost nondecreasing} is strict are of particular interest:
	
        \begin{definition}\label{def:irreducible LPF value}
            Let $\mathcal S\in\Sigma^n$ be a string and
            $\pi:[n]\rightarrow [n]$ be an order-preserving permutation. We say that $i\in [n]$ is an \emph{irreducible $\LPF_{\pi}$ position} if and only if either $i=1$ or $\LPF_{\pi}[i] \neq \LPF_{\pi}[i-1]-1$.
        \end{definition}	

        It is a simple observation that $\LPF_{\pi}[i] = \LPF_{\pi}[i-1]-1$ is equivalent to $(i-1) + \LPF_{\pi}[i-1] = i + \LPF_{\pi}[i]$, from which we obtain: 

        \begin{remark}\label{rem: PDA = irreducible}
            For any order-preserving permutation $\pi:[n] \rightarrow [n]$, 
            $\{x \in \PDA_\pi\} = \{i + \LPF_{\pi}[i]\ :\ i\ \mathrm{is\ an\ irreducible\ }\LPF_\pi \mathrm{\ position}\}$, hence $|\PDA_\pi|$ is equal to the number of irreducible $\LPF_{\pi}$ positions (since $\PDA_\pi$ contains distinct values).
        \end{remark}

        The following lemma generalizes the well-known relation between irreducible \PLCP~values and the Burrows-Wheeler transform \cite[Lemma~4]{KarkkainenMP09} in more general terms.

        \begin{lemma}
            \label{lemma: generalized iLCP vs bwt}
            Let $\mathcal S \in\Sigma^n$ be a string and $\pi : [n] \rightarrow [n]$ be an order-preserving permutation. Let moreover $i>1$ be 
            an irreducible $\LPF_\pi$ position. 
            Then, for every $j>1$ with $\pi(j-1)<\pi(i-1)$ and $\rlce(i,j)=\LPF_\pi[i]$, it holds that $\mathcal S[i-1]\ne\mathcal S[j-1]$.
        \end{lemma}
        \begin{proof}
            Let $i>1$ be an irreducible $\LPF_\pi$ position. We analyze separately the cases (i) $\LPF_\pi[i]=0$ and (ii) $\LPF_\pi[i]>0$.
            
            (i) If $\LPF_\pi[i]=0$ and $i>1$ is irreducible, then by definition of irreducible position  $\LPF_\pi[i-1] \neq \LPF_\pi[i]+1$ holds. Combining this with Lemma \ref{lem: LPF almost nondecreasing} we obtain $\LPF_\pi[i-1] < \LPF_\pi[i]+1 = 1$, hence $\LPF_\pi[i-1]=0$.
            Assume now, for a contradiction, that there 
            exists $j>1$ with $\pi(j-1)<\pi(i-1)$, $\rlce(i,j) = \LPF_\pi[i] = 0$, and $\mathcal S[i-1] = \mathcal S[j-1]$. 
            In particular, this implies $\rlce(i-1,j-1) = 1$.
            Then, $0 = \LPF_\pi[i-1] \ge \rlce(i-1,j-1) = 1$, a contradiction. 

            (ii) Let $\LPF_\pi[i]>0$ and $i>1$ be irreducible. Let moreover $j>1$ be such that $\rlce(i,j) = \LPF_\pi[i]$. Assume, for a contradiction, that $\mathcal S[i-1] = \mathcal S[j-1]$. 
            Then, since $\pi(j-1)<\pi(i-1)$, $\LPF_\pi[i-1] \ge \rlce(i-1,j-1) = \LPF_\pi[i]+1$ holds. 
            On the other hand, by Lemma \ref{lem: LPF almost nondecreasing} any order-preserving $\pi$ must satisfy $\LPF_\pi[i-1] \le \LPF_\pi[i]+1$.
            We conclude that $\LPF_\pi[i-1] = \LPF_\pi[i]+1$, 
            hence $i$ cannot be an irreducible position and we obtain a contradiction.
        \end{proof}

        Later we will use Remark~\ref{rem: PDA = irreducible} and Lemma~\ref{lemma: generalized iLCP vs bwt} to bound the size of our compressed suffix tree.

    Lemma \ref{lem:PDA samples} formalizes the following intuitive fact about order-preserving STPDs.
    Consider an STPD built on string $\mathcal S$ using order-preserving permutation $\pi$.
    If for a suffix tree node $u = \locus(\alpha)$, the three nodes $parent(u)$, $u$, and $\child(u,a)$ belong to the same STPD path (for some $a\in \Sigma$), then for any other outgoing label $b\in\Sigma$ of node $u$, string $\alpha\cdot b$ suffixes at least one sampled prefix $\mathcal S[1,t]$, for some $t\in \PDA_\pi$.

    \begin{lemma}\label{lem:PDA samples}
    Let $\mathcal S \in\Sigma^n$ be a string and $\pi : [n] \rightarrow [n]$ be an order-preserving permutation.
    For any one-character right-extension $\alpha\cdot a$ 
    of a right-maximal substring $\alpha$ of $\mathcal S$, define $\pi(\alpha\cdot a) = \min\{\pi(i-1)\ :\ \mathcal S[1,i]\ \mathrm{is\ suffixed\ by\ }\alpha\cdot a\}$. 

    Then, for any two distinct one-character right-extensions $\alpha\cdot a \neq \alpha\cdot b$ of $\alpha$ with $\pi(\alpha\cdot a) < \pi(\alpha\cdot b)$, there exists $t \in \PDA_\pi$ such that $\mathcal S[1,t]$ is suffixed by $\alpha\cdot b$.
    \end{lemma}
    \begin{proof}
    Let $a,b\in\Sigma$ with $\pi(\alpha\cdot a) < \pi(\alpha\cdot b)$. 
    Let $j',j$ be such that $\pi(j') = \pi(\alpha\cdot a) < \pi(\alpha\cdot b) = \pi(j)$.
    In particular, $\mathcal S[1,j'+1]$ is suffixed by $\alpha\cdot a$ and  $\mathcal S[1,j+1]$ is suffixed by $\alpha\cdot b$.  
    Then, by the order-preserving property of $\pi$ it holds that $\pi(j'-|\alpha|+1) < \pi(j-|\alpha|+1)$, hence
    $\LPF_\pi[j-|\alpha|+1] \ge \rlce(j-|\alpha|+1, j'-|\alpha|+1) = |\alpha|$. On the other hand, 
    by definition of $\pi(\alpha\cdot b) = \min\{\pi(i-1)\ :\ \mathcal S[1,i]\ \mathrm{is\ suffixed\ by\ }\alpha\cdot b\}$, 
    we also have that $\LPF_\pi[j-|\alpha|+1] \le |\alpha|$. We conclude $\LPF_\pi[j-|\alpha|+1] = |\alpha|$. 
    But then, $(j-|\alpha|+1) + \LPF_\pi[j-|\alpha|+1] = j+1 \in \PDA_\pi$. Since, as observed above, $\mathcal S[1,j+1]$ is suffixed by $\alpha\cdot b$, the claim follows by taking $t=j+1$.
    \end{proof}

We conclude with the following result, implying that $|\PDA_\pi|$ is a reachable compressibility measure for any order-preserving $\pi$. 

\begin{theorem}\label{thm: compressing PDA}
    Let $\mathcal S \in\Sigma^n$ be a string over integer alphabet $\Sigma$ and let $\pi : [n] \rightarrow [n]$ be an order-preserving permutation. Then, $\mathcal S$ can be compressed in $O(|\PDA_\pi|\log n)$ bits of space. 
\end{theorem}
\begin{proof}
    Our compressed representation is as follows. For each irreducible $\LPF_\pi$ position $i$, we store the quadruple $(i,\LPF_\pi[i],s_i, \mathcal S[i+\LPF_\pi[i]])$, where $s_i$ is the ``source'' of $i$, that is, any integer $s_i \in [n]$ such that $\pi(s_i) < \pi(i)$ and $\mathcal S[s_i,s_i+\LPF_\pi[i]-1] = \mathcal S[i,i+\LPF_\pi[i]-1]$ (if $\LPF_\pi[i]=0$, the latter condition is true for any $s_i\in[n]$). This set of quadruples takes $O(|\PDA_\pi|\log(n\sigma)) = O(|\PDA_\pi|\log n)$ bits of space by Remark \ref{rem: PDA = irreducible}. 

    We show how to reconstruct any character $\mathcal S[j]$ given $j \in [n]$ and the above representation. We consider two cases. 
    
    (i) If $j=i+\LPF_\pi[i]$ for some quadruple $(i,\LPF_\pi[i],s_i, \mathcal S[i+\LPF_\pi[i]])$, then $\mathcal S[j] = \mathcal S[i+\LPF_\pi[i]]$ (explicitly stored) and we are done. 

    (ii) Otherwise, find any quadruple $(i,\LPF_\pi[i],s_i, \mathcal S[i+\LPF_\pi[i]])$ such that $i \le j < i + \LPF_\pi[i]$ 
    (there must exist at least one such quadruple since $j$ is either irreducible or reducible). 
    Let $j_1 = s_i + (j-i)$.
    By the definition of $s_i$, it holds that $\mathcal S[j] = \mathcal S[j_1]$. Moreover, since $\mathcal S[s_i,s_i+\LPF_\pi[i]-1] = \mathcal S[i,i+\LPF_\pi[i]-1]$, $\pi(s_i) < \pi(i)$, and by the order-preserving property of $\pi$, it follows that $\pi(j_1) < \pi(j)$. We repeat recursively the above procedure to extract $\mathcal S[j_1]$. 
    Let $j, j_1, j_2, \dots$ be the sequence of text positions, with $\mathcal S[j] = \mathcal S[j_1] = \mathcal S[j_2] = \dots$, obtained by repeating recursively the above procedure. Since $\pi(j) > \pi(j_1) > \pi(j_2) > \dots$, and since $\pi(q) \in [n]$ for all $q\in[n]$, it follows that the above recursive procedure must stop in case (i) after at most $n$ steps.  
\end{proof}

We proceed as follows. 
First, we discuss the notable order-preserving permutation given by the lexicographic rank of the text's suffixes (Section \ref{sec:stlex}). This permutation will enable us to support most suffix tree operations in $O(r)$ space on top of the text oracle. Then, in Section \ref{sec: general STPD index} we describe a general locating mechanism working on any order-preserving STPD. 
After that, we focus on another particular order-preserving permutation: the
colexicographic rank of the text's prefixes (Section \ref{sec:stcolex}). This case is particularly interesting because it allows us to simplify the locating algorithm of Section~\ref{sec: general STPD index}, and will lead to a practical solution (used in our experiments in Section~\ref{sec:experiments}). Finally, in Section~\ref{sec:stpos} we consider yet another remarkable order-preserving permutation: identity. This permutation will allow us to locate efficiently the leftmost and rightmost pattern occurrences. 

\subsection{Lexicographic rank ($\stlex$): suffix tree navigation}\label{sec:stlex}

Figure \ref{fig:STLEX} depicts the STPD obtained by choosing $\pi = \ISA=\SA^{-1}$ to be the Inverse Suffix Array (in this subsection, $\pi$ will always be equal to $\ISA$). We denote with $\stlexm = \PDA_\pi$ the path decomposition array associated with this permutation $\pi$. Similarly, $\stlexp$ denotes the path decomposition array associated with the dual permutation $\bar\pi(i) = n-\ISA[i]+1$. The following properties hold: 

\begin{lemma}\label{lem: stlex <= r}
    Let $\mathcal T\in \Sigma^n$ be a text.
    The permutations $\pi, \bar\pi$ defined as	$\pi(i) = \ISA[i]$ and $\bar\pi(i) = n-\ISA[i]+1$ for $i\in[n]$ are order-preserving for $\mathcal T$. Furthermore, it holds that $|\stlexm| \le r$ and $|\stlexp| \le r$.
\end{lemma}
\begin{proof}
    For every $i,j\in[n-1]$ such that $\mathcal{T}[i,n]<_{\text{lex}}\mathcal{T}[j,n]$ and $T[i]=T[j]$, it holds that $\mathcal{T}[i+1,n]<_{\text{lex}}\mathcal{T}[j+1,n]$ by definition of the lexicographic order. 
    This proves the order-preserving property for $\pi$ and $\bar\pi$.
    
    By Remark~\ref{rem: PDA = irreducible}, it holds that $|\stlexm|$ is equal to the number of irreducible $\LPF_\pi$ positions. 
    We bijectively map each irreducible $\LPF_\pi$ position $i$ to $\BWT[\ISA[i]]$ and show that $\BWT[\ISA[i]]$ is the beginning of an equal-letter BWT run. 
    This will prove $|\stlexm|\le r$.
    Symmetrically, to prove $|\stlexp|\le r$ we bijectively map each irreducible $\LPF_{\bar\pi}$ position $i$ to $\BWT[\ISA[i]]$ and show that $\BWT[\ISA[i]]$ is the \emph{end} of an equal-letter BWT run. Since this case is completely symmetric to the one above, we omit its proof.
    Let $i$ be an irreducible $\LPF_\pi$ position. We analyze the cases (i) $i=1$ and (ii) $i>1$ separately.
    \begin{compactenum}[(i)]
        \item If $i=1$, then $\BWT[\ISA[i]] = \$$. Since the symbol $\$$ occurs only once in $\mathcal T$, $\BWT[\ISA[i]]$ is the beginning of an equal-letter BWT run. 
    
        \item If $i>1$, then either $\ISA[i] = 1$ and therefore $\BWT[\ISA[i]] = \BWT[1]$ is the beginning of an equal-letter BWT run, or $\ISA[i] > 1$. In the latter case, let $j = \SA[\ISA[i]-1]$ (in particular, $\ISA[j] = \ISA[i] - 1$). 
        If $j=1$, then $\BWT[\ISA[j]] = \$$. Since the symbol $\$$ occurs only once in $\mathcal T$, $\BWT[\ISA[i]]$ is the beginning of an equal-letter BWT run. 
        In the following, we can therefore assume $i>1$ and $j>1$.
    \end{compactenum}
    Then, by the definition of $\LPF_\pi = \PLCP$, it holds that $\LPF_\pi[i] = \rlce(i,j)$. We show that it must be $\BWT[\ISA[i]] = \mathcal T[i-1] \neq \mathcal T[j-1] = \BWT[\ISA[j]] = \BWT[\ISA[i]-1]$, which proves the main claim. If, for contradiction, it was $\mathcal T[i-1] = \mathcal T[j-1]$, then $\ISA[j] < \ISA[i]$ would imply $\pi(j-1) = \ISA[j-1] < \ISA[i-1] = \pi(i-1)$.
    But then, Lemma~\ref{lemma: generalized iLCP vs bwt} would imply $\mathcal T[i-1] \neq \mathcal T[j-1]$, a contradiction. 
\end{proof}

\subsubsection{Data structure.}
Let $\mathcal T\in \Sigma^n$ be a text. 
We describe a data structure of $O(r)$ space supporting the suffix tree queries of Theorem \ref{thm:st operations} on top of any text oracle supporting Longest Common Extension (and, optionally, fingerprinting) queries on $\mathcal T$. We store the following components. 

\begin{compactenum}[(1)]
    \item The array $\stlex$ containing the integers $\{j=(i-1)\ :\ i\in \stlexm \cup \stlexp \cup \{n+1\}\} \subseteq \{0,\dots, n\}$ sorted increasingly according to the colexicographic order of the corresponding text prefixes $\mathcal T[1,j]$ (if $j=0$, then $\mathcal T[1,j]$ is the empty string). Note that, by Lemma \ref{lem: stlex <= r}, it holds that $|\stlex| \le 2r+1$. 

    \item Let $\# \notin \Sigma$ be a new character not appearing in $\Sigma$ (taken to be lexicographically larger than all the characters in $\Sigma$). We store a string $\mathcal L[1,|\stlex|] \in (\Sigma \cup \{\#\})^{|\stlex|}$ defined as 
    \[
    \mathcal L[i] =
    \begin{cases}
        \#                      & \text{if }\stlex[i]=n,\\
       \mathcal T[\stlex[i]+1]  & \text{otherwise.}
    \end{cases}
    \] 

We store $\mathcal L$ with a wavelet tree \cite{navarro2014wavelet}, taking $O(|\stlex|)$ space and supporting rank and select operations in $O(\log \sigma)$ time. 

\item The string $\mathcal F[1,|\stlex|] \in (\Sigma \cup \{\#\})^{|\stlex|}$ obtained by sorting lexicographically the characters of $\mathcal L$.
We store $\mathcal F$ with a wavelet tree supporting rank and select operations in $O(\log \sigma)$ time. 

\item Let $\mathcal{LF}$ and $\mathcal{FL}$ be the permutations of $[|\stlex|]$ defined as follows. 
For any integers $i,j,k$, if $\mathcal L[i] = a$ is the $k$-th occurrence of $a$ in $\mathcal L$ and $\mathcal F[i] = a$ is the $k$-th occurrence of $a$ in $\mathcal F$, then $\mathcal{LF}(i) = j$ and $\mathcal{FL}(j) = i$ (note that $\mathcal{FL} = \mathcal{FL}^{-1}$). 
$\mathcal{LF}$ and  $\mathcal{FL}$ can be evaluated in $O(\log\sigma)$ time by running a constant number of rank and select operations on $\mathcal L$ and $\mathcal F$. 
\end{compactenum}

\begin{definition}\label{def:stlex'}
    We denote with $\stlex'$ the permutation of $\stlex$ defined as follows: for any $j\in [|\stlex|]$, $\stlex'[j] = \stlex[\mathcal{FL}(j)]$ or, equivalently (since $\mathcal{FL}$ and $\mathcal{LF}$ are inverse of each other), $\stlex[j] = \stlex'[\mathcal{LF}(j)]$.
\end{definition}

We store one Range Minimum and one Range Maximum data structure on the array $\pi(\stlex')$, supporting queries in constant time in $O(|\stlex|)$ bits of space (we do not need to store $\stlex'$ explicitly).

\begin{remark}
    While this is not fundamental for our discussion below, it may be helpful from an intuitive point of view to observe that string $\mathcal L$, once removed character $\#$ from it, is a subsequence of length $|\stlex|-1$ of the colexicographic Burrows-Wheeler transform ($\BWT$). Similarly, permutations $\mathcal{LF}$ and $\mathcal{FL}$ are the counterpart of functions $LF$ and $FL$ typically used with the BWT.
\end{remark}

\subsubsection{Node representation.}

Suffix tree navigation operations are supported on a particular representation of (explicit) suffix tree nodes that we describe next. First, Lemma \ref{lem:PDA samples} immediately implies:

\begin{corollary}\label{cor:stlex - st nodes}
    Let 
    $u$ be an explicit suffix tree node and $\alpha$ be such that $u= \locus(\alpha)$. Then, $\alpha$ suffixes $\mathcal T[1,j]$ for at least one value $j\in \stlex$.
\end{corollary}


In the Following, we describe our representation of suffix tree nodes.

\begin{definition}[Suffix tree node representation]\label{def: ST node representation}
We represent suffix tree node $u = \locus(\alpha)$ with the tuple of integers 
$$
R_u = (b,e,i_{min}, i_{max}, |\alpha|), \text{ where }
$$ 
\begin{itemize}
    \item $\stlex[b,e]$ is the colexicographic range of $\alpha$ in $\stlex$ (by Corollary \ref{cor:stlex - st nodes}, $e \ge b$ always holds);
    \item $\mathcal T[i_{min}, i_{min} + |\alpha|-1] = \alpha$ is the occurrence of $\alpha$ minimizing $\pi(i_{min})$, that is, $\mathcal T[i_{min}, n]$ is the lexicographically-smallest suffix prefixed by $\alpha$;
    \item $\mathcal T[i_{max}, i_{max} + |\alpha|-1] = \alpha$ is the occurrence of $\alpha$ maximizing $\pi(i_{max})$, that is, $\mathcal T[i_{max}, n]$ is the lexicographically-largest suffix prefixed by $\alpha$; and
    \item $|\alpha|$ is the string depth of $u$.
\end{itemize}    
\end{definition}

\begin{remark}\label{rem: range of leaves}
    Observe that $n \in \stlex[1,2]$.   
    In particular, if $1\in \stlexm \cup \stlexp$ then $0 \in \stlex$. Since $0$ corresponds to the text's prefix $\mathcal T[1,0]$ (the empty string), in this case $\stlex[1]=0$ because the empty string is smaller than any other text prefix. In that case, $\stlex[2]=n$ because $\mathcal T[1,n]$ (ending with $\$$) is the second colexicographically-smallest sampled prefix. If, on the other hand, $1\notin \stlexm \cup \stlexp$, then $\stlex[1]=n$. 
\end{remark}

Based on the above remark: 

\begin{definition}\label{def: range of leaves}
    We denote with $i^*\in\{1,2\}$ the integer such that $\stlex[i^*]=n$.
\end{definition}

Letting $u$ be a leaf, observe that $\alpha(u)$ ends with character $\$$. It follows that the colexicographic range of $\alpha(u)$ in $\stlex$ is $\stlex[i^*,i^*]$. This will be used later. 

\subsubsection{Suffix tree operations.} 

Next, we show how to support a useful subset of suffix tree operations on our data structure. 
This will prove Theorem \ref{thm:st operations}.
In the description below, suffix tree operations will take as input node representations ($R_u$) instead of nodes themselves ($u$).
Recall that we are assuming we have access to a text oracle supporting longest common extension ($\lce$) and random access queries on $\mathcal T$ in $O(t)$ time and fingerprinting queries in $O(h)$ time. 

\paragraph{Root.} Let $u = \locus(\epsilon)$ be the suffix tree root. $\root()$  returns $R_u = (1,|\stlex|,i_{min},i_{max},0)$ in $O(1)$ time, where $\mathcal T[i_{min},n]$ is the lexicographically-smallest text suffix and $\mathcal T[i_{max},n]$ is the lexicographically-largest text suffix (in other words, $i_{min} = \SA[1] = n$ and $i_{max} = \SA[n]$).

\paragraph{String depth.} Let $R_u = (b,e,i_{min}, i_{max}, \ell)$. Then, $\sdepth(R_u)$ simply returns $\ell$ in $O(1)$ time.

\paragraph{Ancestor.} Let $R_u = (b,e,i_{min}, i_{max}, \ell)$ and $R_{u'} = (b',e',i_{min}', i_{max}', \ell')$. 
If $\ell > \ell'$, $u$ cannot be an ancestor of $u'$ so $\ancestor(R_u,R_{u'})$ returns false. Otherwise, $\ancestor(R_u,R_{u'})$ returns true if and only if $\rlce(i_{min},i_{min}') \ge \ell$, that is, if and only if the string $\alpha = \mathcal T[i_{min}, i_{min}+\ell-1]$ with $u=\locus(\alpha)$ is a prefix of $\alpha' = \mathcal T[i_{min}', i_{min}'+\ell'-1]$ with $u'=\locus(\alpha')$. This operation runs in $O(t)$ time.

\paragraph{Is leaf.}  Let $R_u = (b,e,i_{min}, i_{max}, \ell)$. Then, $\isleaf(R_u)$ returns true (in $O(1)$ time) if and only if $i_{min} = i_{max}$, if and only if $b=e= i^*$ (see Definition \ref{def: range of leaves}), if and only if the string $\alpha = \mathcal T[i_{min}, i_{min}+\ell-1]$ with $u = \locus(\alpha)$ ends with $\$$, that is, $i_{min}+\ell-1 = n$ (equivalently, $\mathcal T[i_{min}+\ell-1] = \$$).

\paragraph{Locate leaf.} Let $R_u = (b,e,i_{min}, i_{max}, \ell)$. The output of $\locate(R_u)$ is defined only if $\isleaf(R_u)$ is true. In that case, $\locate(R_u)$ returns $i_{min}$ ($=i_{max}$) in $O(1)$ time.

\paragraph{Leftmost/rightmost leaves.} Let $R_u = (b,e,i_{min}, i_{max}, \ell)$. Then, in $O(1)$ time we can compute $\lleaf(R_u) = (i^*,i^*,i_{min}, i_{min}, n-i_{min}+1)$ and $\rleaf(R_u) = (i^*,i^*,i_{max}, i_{max}, n-i_{max}+1)$.

\paragraph{Edge label.} Let $(u,u')$ be a suffix tree edge, with $R_u = (b,e,i_{min}, i_{max}, \ell)$ and $R_{u'} = (b',e',i_{min}', i_{max}', \ell')$. The function $\labeln(R_u,R_{u'})$ returns $(i_{min}'+\ell, i_{min}'+\ell'-1)$ in $O(1)$ time.

\paragraph{Next leaf.}
Nishimoto and Tabei~\cite{Nishimoto21Move} showed that, starting from $i\in[n]$, $k$ consecutive applications $\bar\phi(i), \bar\phi^2(i), \dots, \bar\phi^k(i)$ of the permutation defined below\footnote{This permutation is usually denoted as $\phi^{-1}$. Here we instead use the symbol $\bar\phi$.} can be computed in $O(\log\log(n/r) + k)$ time with a data structure using $O(r)$ words of space.

\begin{definition}[$\phi$-function]\label{def:phi}
    Let $\bar\phi : [n] \rightarrow [n]$ be defined as
    $$
    \bar\phi(i) = 
    \begin{cases}
        \SA[\ISA[i]+1]  & \text{if } \ISA[i]<n,\\
        \SA[1]          & \text{if } \ISA[i]=n.
    \end{cases}
    $$
\end{definition}

Observe that, by the very definition of $\bar\phi$, leaf $\locus(\mathcal T[\bar\phi(i),n])$ is the next leaf in lexicographic order after leaf $\locus(\mathcal T[i,n])$ (unless the latter is the suffix tree's rightmost leaf). 

Let $u$ be a leaf, with $R_u = (i^*,i^*,i, i, n-i+1)$ (see Definition \ref{def: range of leaves} for the definition of $i^*$). Then, unless $u$ is the rightmost leaf, $next(R_u) = (i^*,i^*,\bar\phi(i), \bar\phi(i), n-\bar\phi(i)+1)$. It follows that the structure of Nishimoto and Tabei can be used to evaluate $k$ consecutive applications of $\next(\cdot)$ in $O(\log\log(n/r)+k)$ time. 

\paragraph{Smallest children label and successor child.} Let $u$ be a node with $R_u = (b,e,i_{min}, i_{max}, \ell)$. Operation $\first(R_u)$ returns the alphabetically-smallest label in $\mathcal L[b,e]$ (i.e., in $\out(u)$). Operation $succ(R_u,a)$ returns the alphabetically-smallest label in $\mathcal L[b,e]$ ($\out(u)$) being larger than $a$. Both queries reduce to an \emph{orthogonal range successor} (also known as \emph{range next value}) query in the range $\mathcal L[b,e]$, an operation that can be solved in $O(\log\sigma)$ time on wavelet trees \cite{navarro2014wavelet}. In both queries, if $\# \in \mathcal L[b,e]$ then we simply ignore it.

\paragraph{Child by letter.} 
The function $\child(R_u, a)$ is the most technically-interesting operation.
Let $u$ be a node with $R_u = (b,e,i_{min}, i_{max}, \ell)$ and $a\in\Sigma$ be a letter. 
Let $\alpha = \mathcal T[i_{min}, i_{min} + \ell -1]$ be such that $u = \locus(\alpha)$.
Lemma \ref{lem:PDA samples} implies the following corollary:
\begin{corollary}\label{cor:outgoing labels from u}
    Let $u = \locus(\alpha)$, and let $\stlex[b,e]$ be the range containing the text positions $j \in \stlex$ such that $\mathcal T[1,j]$ is suffixed by $\alpha$.
    All (and only) the letters labeling the outgoing edges from node $u$ appear in $\mathcal L[b,e] \setminus \{\#\}$, that is, $\out(u) = \{c\ :\ c \in \mathcal L[b,e] \wedge c \neq \#\}$.
\end{corollary}
\begin{proof}
    Ignoring character $\#$, by definition  $\mathcal L[b,e]$ only contains characters following occurrences of $\alpha$ in $\mathcal T$, that is, characters in $\out(u)$.

    We now show that every $c\in \out(u)$ belongs to $\mathcal L[b,e]$.
    Assume that $c\in \out(u)$ is not the alphabetically-smallest character in $\out(u)$ (the other case --- $c$ is not the alphabetically-largest character in $\out(u)$ --- is symmetric and the following proof adapts by replacing $\pi$ with $\bar\pi$).     
    Let $a\neq c$ be the alphabetically-smallest character in $\out(u)$. 
    Then, any text suffix being prefixed by $\alpha\cdot a$ is lexicographically smaller than any text suffix being prefixed by $\alpha\cdot c$, hence (since $\pi=\ISA$) it holds $\pi(\alpha\cdot a)<\pi(\alpha\cdot c)$ (see Lemma \ref{lem:PDA samples} for the definition of this overloading of $\pi$ to right-extensions of right-maximal strings).
    But then, Lemma \ref{lem:PDA samples} implies that there exists $t \in \PDA_\pi = \stlexm$ with $\mathcal T[1,t]$ being suffixed by $\alpha\cdot c$, hence $t-1 \in \stlex[b,e]$ and therefore $c \in \mathcal L[b,e]$. 
\end{proof}


%

Following Corollary \ref{cor:outgoing labels from u}, if $a \notin \mathcal L[b,e]$ (a test taking $O(\log\sigma)$ time using rank and select operations on $\mathcal L$), then we can return $\child(R_u,a) = (0,0,0,0,0)$, signaling that no outgoing edge from $u$ is labeled with $a$.

Otherwise, we first show how to compute $i_{min}'$ and $i_{max}'$
such that
$\mathcal T[i_{min}',n]$ and $\mathcal T[i_{max}',n]$  are the lexicographically-smallest and lexicographically-largest suffixes being prefixed by $\alpha\cdot a$, respectively. After that, we show how to use this information to compute $\child(R_u,a)$.

Observe that $|\out(u)| \ge 2$, since $u$ is an explicit suffix tree node. 
We distinguish three cases. (i) $a$ is neither the alphabetically-largest nor the alphabetically-smallest label in $\mathcal L[b,e] \setminus \{\#\}$, i.e. $a \neq \min\{c \in \mathcal L[b,e] \setminus \{\#\}\}$ and $a \neq \max\{c \in \mathcal L[b,e] \setminus \{\#\}\}$ hold; (ii) $a$ is the alphabetically-smallest label in $\mathcal L[b,e] \setminus \{\#\}$, i.e., $a = \min\{c \in \mathcal L[b,e] \setminus \{\#\}\}$; (iii) $a$ is the alphabetically-largest label in $\mathcal L[b,e] \setminus \{\#\}$, i.e., $a = \max\{c \in \mathcal L[b,e] \setminus \{\#\}\}$. 

To simplify the discussion below, we rephrase Lemma~\ref{lem:PDA samples} to the particular STPDs we are using in this section (i.e. those derived from $\pi = \ISA$ and $\bar\pi[i]=n-\ISA[i]+1$):

\begin{corollary}\label{cor:stlex samples}

Let $\alpha$ be right-maximal, and let $\stlex[b,e]$ be the range containing the text positions $i \in \stlex$ such that $\mathcal T[1,i]$ is suffixed by $\alpha$. 
Let moreover $C = \{c\in \Sigma\ :\ \alpha\cdot c\ \mathrm{occurs\ in\ }\mathcal T\}$ be the set of all the characters extending $\alpha$ in the text. Then:

\begin{compactenum}[(a)]
        \item Let $c \in (C~\setminus~\min C)$, and let $\mathcal T[j -|\alpha|+1,n]$ be the lexicographically-smallest suffix being prefixed by $\alpha\cdot c = \mathcal T[j -|\alpha|+1,j+1]$. Then, $j \in \stlex[b,e]$.
         \item Let $c \in (C~\setminus~\max C)$, and let $\mathcal T[j -|\alpha|+1,n]$ be the lexicographically-largest suffix being prefixed by $\alpha\cdot c = \mathcal T[j -|\alpha|+1,j+1]$. Then, $j \in \stlex[b,e]$.       
    \end{compactenum}
\end{corollary}

We now show how to compute $i_{min}'$ and $i_{max}'$ in cases (i-iii).
\begin{compactenum}[(i)]
\item By Corollary \ref{cor:stlex samples}, $i_{min}'',i_{max}'' \in \stlex[b,e]$, where $\mathcal T[i_{min}''-\ell+1,n]$ and $\mathcal T[i_{max}''-\ell+1,n]$  are the lexicographically-smallest and lexicographically-largest suffixes being prefixed by $\alpha\cdot a$, respectively. Observe that $i_{min}''$ and $i_{max}''$ are not necessarily distinct (they are equal precisely when descending to a leaf).
We locate the leftmost $\mathcal L[L] = a$ and rightmost $\mathcal L[R] = a$ occurrences of letter $a$ in $\mathcal L[b,e]$ using rank and select operations on $\mathcal L$, in $O(\log\sigma)$ time.
Then, we map all the occurrences of $a$ in $\mathcal L[b,e]$ to the corresponding $a$'s in $\mathcal F$ by applying $\mathcal{LF}$: those occurrences correspond to the range $\mathcal F[\mathcal{LF}(L), \mathcal{LF}(R)] = \mathcal F[b', e']$, hence $b'$ and $e'$ can be found in $O(\log\sigma)$ time using rank and select operations on $\mathcal L$ and $\mathcal F$. 
Observe that, by definition of the permutation $\stlex'$ of $\stlex$ (Definition \ref{def:stlex'}), it holds that $i_{min}'',i_{max}'' \in \stlex'[b',e']$. 
Then, our Range Minimum/Maximum data structures queried in range $\pi(\stlex')[b',e']$ will yield two positions $j_{min}, j_{max}$ such that $\stlex'[j_{min}] = i_{min}''$ and $\stlex'[j_{max}] = i_{max}''$. Although we are not storing $\stlex'$ explicitly (we could, since it would just take $|\stlex|$ further memory words), we can retrieve $i_{min}''$ and $i_{max}''$ by just applying $\mathcal FL$: by definition of $\stlex'$ (Definition \ref{def:stlex'}), we have that $\stlex'[j] = \stlex[\mathcal{FL}(j)]$ for any $j \in [|\stlex|]$. Then, $i_{min}' = i_{min}''-\ell+1$ and $i_{max}' = i_{max}''-\ell+1$. 

\item The letter $a$ is the alphabetically-smallest label in $\mathcal L[b,e] \setminus \{\#\}$. Since $|\out(u)| \ge 2$ (because $u$ is an internal suffix tree node), we obtain that $a$ is not the alphabetically-largest label in $\mathcal L[b,e] \setminus \{\#\}$. Then, by Corollary \ref{cor:stlex samples}, $i_{max}'' \in \stlex[b,e]$, where $\mathcal T[i_{max}''-\ell+1,n]$  is the lexicographically-largest suffix being prefixed by $\alpha\cdot a$. We can find $i_{max}''$ following the same procedure described in point (i) above (resorting to Range Maxima Queries). At this point, since $a$ is the alphabetically-smallest label in $\mathcal L[b,e] \setminus \{\#\}$ (equivalently, in $\out(u)$) and $\mathcal T[i_{min},n]$ is the lexicographically-smallest suffix being prefixed by $\alpha$, it also holds that $\mathcal T[i_{min},n]$ is the lexicographically-smallest suffix being prefixed by $\alpha\cdot a$. But then, we are done: $i_{min}' = i_{min}$ and $i_{max}' = i_{max}''-\ell+1$.

\item The letter $a$ is the alphabetically-largest label in $\mathcal L[b,e] \setminus \{\#\}$. This case is completely symmetric to case (ii) so we omit the details.
\end{compactenum}

At this point, we know that $\mathcal T[i_{min}',n]$ and $\mathcal T[i_{max}',n]$ are the lexicographically-smallest and lexicographically-largest suffixes prefixed by $\alpha\cdot a$, respectively. 
Note that $i_{min}' = i_{max}'$ if and only if $\child(R_u,a)$ is a leaf. In this case, we simply return $\child(R_u,a) = (i^*,i^*,i_{max}',i_{max}',n-i_{max}'+1)$.

In the following, we can therefore assume that  $i_{min}' \neq i_{max}'$, hence $\child(R_u,a)$ is not a leaf.
Then, note that $\ell' = \rlce(i_{min}',i_{max}')$ is the string depth of $\child(u,a)$. In particular, $\child(u,a) = \locus(\alpha')$, where $\alpha' = \mathcal T[i_{min}',i_{min}'+\ell'-1]$.

Binary-searching $\alpha'$ in $\stlex$ then yields the maximal range $\stlex[b'',e'']$ such that, for all $j\in [b'',e'']$, $\alpha'$ is a suffix of $\mathcal T[1, \stlex[j]]$. Using random access queries on $\mathcal T$ to guide binary search, this process requires in total $O(\ell'\log |\stlex|)$ random access queries. 
We can do better if the text oracle supports $\lce$ queries.
Since $\alpha' = \mathcal T[i_{min}',i_{min}'+\ell'-1]$ is a substring of $\mathcal T$ itself, each binary search step can be implemented with one $\llce$ and one random access query (the latter, extracting one character). This yields $[b'',e'']$ in time $O(t \cdot \log|\stlex|) \subseteq O(t \cdot \log r)$.

Yet another solution uses \emph{z-fast tries} \cite{BelazzouguiBPV09,boldi-z-fast}, a machinery supporting \emph{internal suffix searches} (that is, the queried string is guaranteed to suffix at least one string in the pre-processed set) in a set of $q$ strings each of length bounded by $n$. This solution uses space $O(q)$ and answers suffix-search queries in $O(\log n)$ steps, each requiring computing the fingerprint of a substring of the query $\alpha'$. 
In our case, the set of strings is $\{\mathcal T[1,\stlex[i]]\ :\ i\in [|\stlex|]\}$ and we have that $\alpha' = \mathcal T[i_{min},i_{min}+\ell'-1]$ indeed suffixes at least one of them by Corollary \ref{cor:stlex - st nodes} and by the fact that $i_{min}' \neq i_{max}'$, hence $\alpha'$ is right-maximal. 
This solution finds  $[b'',e'']$ in $O(h\log n)$ time (or I/O operations, if $h$ is the I/O complexity of fingerprinting). 

We finally have all ingredients to return our result: $\child(R_u,a) = (b'',e'',i_{min}',i_{max}',\ell')$. If the oracle only supports $\lce$ queries and random access in $O(t)$ time, then $\child(R_u,a)$ is answered in $O(t\log r + \log\sigma)$ time. If the oracle also supports fingerprinting queries in $O(h)$ time, then the running time is $O(t + h\log n + \log\sigma)$. The latter is preferable if fingerprinting queries are answered faster than $\lce$ queries (e.g., with the text oracle of Prezza \cite{prezzaTalg21}). 

\paragraph{Putting everything together.}
To sum up, our data structure uses $O(|\stlex| + r) = O(r)$ space on top of the text oracle. Taking into account the complexity of all queries discussed above, this proves Theorem \ref{thm:st operations}. 

We can finally prove Corollary \ref{cor:locating on our ST}.

\locatingOurST*

\begin{proof}
We plug the random access oracle \cite{prezzaTalg21} in Theorem \ref{thm:st operations}. This oracle uses $n\log\sigma + O(\log n)$ bits of space and supports $\lce$ queries with $O(\log n)$ I/O complexity, fingerprinting with $O(1)$ I/O complexity, and extraction of $\ell$ consecutive characters with $O(1+\ell/B)$ I/O complexity\footnote{Even if in \cite{prezzaTalg21} they only prove $O(\ell)$ random access time for a contiguous block of $\ell$ characters in the RAM model, it is not hard to see that their data structure triggers $O(\ell/B)$ I/O operations in the I/O model: the data structure consists in the Karp-Rabin fingerprints of a sample of the text's prefixes. To extract a character, the structure combines the fingerprints (adjacent in memory) of two consecutive prefixes. This locality property leads immediately to the claimed $O(\ell/B)$ complexity.} provided that the alphabet size is polynomial ($\sigma \le n^{O(1)}$). Starting with $k=0$ and $R_u = root()$, suppose we have matched a right-maximal proper pattern's prefix $P[1,k]$ and found the representation $R_u$ of node $u = \locus(P[1,k])$. 
To continue matching, we compute $R_v = child(R_u,P[k+1])$. If the child exists, we get the edge's text pointers $(i,i+q-1) = \labeln(R_u,R_v)$ ($q$ is the edge's length) and match $P[k+1,\dots]$ with $\mathcal T[i,i+q-1]$ by random access, either reaching node $v$ (i.e. the end of the edge  $(u,v)$) or reaching the end of $P$. 
In the former case, we repeat the above operations to continue matching $P$ on descendants of $v$.
In the latter case, we are left to locate all the $occ$ pattern occurrences in the subtree rooted in $v$. 

To locate the pattern's occurrences, we have two choices. The first (less efficient, corresponding to the original one of Weiner~\cite{weiner1973linear}) is to navigate the subtree (of size $O(occ)$) rooted in $v$ using $\child(\cdot,\cdot)$, $\first(\cdot)$, and $\succn(\cdot,\cdot)$ queries. For each reached leaf $u$ (we use $\isleaf(R_u)$ to check whether we reached a leaf), we call $\locate(R_u)$ to report the corresponding pattern occurrence. 
The I/O complexity of this navigation using the text oracle \cite{prezzaTalg21} is $O(occ\cdot(t+h\log n+\log\sigma)) = O(occ\cdot \log n)$.

A more efficient solution consists of navigating just the leaves using $\next(\cdot)$. This can be done by computing $R_{left} = \lleaf(R_v)$, $R_{right} = \rleaf(R_v)$ and evaluating $occ$ consecutive applications of $\next(\cdot)$ starting from $R_{left}$ until reaching $R_{right}$ (again, calling $\locate(R_u)$ on each leaf $u$ to get a pattern occurrence). This solution locates the $occ$ pattern occurrences with $O(\log\log(n/r) + occ)$ I/O complexity.

Finding the locus $v$ of pattern $P$ required calling a $\child(\cdot,\cdot)$ operation for each of the $d$ traversed suffix tree edges, plus $O(d + m/B)$ I/O complexity for matching the pattern against the edges' labels. In total, finding $v$ has therefore  $O(d\cdot(t+h\log n+\log\sigma) + d + m/B) = O(d\log n + m/B)$ I/O complexity. Taking into account the I/O complexity of locating the pattern's occurrences, this proves Corollary \ref{cor:locating on our ST}.   
\end{proof}

\subsection{General locating mechanism for order-preserving STPDs}\label{sec: general STPD index}

In this subsection we provide a general locating mechanism working on any order-preserving STPD. First, we show how to locate the pattern occurrence $P = \mathcal T[i,j]$ minimizing $\pi(i)$ among all pattern occurrences. We call this occurrence \emph{primary}. Then we show that, starting from the (unique) primary occurrence, we can locate the remaining ones (called \emph{secondary}) by resorting to orthogonal point enclosure. 
As a matter of fact, this technique generalizes the $r$-index' $\phi$ function (Definition \ref{def:phi}, cases $\pi =\ISA$ and $\pi = \IPA$) to arbitrary order-preserving permutations. 
This increased generality with respect to Section \ref{sec:stlex}, comes at the price of not being able to support suffix tree queries (only pattern matching).

After this section, we will tackle the particular case $\pi = \IPA$ (colexicographic order of the text's prefixes), for which the locating algorithm that we describe here can be simplified. That particular case will lead to our optimized implementation able to beat the $r$-index both in query time (by orders of magnitude) and space usage.

From here until the end of the section, we assume $\mathcal T\in \Sigma^n$ is a text and $\pi$ is any order-preserving permutation on $\mathcal T$.
We start by classifying the occurrences of a given pattern $P\in\Sigma^+$ as follows.
    
    \begin{definition}[Primary/Secondary occurrence]
         For a string $P\in\Sigma^+$, an occurrence $\mathcal{T}[i,i+|P|-1]=P$ is said to be a \emph{primary occurrence} if and only if $\LPF_\pi[i]<|P|$. All the other occurrences are called \emph{secondary occurrences}.
    \end{definition}
    
    \begin{lemma}
        For any string $P\in\Sigma^+$ that occurs in $\mathcal{T}$ there exists exactly one primary occurrence $P = \mathcal T[i,i+|P|-1]$. Furthermore, such occurrence is the one minimizing $\pi(i)$.
    \end{lemma}
    \begin{proof}
        Let $P\in\Sigma^+$ occur in $\mathcal{T}$.
        To prove the existence of a primary occurrence, suppose for a contradiction that there is no primary occurrence of $P$. Let $\mathcal{T}[i,i+|P|-1]=P$ be the secondary occurrence minimizing $\pi(i)$.  
        By definition of secondary occurrence we have $\LPF_\pi[i]\ge|P|$, hence, by definition of $\LPF_\pi$, there must exist $i'\in[n]$ such that $\pi(i')<\pi(i)$ and $\rlce(i',i)\ge|P|$ implying that $\mathcal{T}[i',i'+|P|-1]=P$, which contradicts the minimality of $\pi(i)$.
        
        To prove the uniqueness of the primary occurrence, suppose for a contradiction that there are at least two primary occurrences $i,i'\in[n]$ with $i\ne i'$. Since $\pi$ is a permutation, exactly one of (i)~$\pi(i)<\pi(i')$ and (ii)~$\pi(i)>\pi(i')$ holds. Without loss of generality, assume that $\pi(i)<\pi(i')$. Then by definition of $\LPF_\pi$, $\LPF_\pi[i'] \ge \rlce(i',i)\ge|P|$, a contradiction. 
        This also shows, as a byproduct, that the unique primary occurrence $P = \mathcal T[i,i+|P|-1]$ is the one minimizing $\pi(i)$ among all (primary and secondary) occurrences.
    \end{proof}

\subsubsection{Finding the primary occurrence.}

The idea to locate the primary pattern occurrence is simple and can be visualized as a process of walking along the paths of the STPD, starting from the root. Whenever we find a mismatch between the pattern $P$ and the current STPD path, we change path by running a suffix search (e.g.\ binary search) on $\PDA_\pi$.
We first provide an intuition through an example (Example \ref{ex:locate primary}). Our algorithm is formalized in Algorithm~\ref{alg: locate primary}. Then, we prove the algorithm's complexity, correctness, and completeness. 

\begin{example}\label{ex:locate primary}
Consider the STPD of Figure \ref{fig:STLEX}, and let $P = CGCGAA$ be the query pattern. 
We start by matching $P$ with the characters of the path $\mathcal T[11,11] = \$$ starting at the root. The longest pattern prefix matching the path is $\epsilon$ (the empty string). To continue matching $P$, we need to change path. We binary search $\PDA_\pi = \stlexm$ looking for a sampled text prefix being suffixed by $\epsilon \cdot P[1] = C$ (i.e.\ the concatenation of the pattern's prefix matched so far and the first unmatched pattern's character). Two sampled prefixes (elements in $\stlexm$) satisfy this requirement: $\mathcal T[1,3]$ and $\mathcal T[1,7]$. By definition of our STPD, among them we need to choose the one $\mathcal T[1,i]$ minimizing $\pi(i) = \ISA[i]$. 
This choice can be performed in constant time with the aid of a Range Minimum Query data structure built on top of $\pi(\stlexm)$. 
Such prefix minimizing $\pi(i)$ is $\mathcal T[1,7]$. This means that we choose an STPD path labeled with $\mathcal T[7,11]$. 
Observe that this also means that $\mathcal T[7]$ is the lexicographically-smallest occurrence of $P[1]$.
We repeat the process, matching the remaining characters $P[1,6]$ of $P$. As can be seen in Figure \ref{fig:STLEX} characters $P[1,2] = CG$ match (purple path starting below the root). Then, the path continues with $A$ and the pattern with $P[3] = C$. As done above, we binary search $\PDA_\pi = \stlexm$ looking for a sampled text prefix being suffixed by $P[1,2]\cdot P[3] = CGC$. Now, only one sampled prefix matches: $\mathcal T[1,7]$, corresponding to an STPD path labeled with $\mathcal T[7,11]$.
Observe that this means that $\mathcal T[7]$ is the lexicographically-smallest occurrence of $P[3]$ being preceded by $P[1,2]$; in other words, $\mathcal T[7-2,7] = \mathcal T[5,7]$ is the lexicographically-smallest occurrence of $P[1,3]$. 
We therefore continue matching the remaining pattern's suffix $P[3,6]$ on 
$\mathcal T[7,11]$. This time, the whole pattern's suffix matches, hence we are done. 
Since the last binary search returned the sampled text's prefix $\mathcal T[1,7]$, and before running the search we already matched $P[1,2]$, we return pattern occurrence $\mathcal T[7-2,(7-2) + |P| -1] = \mathcal{T}[5,10]$ which, by construction of our STPD and by the order-preserving property of $\pi$, is the lexicographically-smallest one (in this example, also the only one). 
\end{example}

\paragraph{Data structures and search algorithm.} 
We show that, interestingly, a small modification of the search algorithm of Suffixient Arrays \cite{cenzato2025suffixientarraysnewefficient} (see Appendix \ref{app:suffixient arrays}) allows locating the primary pattern occurrence on any order-preserving STPD.
The search algorithm is sketched in Example~\ref{ex:locate primary} and formalized in Algorithm~\ref{alg: locate primary}. We need the following data structures:

\begin{enumerate}
    \item A Range Minimum Data structure on $\pi(\PDA_\pi)$, requiring just $O(|\PDA_\pi|)$ bits and answering queries in constant time. In Algorithm~\ref{alg: locate primary}, this data structure allows finding the $\arg\min$ at Line \ref{line:argmin} in constant time. 
    \item A data structure supporting suffix searches on the text's prefixes $\{\mathcal T[1,\PDA_\pi[t]]\ :\ t\in [|\PDA_\pi|]\}$. 
    Given $(i,j,c) \in [n]\times [n] \times \Sigma$,  $\sufsearch(i,j,c)$ returns the maximal range $[b,e]\subseteq [|\PDA_\pi|] $ such that $\mathcal T[1,\PDA_\pi[t]]$ is suffixed by $\mathcal T[i,j]\cdot c$ for all $t\in [b,e]$. 
    This structure is used in Line \ref{line:sufsearch} of Algorithm~\ref{alg: locate primary}.
    We discuss different possible implementations for this structure below. 
    \item A text oracle supporting random access on $\mathcal T$. To speed up $\sufsearch$ we may also require the oracle to support $\lce$ and/or fingerprinting queries on $\mathcal T$ (we obtain different performance depending on which queries are available, read below).
\end{enumerate}

\SetKwComment{Comment}{/* }{ */}
\begin{algorithm}[H]
\caption{Locating the primary pattern occurrence} \label{alg: locate primary}
\LinesNumbered
\KwIn{Pattern $P\in (\Sigma\setminus \{\$\})^m$.}
\KwOut{Primary occurrence of $P$, i.e. position $i\in [n]$ such that (1) $\mathcal T[i,i+m-1]=P$ and (2) $i$ minimizes $\pi(i)$ among all the occurrences of $P$. If $P$ does not occur in $\mathcal T$, return \texttt{NOT\_FOUND}.}
$i\gets \pi^{-1}(1)$\;
$j\gets 1$\ \Comment*[r]{Invariant: $P[1,j-1] = \mathcal T[i-j+1,i-1]$ minimizes $\pi(i-j+1)$}
\While{ $j \le m$ }{
    \If{$\mathcal T[i] \ne P[j]$\label{line:locate one if}} {
        $[b,e] \gets \sufsearch(i-j+1,i-1,P[j])$\;\label{line:sufsearch}
        $i' \gets \arg\min_{i'\in [b,e]}  \{ \pi(\PDA_\pi[i']) \}$\; \label{line:argmin}
        $i \gets \PDA_\pi[i']$\;\label{line:get PDA sample}
    }
    $i \gets i+1;\ j\gets j+1$\;\label{line:incr i,j}
}

\eIf{$\mathcal T[i-m,i-1] = P$\label{line:compare random access}}{
    \textbf{return } $i-m$\;
}{
    \textbf{return } \texttt{NOT\_FOUND}\;
}

\end{algorithm}

\paragraph{Complexity.}

As observed above, the $\arg\min$ operation in Line  \ref{line:argmin} takes just $O(1)$ time using a Range Minimum data structure built over $\pi(\PDA_\pi)$.

A first simple implementation of $\sufsearch(i,j,c)$ uses binary search on $\PDA_\pi$ and random access on $\mathcal T$. This solution executes $O(\log|\PDA_\pi|) \subseteq O(\log n)$ random access queries, each requiring the extraction of $O(m)$ contiguous text characters. If the random access oracle supports the extraction of $\ell$ contiguous text characters with $O(1+\ell/B)$ I/O complexity, then this solution has $O((1+m/B)\log n)$ I/O complexity. Due to its simplicity, small space usage (only array $\PDA_\pi$ and the text oracle are required), and attractive I/O complexity, this is the solution we implemented in practice in the index tested in Section \ref{sec:experiments}.

A second optimized implementation of $\sufsearch$ executes one $\llce$ and one random access query for every binary search step. This solution runs in $O(t\cdot \log|\PDA_\pi|) \subseteq O(t \log n)$ time (respectively, I/O complexity), where $t$ is the time complexity (respectively, I/O complexity) of $\llce$ and random access queries. 

As done in Section \ref{sec:stlex}, a faster solution can be obtained using z-fast tries \cite{BelazzouguiBPV09,boldi-z-fast}. 
Since z-fast tries are guaranteed to answer correctly only internal suffix search queries (that is, the queried string suffixes at least one of the strings in the trie), we build a separate z-fast trie for every distinct character $c\in\Sigma$. The z-fast trie associated with character $c$ is built over text prefixes $\{\mathcal T[1,i-1]\ :\ i\in \PDA_\pi \wedge \mathcal T[i]=c\}$. 
We store the z-fast tries in a map associating each $c\in \Sigma$ to the corresponding z-fast trie and supporting constant-time lookup queries. At this point, query $\sufsearch(i-j+1,i-1,c)$ is solved by issuing an internal suffix search query $\mathcal T[i-j+1,i-1]$ on the z-fast trie associated with character $c$. 
In order to return the result $[b,e] = \sufsearch(i-j+1,i-1,c)$, we store along the z-fast trie for character $c$ the number $\Delta_c = |\{i\in \PDA_\pi\ :\  \mathcal T[i]<c\}|$ of sampled text prefixes ending with a character being smaller than $c$. Letting $[b',e']$ be the colexicographic range (retrieved with the z-fast trie for $c$) of $\mathcal T[i-j+1,i-1]$ among the text prefixes $\{\mathcal T[1,i-1]\ :\ i\in \PDA_\pi \wedge \mathcal T[i]=c\}$, then the result of $\sufsearch(i-j+1,i-1,c)$ is $[b,e] = [b'+\Delta_c,e'+\Delta_c]$.
This implementation of $\sufsearch$ runs in $O(h\log m)$ time (respectively, I/O complexity), where $h$ is the running time (respectively, I/O complexity) of fingerprinting queries on the text oracle.

Observe that Algorithm~\ref{alg: locate primary} calls $\sufsearch$ once for every STPD path crossed while reaching $\locus(P)$ from the suffix tree root. The number of crossed paths is upper-bounded by the node depth $d$ of $\locus(P)$ in the suffix tree of $\mathcal T$.
In the end (Line \ref{line:compare random access}), Algorithm~\ref{alg: locate primary} compares the pattern with a text substring of length $m$.
Using the latter implementation of $\sufsearch$ (z-fast tries), we obtain: 

\begin{lemma}\label{lem:locate primary complexity}
    Let $\mathcal T\in \Sigma^n$ be a text and $\pi:[n]\rightarrow [n]$ be an order-preserving permutation. 
    Suppose we have access to an oracle  supporting fingerprinting queries on $\mathcal T$ in $O(h)$ time (respectively, I/O complexity) and extraction of $\ell$ contiguous text characters in $e(\ell)$ time (respectively, I/O complexity).
    Then, Algorithm~\ref{alg: locate primary} runs in  $O(d\cdot h\log m + e(m))$ time (respectively, I/O complexity), where $d$ is the node depth of $\locus(P)$ in the suffix tree of $\mathcal T$.
\end{lemma}

For example, using the text oracle of Prezza \cite{prezzaTalg21} the I/O complexity of Algorithm~\ref{alg: locate primary} becomes $O(d\log m + m/B)$.

\begin{lemma}\label{lem: locate primary correct}
    Algorithm~\ref{alg: locate primary} is correct and complete. 
\end{lemma}
\begin{proof}
    We prove that the following invariant is maintained by the \texttt{while} loop: if $P$ occurs in $\mathcal T$, then $P[1,j-1] = \mathcal T[i-j+1,i-1]$ is the occurrence of $P[1,j-1]$ minimizing $\pi(i-j+1)$. At the beginning, the invariant is clearly true by the way we choose $i$ and $j$ (in particular, $P[1,j-1]$ is the empty string).

    We show that the invariant is maintained after one execution of the \texttt{while} loop, assuming that $P$ occurs in $\mathcal T$.
    If the condition of the \texttt{if} statement in Line \ref{line:locate one if} fails, then $\mathcal T[i] = P[j]$, we just increment $i$ and $j$ in Line \ref{line:incr i,j}, and the invariant still holds by the order-preserving property of $\pi$. Otherwise, let $\mathcal T[i] \neq P[j]$. Then, since $P$ occurs in $\mathcal T$, string $\alpha = P[1,j-1]$ is right-maximal in $\mathcal T$: there exists a copy of $\alpha$ followed by $P[j]$ and one followed by $\mathcal T[i]\neq P[j]$. 
    By the inductive hypothesis, $P[1,j-1] = \mathcal T[i-j+1,i-1]$ minimizes $\pi(i-j+1)$ across all occurrences of $P[1,j-1]$. By the order-preserving property of $\pi$, $\mathcal T[i-j+1,i-1]$ also minimizes $\pi(i-1)$ across all occurrences of $P[1,j-1]$. Then, this means that $\pi(\alpha\cdot \mathcal T[i]) < \pi(\alpha\cdot P[j])$ (see Lemma \ref{lem:PDA samples} for the definition of this extension of $\pi$ to right-extensions of right-maximal strings) so, by Lemma~\ref{lem:PDA samples}, we have that there exists $t\in \PDA_\pi$ such that $\mathcal T[1,t]$ is suffixed by $\alpha\cdot P[j] = P[1,j]$. But then, $\sufsearch(i-j+1,i-1,P[j])$ (implemented with z-fast tries) is guaranteed to return the (non-empty) range $\PDA_\pi[b,e]$ containing all and only the text prefixes in $\PDA_\pi$ being suffixed by $P[1,j]$. Among those, Lines \ref{line:argmin}-\ref{line:get PDA sample} choose the one $i \in \PDA_\pi[b,e]$ minimizing $\pi(i)$. 
    At this point, by the order-preserving property of $\pi$, we have that $\mathcal T[i-j+1,i] = P[1,j]$ is the occurrence of $P[1,j]$ minimizing $\pi(i-j+1)$.
    In Line \ref{line:incr i,j} we increment $i$ and $j$, hence the invariant is maintained. This proves completeness and correctness under the assumption that $P$ occurs in $\mathcal T$.

    If, on the other hand, $P$ does not occur in $\mathcal T$, then the comparison in Line \ref{line:compare random access} detects this event and we correctly return \texttt{NOT\_FOUND}. This completes the proof. 
\end{proof}

\subsubsection{Locating the secondary occurrences.}

   We now describe how to locate the secondary occurrences. This subsection is devoted to proving the following lemma.
    \begin{lemma}\label{lem:finding secondary occ}
        Let $\mathcal T\in \Sigma^n$ be a text and $\pi:[n]\rightarrow[n]$ be an order-preserving permutation. Let $P\in\Sigma^m$ occur in $\mathcal T$. There exists a data structure that takes $O(|\PDA_\pi|)$ words of space and that, given the primary occurrence of $P$, finds all the $occ-1$ secondary occurrences in $O(occ\cdot \log|\PDA_\pi|) \subseteq O(occ\cdot \log n)$ time (equivalently, I/O complexity).
    \end{lemma}

    We start by covering $\mathcal{T}$ with $|\PDA_\pi|$ (possibly overlapping) \emph{phrases} (that is, substrings of $\mathcal{T}$), as follows: 

    \begin{definition}[Phrase cover of $\mathcal T$]
    \begin{itemize}
        \item \emph{(Type-1 phrases)} we associate phrase $\mathcal T[i]$ (one character) to each position $i \in [n]$ such that $\LPF_\pi[i]=0$.
         \item \emph{(Type-2 phrases)} We associate phrase $\mathcal T[i,i+\LPF_\pi[i]-1]$ to each irreducible $\LPF_\pi$ position $i$ such that $\LPF_\pi[i]>0$.
    \end{itemize}
    \end{definition}

    Observe that there are at most $|\PDA_\pi|$ type-2 phrases. 
    By definition, our cover of $\mathcal T$ into phrases satisfies the following properties: 

    \begin{remark}
        Any pattern occurrence $\mathcal T[i',i'+m-1]$ crossing a type-1 phrase $\mathcal T[i]$ (i.e. $i\in [i',i'+m-1]$) is primary. 
        To see this observe that, due to $\LPF_\pi[i]=0$, we know that $c = \mathcal T[i]$ is the occurrence of $c$ minimizing $\pi(i)$, hence by the order-preserving property of $\pi$, $\mathcal T[i',i'+m-1] = P$ is the occurrence of $P$ minimizing $\pi(i')$.
    \end{remark}  

    \begin{remark}
        For each secondary occurrence $\mathcal T[i',i'+m-1]$, there exists a type-2 phrase $\mathcal T[i,i+\LPF_\pi[i]-1]$ containing it, i.e. $[i',i'+m-1] \subseteq [i,i+\LPF_\pi[i]-1]$. This immediately follows from the definition of secondary occurrence. 
    \end{remark}

    Observe that each type-2 phrase is copied from another text position with a smaller $\pi$. We formalize this fact as follows:

    \begin{definition}[Phrase source]\label{def:phrase source}
        Let $\mathcal T[i,i+\LPF_\pi[i]-1]$ be a type-2 phrase. Then, we define $\SRC[i] = j$ to be the position $j$ with $\pi(j)<\pi(i)$ such that $\mathcal T[i,i+\LPF_\pi[i]-1] = \mathcal T[j,j+\LPF_\pi[i]-1]$ (in case of ties, choose the leftmost such position $j$).
    \end{definition}

    We also use the fact that each phrase can be split into a prefix containing strictly decreasing $\LPF_\pi$ values and the remaining suffix. This will play a crucial role in our locating algorithm, as it will allow us to report each secondary occurrence exactly once. 

    \begin{definition}[Reducible prefix]\label{def:reducible prefix}
        Let $\mathcal T[i,i+\LPF_\pi[i]-1]$ be a type-2 phrase. The \emph{reducible prefix} of $\mathcal T[i,i+\LPF_\pi[i]-1]$ is its longest prefix $\mathcal T[i,i+\ell_i-1]$, with $1\le \ell_i \le \LPF_\pi[i]$, such that $\LPF_\pi[j] = \LPF_\pi[j-1]-1$ for each $j \in [i+1,i+\ell_i-1]$.
    \end{definition}

    \begin{remark}\label{rem: reducible prefix}
    Observe that all $\LPF_\pi$ positions in $[i+1,i+\ell_i-1]$ are reducible (that is, not irreducible).    
    \end{remark}

    See Figure \ref{fig:locateall} for a running example. 

    \begin{figure}
        \centering
        \includegraphics[width=0.7\linewidth]{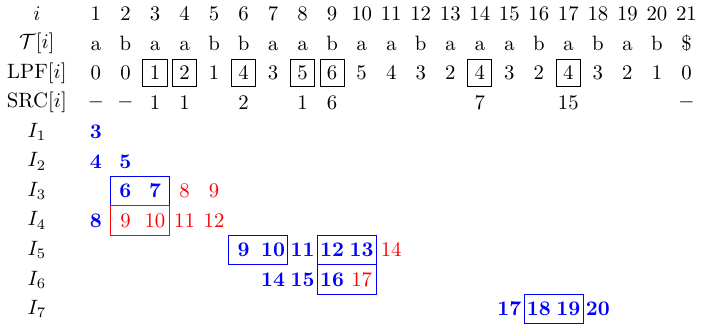}
        \caption{
        Consider this particular string $\mathcal T$ and take $\pi = id$ to be the identity function. Then, $\LPF_\pi = \LPF$ is the classic Longest Previous Factor array. In the third line, we put in boxes the irreducible $\LPF$ positions $i$ such that $\LPF[i]>0$; these are the beginnings of our phrases (7 phrases in total). For example, consider such a position $i=9$. Since $\LPF[9]=6$, the corresponding phrase is $\mathcal T[9,9+6-1] = \mathcal T[9,14]$. The corresponding $\LPF$ sub-array is $\LPF[9,14] = (\underline 6,\underline 5,\underline 4,\underline 3,\underline 2,4)$, where we underlined the phrase's reducible prefix (Definition \ref{def:reducible prefix}) of length $\ell_9=5$. In Lines $I_1$---$I_7$ in the figure, we show each type-2 phrase's source as an interval containing the phrase's positions, colored according to its reducible prefix (in blue) and the remaining suffix (red). For example, 
        consider again phrase $\mathcal T[9,14]$, and let $s_9 = \SRC[9] = 6$ be its source. Row $I_5$ shows the phrase's source $\mathcal T[s_9, s_9 + \LPF[9] - 1] = \mathcal T[6, 11]$, highlighting in blue the source of the phrase's reducible prefix and in red the remaining suffix. In lines $I_1$---$I_7$, we moreover show, using boxes, how the occurrences of string $ba$ intersect the phrases' sources; the box is blue when the occurrence intersects (the source of) the phrase's reducible prefix, red otherwise. By the way we design the rectangles indexed in our orthogonal point enclosure data structure (see Definition \ref{def: rectangle} and Lemma \ref{lemma: orthogonal point enclosure} below), only occurrences in a blue box trigger the location of further secondary occurrences. 
        \textbf{Example of locate}. After locating the primary occurrence $\mathcal T[2,3]=ba$, an orthogonal point enclosure query on point $(2,3)$ yields rectangle $[2,3]\times[2,5]$, corresponding to the phrase starting in position 6 (and whose source is shown in line $I_3$). This leads us to discover the secondary occurrence $\mathcal T[6,7]$. 
        Observe that, even if $\mathcal T[2,3]$ overlaps also the source of the phrase starting in position 8, point $(2,3)$ is not contained in the phrase's rectangle $[1,1]\times [1,5]$ (in Line $I_4$, this fact can be visualized by noting that the box lies completely in the red area). This is correct, as otherwise following this source would lead us locating secondary occurrence $\mathcal T[9,10]$, which will be located again later when issuing the orthogonal point enclosure query on point $(6,7)$.
        }
        \label{fig:locateall}
    \end{figure}
        
    The idea to locate secondary occurrences, is to associate every type-2 phrase $\mathcal T[i,i+\LPF_\pi[i]-1]$ with a 2-dimensional rectangle whose coordinates reflect the source of the whole phrase and of its reducible prefix: 

    \begin{definition}[Rectangle associated with a type-2 phrase]\label{def: rectangle}
        Let $\mathcal T[i,i+\LPF_\pi[i]-1]$ be a type-2 phrase, $s_i = \SRC[i]$ be its source, and $\ell_i$ be the length of its reducible prefix. We associate with $\mathcal T[i,i+\LPF_\pi[i]-1]$ the 2-dimensional rectangle $[s_i,s_i+\ell_i-1]\times[s_i,s_i+\LPF_\pi[i]-1] \subseteq [n]^2$, and label it with position $i$.
    \end{definition}

    \begin{example}
        Continuing the running example of Figure \ref{fig:locateall}, consider again phrase $\mathcal T[9,14]$, let $s_9 = \SRC[9] = 6$ be its source, and $\ell_9 = 5$ be the length of its reducible prefix. This phrase is associated with rectangle $[s_9,s_9+\ell_9-1]\times[s_9,s_9+\LPF[9]-1] = [6,10] \times [6,11]$. In Figure \ref{fig:locateall}, Line $I_5$, the first coordinate $[6,10]$ is highlighted in blue, while the second coordinate $[6,11]$ is the whole colored interval in Line $I_5$ (red and blue). 
    \end{example}

    Our locating algorithm relies on the following well-known data structure result on the \emph{orthogonal point enclosure problem}: 
    
    \begin{lemma}[Orthogonal point enclosure, {\cite[Theorem~6]{Chazelle86}}]
        \label{lemma: orthogonal point enclosure}
        Let $\mathcal{R}$ be a collection of axis-parallel two-dimensional rectangles in $[n]^2$. There exists an $O(|\mathcal{R}|)$-space data structure supporting the following query: given a point $(x,y)$, find all rectangles $[a,b]\times [c,d] \in \mathcal{R}$ containing $(x,y)$, i.e. such that $x\in[a,b]$ and $y\in[c,d]$. The query is answered in $O(\log|\mathcal{R}|+k)$ time, where $k$ is the number of returned rectangles. 
    \end{lemma}

    Let $\mathcal R$ be the set of rectangles defined in Definition \ref{def: rectangle}. As noted above, $|\mathcal R| \le |\PDA_\pi|$. We build the data structure of Lemma \ref{lemma: orthogonal point enclosure} on $\mathcal R$. The structure uses $O(|\PDA_\pi|)$ words of space and answers orthogonal point enclosure queries in $O(\log|\PDA_\pi|+k) \subseteq O(\log n+k)$ time.

    \paragraph{Locating algorithm.} Let $\mathcal T[j,j+m-1] = P$ be the primary occurrence of $P$, found with Algorithm~\ref{alg: locate primary}. To locate the secondary pattern occurrences, initialize a stack $Q\gets\{j\}$. While $Q$ is not empty:
    \begin{enumerate}
        \item Pop an element $x$ from $Q$ and report pattern occurrence $x$.
        \item Locate all rectangles in $\mathcal R$ containing point $(x,x+m-1)$. For each such retrieved rectangle $[s_i,s_i+\ell_i-1]\times[s_i,s_i+\LPF_\pi[i]-1]$ labeled with position $i$, push $i+x-s_i$ in $Q$.
    \end{enumerate}

    \paragraph{Correctness.}

    To prove that we only report pattern occurrences, we prove inductively that the stack always contains only pattern occurrences. 
    At the beginning, the stack contains the primary occurrence of $P$. 
    Assume inductively that, before entering in Step 1 of the above algorithm, the stack contains only pattern occurrences. In Step 1, we pop $x$. Then, let $[s_i,s_i+\ell_i-1]\times[s_i,s_i+\LPF_\pi[i]-1]$ be a rectangle, labeled with position $i$, found in Step 2 by querying point $(x,x+m-1)$. 
    In particular, $[x, x+m-1] \subseteq [s_i, s_i + \LPF_\pi[i]-1]$. Then, 
    by the definition of $\SRC$ (Definition \ref{def:phrase source}), it holds $\mathcal T[i, i+\LPF_\pi[i]-1] = \mathcal T[s_i, s_i + \LPF_\pi[i]-1]$, hence $[i+x-s_i, i+x-s_i+m-1] \subseteq [i, i + \LPF_\pi[i]-1]$, therefore $\mathcal T[i+x-s_i, i+x-s_i+m-1] = \mathcal T[x, x+m-1]$ is a pattern occurrence. This proves that in Step 2 we only push pattern occurrences in $Q$.

    \paragraph{Completeness.} We prove that every pattern occurrence at some point is pushed in the stack (and is therefore located). Assume, for a contradiction, that this is not true. Let $\mathcal T[j,j+m-1]$ be the secondary occurrence minimizing $\pi(j)$ that is never pushed on the stack. Then, there exists a type-2 phrase $\mathcal T[i,i+\LPF_\pi[i]-1]$ entirely containing $\mathcal T[j,j+m-1]$, i.e. $[j,j+m-1] \subseteq [i,i+\LPF_\pi[i]-1]$. Without loss of generality, let $i$ be the largest position with such a property. 
    Let $\ell_i$ be the length of the phrase's reducible prefix. 
    Observe that $\mathcal T[j',j'+m-1] = \mathcal T[j,j+m-1]$, where $j' = s_i + (j-i)$, is another pattern occurrence with $\pi(j')<\pi(j)$. 
    We distinguish two cases.
    
    (i) $j \in [i,i+\ell_i-1]$. Then, the point $(j',j'+m-1)$ belongs to the rectangle $[s_i,s_i+\ell_i-1]\times[s_i,s_i+\LPF_\pi[i]-1]$. This means that also pattern occurrence $\mathcal T[j',j'+m-1]$ is never pushed on the stack, otherwise at some point the algorithm would pop $j'$ from the stack and locate $j$ as well. This is a contradiction, since either $j'$ is the primary occurrence, or $\pi(j')<\pi(j)$ and we assumed that $\mathcal T[j,j+m-1]$ is the secondary occurrence minimizing $\pi(j)$ that is never pushed on the stack.

    (ii) $j \in [i+\ell_i,i+\LPF_\pi[i]-1]$. Then, either $i+\ell_i = n+1$ (a contradiction with the fact that $j \ge i+\ell_i$) or, by Definition \ref{def:reducible prefix}, position $i' = i+\ell_i$ is irreducible with $\LPF_\pi[i']>0$. Then, $\mathcal T[i',i'+\LPF_\pi[i']-1]$ is a phrase with $i'>i$ and, as observed in Remark \ref{lem: LPF almost nondecreasing}, it must be $i'+\LPF_\pi[i']-1 \ge i+\LPF_\pi[i]-1 \ge j+m-1$. This means that $\mathcal T[i',i'+\LPF_\pi[i']-1]$ is a type-2 phrase with $i'>i$ that contains $\mathcal T[j,j+m-1]$ (i.e. $[j,j+m-1] \subseteq [i',i'+\LPF_\pi[i']-1]$), a contradiction since we assumed that $i$ was the largest such position. 

    \paragraph{Complexity.} In order to prove Lemma \ref{lem:finding secondary occ}, we are left to show that every occurrence is pushed at most once in the stack, that is, we never report an occurrence twice.
    Assume, for a contradiction, that secondary occurrence $\mathcal T[j,j+m-1]$ is pushed at least twice on the stack. 
    Without loss of generality, we can assume that $\mathcal T[j,j+m-1]$ is the secondary occurrence pushed at least twice on the stack that minimizes $\pi(j)$. 
    This can happen in two cases, treated below.

    (i) There exists \emph{one} irreducible $\LPF_\pi$ position $i$, associated with rectangle $[s_i,s_i+\ell_i-1]\times[s_i,s_i+\LPF_\pi[i]-1] \in \mathcal R$, that causes pushing $j$ twice on the stack. This can happen only in one situation: when position $j' = j - i + s_i$ (uniquely determined from $i$ and $j$) is pushed twice on the stack, therefore querying twice the orthogonal point enclosure data structure on point $(j',j'+m-1)$ leads to pushing twice $j = i + j' - s_i$ in $Q$. In such a case, however, note that $\pi(j') < \pi(j)$, hence $\mathcal T[j',j'+m-1]$ is a secondary occurrence pushed at least twice on the stack with $\pi(j')< \pi(j)$. This contradicts the assumption that $\mathcal T[j,j+m-1]$ is the occurrence with this property minimizing $\pi(j)$. 

    (ii) There exist (at least) \emph{two distinct} irreducible $\LPF_\pi$ positions $i_1 \neq i_2$, associated with rectangles $[s_{i_t},s_{i_t}+\ell_{i_t}-1]\times[s_{i_t},s_{i_t}+\LPF_\pi[{i_t}]-1] \in \mathcal R$ for $t=1, 2$, respectively, each causing to push $j$ on the stack. Similarly to the case above, this happens when both positions $j_t = j - i_t + s_{i_t}$ (for $t=1,2$) are pushed on the stack, therefore querying the orthogonal point enclosure data structure on point $(j_t,j_t+m-1)$ leads to pushing $j = i + j_t - s_{i_t}$ in $Q$ for $t=1,2$.
    In particular, this means that $j \in [i_1, i_1 + \ell_{i_1}-1] \cap [i_2, i_2 + \ell_{i_2}-1]$, that is, $j$ belongs to the reducible prefix of both phrases. Assume, without loss of generality, that $i_1 < i_2$ (the other case is symmetric). Then, $i_1 < i_2 \le j \le i_1 + \ell_{i_1} - 1$: in other words, $i_2 \in [i_1+1, i_1 + \ell_{i_1}-1]$. This is a contradiction, because all $\LPF_\pi$ positions $[i_1, i_1 + \ell_{i_1}-1]$ are reducible (Remark \ref{rem: reducible prefix}), while $i_2$ is irreducible.

    \paragraph{Putting everything together.} Combining Lemmas \ref{lem:locate primary complexity}, \ref{lem: locate primary correct}, and \ref{lem:finding secondary occ} we obtain: 

   \begin{theorem}\label{thm:locate all general}
       Let $\mathcal T\in \Sigma^n$ be a text and $\pi:[n]\rightarrow [n]$ be an order-preserving permutation. 
       Suppose we have access to an oracle  supporting fingerprinting queries on $\mathcal T$ in $O(h)$ time (respectively, I/O complexity) and the extraction of $\ell$ contiguous characters of $\mathcal T$ in $e(\ell)$ time (respectively, I/O complexity).
       Then, there exists a data structure taking $O(|\PDA_\pi|)$ words of space on top of the oracle and able to report the $occ$ occurrences of any pattern $P\in \Sigma^m$ in $O(d\cdot h\log m + e(m) + occ\cdot \log |\PDA_\pi|) \subseteq O(d\cdot h\log m + e(m) + occ\cdot \log n)$ time (respectively, I/O complexity), where $d$ is the node depth of $\locus(P)$ in the suffix tree of $\mathcal T$.
   \end{theorem}

For example, using the text oracle of Prezza \cite{prezzaTalg21}, the I/O complexity of Theorem \ref{thm:locate all general} becomes $O(d\log m + m/B + occ\cdot \log n)$. The resulting index uses $O(|\PDA_\pi|)$ memory words on top of the oracle of $n\log\sigma + O(\log n)$ bits.

\subsection{Colexicographic rank ($\stcolex$): smaller-space I/O-efficient pattern matching}\label{sec:stcolex}

We now move to the particular case $\pi = \IPA$, for which the algorithm described in the previous section can be simplified, leading to a very space-efficient and fast implementation.

Figure \ref{fig:STCOLEX} depicts the STPD obtained by choosing $\pi = \IPA$ to be the Inverse Prefix Array. We denote with $\stcolexm = \PDA_\pi$ the path decomposition array associated with this permutation $\pi$. Similarly, $\stcolexp$ denotes the path decomposition array associated with the dual permutation $\bar\pi(i) = n-\IPA[i]+1$. 
The following properties hold: 

\begin{lemma}\label{lem: stcolex <= bar r}
    Let $\mathcal T\in \Sigma^n$ be a text.
    The permutations $\pi, \bar\pi$ defined as	$\pi(i) = \IPA[i]$ and $\bar\pi(i) = n-\IPA[i]+1$ for $i\in[n]$ are order-preserving for $\mathcal T$. Furthermore, $|\stcolexm| \le \bar r$ and $|\stcolexp| \le \bar r$ hold.
\end{lemma}
\begin{proof}

    For every $i,j\in[n-1]$ such that $\mathcal{T}[1,i]<_{\text{colex}}\mathcal{T}[1,j]$ and $T[i+1]=T[j+1]$, it holds that $\mathcal{T}[1,i+1]<_{\text{colex}}\mathcal{T}[1,j+1]$ by definition of the colexicographic order. This proves the order-preserving property for $\pi$ and $\bar\pi$.

    In order to upper-bound $|\stcolexm|$ and $|\stcolexp|$, we rotate $\mathcal T$ to the right by one position, i.e. we replace  $\mathcal T$ by $\mathcal T[n]\cdot \mathcal T[1,n-1] = \$ \cdot \mathcal T[1,n-1]$. This rotation does not change the number of $\coBWT$ equal-letter runs (see Appendix \ref{app:basic concepts}). 
    The number of irreducible positions in $\LPF_\pi$ and $\LPF_{\bar\pi}$, on the other hand, increases at most by one: since $\$$ is the smallest alphabet character, the colexicographic order of prefixes does not change after the rotation. 
    Then, $i$ was irreducible before the rotation if and only if $i+1$ is irreducible after the rotation. Additionally, $i=1$ is always irreducible after the rotation.  
    It follows that if we prove the bound for the rotated $\mathcal T$, then the bound also holds for the original one.

    By Remark \ref{rem: PDA = irreducible}, $|\stcolexm|$ is equal to the number of irreducible $\LPF_\pi$ positions.
    Let $i' = i+\LPF_\pi[i]-1$.
    We map bijectively each irreducible $\LPF_\pi$ position $i$ to $\coBWT[\IPA[i']]$ and show that $\coBWT[\IPA[i']]$ is the beginning of an equal-letter $\coBWT$ run.
    This will prove $|\stcolexm|\le \bar r$.
    Symmetrically, to prove $|\stcolexp|\le \bar r$ we map bijectively each irreducible $\LPF_{\bar\pi}$ position $i$ to $\coBWT[\IPA[i']]$ and show that $\coBWT[\IPA[i']]$ is the end of an equal-letter $\coBWT$ run. Since this case is completely symmetric to the one above, we omit its proof.

    Let $i$ be an irreducible $\LPF_\pi$ position and let $i' = i+\LPF_\pi[i]-1$ ($\le n$, by definition of $\LPF_\pi$). We analyze separately the cases (i) $i'=n$ and (ii) $i'<n$.

    (i) If $i'=n$, then $\coBWT[\IPA[i']] = \$$. Since the symbol $\$$ occurs only once in $\mathcal T$, $\coBWT[\IPA[i']]$ is the beginning of an equal-letter BWT run. 

    (ii) If $i'<n$, then either $\IPA[i'] = 1$ and therefore $\coBWT[\IPA[i']] = \coBWT[1]$ is the beginning of an equal-letter coBWT run, or $\IPA[i'] > 1$. In the latter case, let $j' = \PA[\IPA[i']-1]$ (in particular, $\IPA[j'] = \IPA[i'] - 1$) and $j = j' - \LPF_\pi[i] +1$.
    If $j'=n$, then $\coBWT[\IPA[i']-1] = \coBWT[\IPA[j']] = \$ \neq \coBWT[\IPA[i']]$, hence again we have that $\coBWT[\IPA[i']]$ is the beginning of an equal-letter BWT run. In the following, we can therefore assume both $i'<n$ and $j'<n$.
    
    Observe that it must be the case that $\mathcal T[j,j'] = \mathcal T[i,i']$. 
    To see this, let $t$ be such that $\pi(t) < \pi(i)$ and $\LPF_\pi[i] = \rlce(i,t)$. 
    Let $t' = t + \LPF_\pi[i]-1$.
    Then, by the order-preserving property of $\pi$, it must be $\pi(t') < \pi(i')$. Moreover, since $\LPF_\pi[i] = \rlce(i,t)$, then $\llce(i',t') \ge \LPF_\pi[i]$. 
    But then, by definition of $j'$ and $\pi = \IPA$, it must be $\llce(i',j') \ge \llce(i',t') \ge \LPF_\pi[i]$, which implies $\mathcal T[j,j'] = \mathcal T[i,i']$. 
    
    Then, since $\pi(j') = \pi(i')-1 < \pi(i')$, by the order-preserving property of $\pi$ it must be $\pi(j)<\pi(i)$.
    Finally, by the very definition of $\LPF_\pi$ it must be the case that $\coBWT[\IPA[i']-1] = \coBWT[\IPA[j']] = \mathcal T[j'+1] \neq \mathcal T[i'+1] = \coBWT[\IPA[i']]$, which proves the claim. In fact, if it were $\mathcal T[j'+1] = \mathcal T[i'+1]$ then we would obtain $\LPF_\pi[i] \ge \rlce(i,j) \ge \LPF_\pi[i] + 1$, a contradiction. 
\end{proof}

\begin{figure}[h!]
    \centering
    \includegraphics[width=0.6\linewidth,page=5]{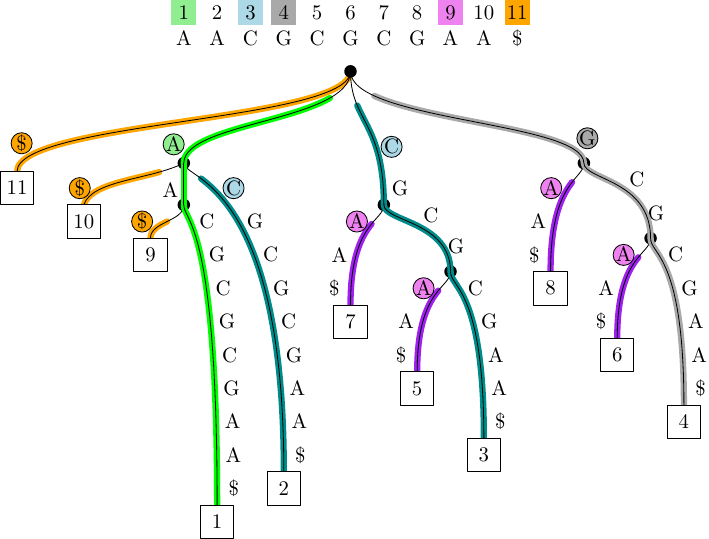}
    
    \caption{Example of the STPD obtained by taking $\pi=\IPA$ to be the colexicographic rank of the text's prefixes. For a better visualization, we do not sort suffix tree leaves according to $\pi$ as this would cause some edges to cross. 
    Paths $\mathcal T[j,n]$ are colored according to the color of $j \in \stcolexm$. 
    To see how the paths are built, consider the Prefix Array $\PA = [11,1,2,10,9,3,5,7,4,6,8]$ and the generalized Longest Previous Factor array $\LPF_\pi = [0,1,0,0,4,3,2,1,2,1,0]$. $\PA$ tells us the order in which we have to imagine inserting the text's suffixes in the trie. We show how the process works for the first three suffixes in the order induced by $\pi$. The first suffix to be inserted is $\mathcal T[11,11]$. Position $11+\LPF_\pi[11]=11+0=11$ is orange, hence an orange path labeled $\mathcal T[11,11]$ starts in the root. The next suffix is $\mathcal T[1,11]$. Position $1+\LPF_\pi[1]=1+0=1$ is green, hence a green path labeled $\mathcal T[1,11]$ is inserted. The next suffix is $\mathcal T[2,11]$. 
    This suffix shares a common prefix of length $\LPF_\pi[2] = 1$ with the previously-inserted suffix $\mathcal T[1,11]$, hence position $2 + \LPF_\pi[2] = 2+1 = 3$ is sampled. This position is blue, hence a blue path labeled $\mathcal T[3,11]$ is inserted as a child of $\locus(\mathcal T[2]=A)$.}
    \label{fig:STCOLEX}
\end{figure}

We now show how to locate the pattern's occurrences with $\stcolexm$. In this section, among the various trade-offs summarized in Theorem \ref{thm:locate all general} for general order-preserving permutations, we will analyze the one obtained by assuming a text oracle supporting the extraction of $\ell$ contiguous text characters with $O(1+\ell/B)$ I/O complexity as it models the scenario of our practical implementation discussed in Section \ref{sec: general STPD index}.

\subsubsection{Finding the primary occurrence.}

For the particular order-preserving permutation $\pi = \IPA$ used in this section, the primary occurrence $P = \mathcal T[i,j]$ of $P$ is the one for which the text prefix $\mathcal T[1,j]$ is the colexicographically-smallest being suffixed by $P$.

We use Algorithm~\ref{alg: locate primary} to locate the primary occurrence. Since $\pi = \IPA$ (colexicographic rank) and $\PDA_\pi = \stcolexm$ is sorted colexicographically, observe that $\pi(\PDA_\pi) = \pi(\stcolexm)$ is increasing. This means that, in Line \ref{line:argmin} of Algorithm~\ref{alg: locate primary}, it always holds $\arg\min_{i'\in [b,e]}  \{ \pi(\PDA_\pi[i']) \} = b$ and therefore we do not need a Range Minimum Data structure over $\pi(\stcolexm)$.

By implementing $\sufsearch$ by binary search and random access, we obtain: 

\begin{lemma}\label{lem: locate primary stcolex}
    Let $\mathcal T\in \Sigma^n$ be a text, and fix $\pi = \IPA$.
    Suppose we have access to a text oracle supporting the extraction of $\ell$ contiguous characters with $O(1+\ell/B)$ I/O complexity. Algorithm~\ref{alg: locate primary} requires just array $\stcolexm$, fitting in $|\stcolexm| \le \bar r$ words, on top of the text oracle and locates the primary occurrence of $P\in \Sigma^m$ with $O(d\log|\stcolexm|\cdot(1+m/B)) \subseteq O(d\log \bar r \cdot (1+m/B))$ I/O complexity, where $d$ is the node depth of $\locus(P)$ in the suffix tree of $\mathcal T$.
\end{lemma}

\subsubsection{Locating the secondary occurrences.}

As it turns out, the mechanism for locating secondary occurrences described in Section \ref{sec: general STPD index} in the particular cases of $\pi = \ISA$ and $\pi = \IPA$ is equivalent to the $\bar\phi$ function of Definition \ref{def:phi}, with the only additional detail that for $\pi=\IPA$ one should replace in Definition \ref{def:phi} $\SA$ and $\ISA$ with $\PA$ and $\IPA$, respectively (below, with symbol $\bar\phi$ we denote this variant using $\PA$ and $\IPA$).

Then, our locating mechanism essentially reduces to the one of the $r$-index \cite{GNP20}, with only minor modifications that we describe below. 
The $r$-index \cite{GNP20} (of the reversed $\mathcal T$) locates pattern occurrences as follows. 
Let $\PA[b,b+occ-1]$ be the Prefix Array range of pattern $P$. The $r$-index (i) finds $\PA[b]$ and $occ$ with a mechanism called the \emph{toehold lemma}, and (ii) it applies $occ-1$ times function $\bar\phi$ 
starting from $\PA[b]$. This yields the sequence $\PA[b], \PA[b+1], \dots, \PA[b+occ-1]$, that is, the list of all pattern's occurrences. As shown by Nishimoto and Tabei \cite{Nishimoto21Move}, function $\bar\phi$ (in the variant using $\PA$ and $\IPA$) can be stored in $O(\bar r)$ memory words in such a way that the above $occ-1$ successive evaluations of $\bar \phi$ take $O(occ + \log\log (n/\bar r))$ time (and I/O complexity). 

Differently from the toehold lemma of the $r$-index, Algorithm~\ref{alg: locate primary} only allows to compute $\PA[b]$ (not $occ$). To also compute $occ$ and output all pattern's occurrences we proceed as follows.

The simplest solution  to locate all occurrences is to use both arrays $\stcolexm$ and $\stcolexp$ to compute $\PA[b]$ and $\PA[b+occ-1]$, respectively, and then extract one by one 
$\PA[b], \PA[b+1], \dots, \PA[b+occ-1]$ by applying $occ-1$ times function $\bar\phi$. This solution, however, requires also storing $\stcolexp$.

We can still locate efficiently all secondary occurrences using just $\stcolexm$, as follows. 
Let $\PA[b,b+occ-1]$ be the range of $P$ in the Prefix Array. 
Let $Q = \lceil m/B \rceil$ be a \emph{block size}. 
Partition the range $[b,b+occ-1]$ into $q = \lfloor occ/Q \rfloor$ equal-size blocks of size $Q$, i.e. $[b+(i-1)Q, b+i\cdot Q-1]$ for $i=1, \dots, q$, plus (possibly) a last block of size $occ\mod Q$. For $i=1, \dots, q$, reconstruct $\PA[b+(i-1)Q, b+i\cdot Q]$ by applying the $\bar\phi$ function $Q$ times starting from $\PA[b+(i-1)Q]$ (available from the previous iteration). 
Note that we do not know $q$ (since we do not know $occ$); in order to discover when we reach the last block of size $Q$ (i.e. the $q$-th block), after reconstructing $\PA[b+(i-1)Q, b+i\cdot Q]$ we compare $\mathcal T[i',i'+m-1] \stackrel{?}{=} P$ for $i' = \PA[b+i\cdot Q-1]$ (i.e., the last Prefix Array entry in the block). Then, we know that $i\le q$ if and only if this equality test succeeds. 
Since we realize that we reached block number $i=q$ only when extracting the $(q+1)$-th block of length $Q$, this means that the total number of applications of function $\bar \phi$ is bounded by $occ + Q \le occ + m/B + 1$ (constant time each, except the very first $\bar\phi(\PA[b])$, costing $O(\log\log(n/\bar r))$).

As far as the $q+1$ comparisons $\mathcal T[i',i'+m-1] \stackrel{?}{=} P$ are concerned, if the random access oracle supports the extraction of $\ell$ contiguous characters of $\mathcal T$ with $O(1+\ell/B)$ I/O complexity, then the I/O cost of these comparisons amounts to $O((1+occ/Q)\cdot (1+m/B)) = O(1+occ + m/B)$.

Note that, when incrementing the block number $i$, we can re-use the same memory ($Q$ words) allocated for $\PA[b+(i-1)Q, b+i\cdot Q]$ in order to store the new block of $Q$ Prefix Array entries. It follows that the above process uses in total $Q \le 1+m/B$ memory words of space.

If $rem = occ\mod Q > 0$, we are left to show how to reconstruct the last block $\PA[b+q\cdot Q,b+q\cdot Q + rem - 1] = \PA[b',e']$, of length $rem$. The idea is to simply use binary search on $\PA[b',b'+Q-1]$, which we already extracted in the $(q+1)$-th iteration. In each of the $O(\log Q)$ binary search steps, we compare $P$ with a substring of $\mathcal T$ of length $m$.
If the random access oracle supports the extraction of $\ell$ contiguous characters of $\mathcal T$ with $O(1+\ell/B)$ I/O complexity, 
finding this last (partial) block of $rem = occ\mod Q$ pattern occurrences via binary search costs $O((1+m/B)\cdot \log Q) = O((1+m/B)\lceil\log(1+m/B)\rceil)$ I/O complexity using $Q$ memory words of space. 

\subsubsection{Putting everything together.} Locating the $occ-1$ secondary occurrences of $P$ costs $O((1+m/B)\lceil\log(1+m/B)\rceil + occ)$ I/O complexity and requires $O(1+m/B)$ words of memory on top of the index at query time. Combining this with Lemma \ref{lem: locate primary stcolex}, we obtain: 

\begin{theorem}\label{thm: locate all colex}
 Let $\mathcal T\in \Sigma^n$ be a text. 
 Assume we have access to a text oracle supporting the extraction of $\ell$ contiguous characters of $\mathcal T$ with $O(1+\ell/B)$ I/O complexity. Our data structure locates all the $occ$ occurrences of $P\in \Sigma^m$ 
 with $O\left((d\log \bar r + \log (1+m/B))\cdot (1+m/B) + occ + \log\log(n/\bar r)\right)$ I/O complexity and uses $O(\bar r)$ memory words on top of the oracle. At query time, $O(m/B)$ further memory words are used.
\end{theorem}

Observe that the extra $O(m/B)$ memory words of space are negligible; most compressed indexes explicitly store (and/or receive as input) the full pattern at query time anyways, in $O(m)$ words.
The remaining $O(\bar r)$ words of space are spent to store array $\stcolexm$ (at most $|\stcolexm| \le \bar r$ words) and Nishimoto and Tabei's data structure \cite{Nishimoto21Move} storing function $\bar\phi$ (less than $2\bar r$ words using the optimized implementation of \cite{BertramMoveR24}).

\subsection{Identity ($\stpos$): leftmost pattern occurrence}\label{sec:stpos}
Another notable STPD is obtained by using the identity function $\pi(i) = i$, trivially satisfying the order-preserving property of Definition \ref{def:order-preserving pi}. We call the corresponding path decomposition array $\stposm$. The corresponding STPD is almost equal to the one in Figure \ref{fig:STCOLEX}, with the only difference being that the path starting in the root is $\mathcal T[1,11]$ (green) rather than $\mathcal T[11,11]$ (orange). In this particular example, the two STPD obtained using $\pi = id$ and $\pi = \IPA$ are essentially equal because (as the reader can easily verify) the colexicographically-smallest occurrence of any substring $\alpha$ of $\mathcal T$ is also the leftmost one (but of course, this is not always the case for any text $\mathcal T$). 

Similarly, $\stposp$ is the path decomposition array associated with the dual permutation $\bar\pi(i) = n-i+1$. 

The size of $\stposm$ is equal to the number of irreducible values in the Longest Previous Factor array $\LPF$, a new interesting repetitiveness measure that, to the best of our knowledge, has never been studied before. While we could not prove a theoretical bound for $|\stposm|$ and $|\stposp|$ in term of known repetitiveness measures\footnote{The techniques used by Kempa and Kociumaka \cite{kempa2020resolution} to prove that the number of irreducible $\PLCP$ values is bounded by $O(\delta\log^2n)$, do not apply to array $\LPF$.}, below we show that these measures are worst-case optimal, meaning that for every $p\ge 1$ we exhibit a family of strings with $p = \Theta(|\stposm|)$ requiring $O(p)$ words to be stored in the worst case.

Moreover, in Section \ref{sec:experiments} we show that in practice $|\stposm|$ and $|\stposp|$ are consistently smaller than $r$.

\begin{theorem}\label{thm:stpos wc optimal}
    For any integers $n\ge 1$ and $p \in [n]$, there exists a string family $\mathcal F$ of cardinality $\binom{n}{p}$ such that every string $\mathcal S\in \mathcal F$ is over alphabet of cardinality $\sigma = p+1$ and satisfies $|\stposm(\mathcal S)| = p+1$ and $|\mathcal S| \le np$.
    In particular, no compressor can compress every individual $\mathcal S\in \mathcal F$ in asymptotically less than $\log_2 \binom{n}{p} = p\log(n/p) + \Theta(p)$ bits. 
    The same holds for $|\stposp|$.
\end{theorem}
\begin{proof}
    Consider any integer set $\{x_1 > x_2 > \dots > x_p\} \subseteq [n]$. Encode the set as the string $\mathcal S = 0^{x_1}\cdot 1 \cdot 0^{x_2}\cdot 2 \cdots 0^{x_p}\cdot p$. Then, it is not hard to see that the $\LPF$ irreducible positions are $1$, $2$, and $i+1$ for all $i<n$ such that $\mathcal S[i]>0$ (in total, $p+1$ irreducible positions). This proves $|\stposm(\mathcal S)| = p+1$. Fix $n\ge 1$ and $p \in [n]$. Then, the family $\mathcal F_{n,p}$ of strings built as above contains $\binom{n}{p}$ elements and satisfies the conditions of the claim. 
    
    The proof for $\stposp$ is symmetric: just sort the set's elements in increasing order $\{x_1 < x_2 < \dots < x_p\} \subseteq [n]$ and build $\mathcal S$ as above. Then, the irreducible positions are $i=1,n$, and $i+1$ for all $i<n$ such that $\mathcal S[i]>0$ (in total, $p + 1$ irreducible positions).
\end{proof}

Theorem \ref{thm:stpos wc optimal} proves that $p=\Theta(|\stposm|)$ words of space are essentially worst-case optimal as a function of $n$ and $p$ to compress the string (the same holds for $p=|\stposp|$). This result is similar to that obtained in \cite{kociumaka2022toward} for the repetitiveness measure $\delta$.
Since by Theorem \ref{thm: compressing PDA} $O(|\stposm|\log|\mathcal S|)$ bits (in the theorem above, $O(p\log(np)) = O(p\log n)$ bits) are also sufficient to compress $\mathcal S$ (the same holds for $|\stposp|$), this indicates that $|\stposm|$ and $|\stposp|$ are two meaningful repetitiveness measures.

\subsubsection{Applications of $\stpos$.}

By running Algorithm~\ref{alg: locate primary} on arrays $\stposm$ and $\stposp$, one can quickly locate the leftmost and rightmost occurrences of $P$ in $\mathcal T$. This is a problem that is not easily solvable with the $r$-index (unless using a heavy context-free grammar of $O(r\log(n/r))$ words supporting range minimum queries on the Suffix Array --- see \cite{GNP20}).
We briefly discuss an application of $\stpos$ and observe interesting connections between this STPD, suffix tree construction, and the \emph{Prediction by Partial Matching} compression algorithm.  

\paragraph{Taxonomic classification of DNA fragments.}
The STPD discussed in this section finds an important application in taxonomic classification algorithms, which we plan to explore in a future publication. 
Imagine having a collection of genomes, organized in a phylogenetic tree.
The taxonomic classifiers Kraken \cite{wood2014kraken} and Kraken 2 \cite{wood2019improved} index the genomes in a given phylogenetic tree such that, given a DNA read (i.e., a short string over the alphabet $\{A,C,G,T\}$), they can map each $k$-mer (substring of length $k$) in that read to the root of the smallest subtree of the phylogenetic tree containing all the genomes containing that $k$-mer.  Considering the roots for all the $k$-mers in the read, the tool then tries to predict from what part of the tree the read came from (with a smaller subtree corresponding to a more precise prediction). The mapping is fairly easy with $k$-mers but becomes more challenging with variable-length substrings such as maximal exact matches (MEMs), for example.

A first step to generalizing the strategy to variable-length strings has been taken in \cite{GagieKatka2022}. In that paper, the authors concatenate the genomes in leaf order and index the resulting string with a structure able to return the leftmost and rightmost occurrence of all the $k$-mers of an input string, where (differently from \cite{wood2014kraken,wood2019improved}) $k$ is provided at query time rather than index construction time. At this point, for each $k$-mer a constant-time lowest-common-ancestor (LCA) query on the (at most) two genomes (tree nodes) containing those occurrences, yields the smallest subtree containing the $k$-mer. 


Using suffix tree path decompositions, we can easily extend this solution to any subset of strings of the query string (e.g.\ the set of al MEMs). 
Once we have a MEM, with $\stposm$ and $\stposp$ we can find the MEM's leftmost and rightmost occurrences, the genomes that contain those occurrences (with a simple predecessor data structure), and those genomes' LCA --- which is the root of the smallest subtree of the phylogenetic tree containing all occurrences of the MEM.  This means $\stposm$ and $\stposp$ finally offer the opportunity to generalize Kraken and Kraken 2's approach to taxonomic classification to work \emph{efficiently} (fast query times and fully-compressed space) with arbitrary-length strings such as MEMs. 


    \paragraph{Relation to Ukkonen's suffix tree construction algorithm.}
    We observe the following interesting relation between the irreducible values in $\LPF$ and 
    Ukkonen's suffix tree construction algorithm \cite{Ukkonen95}. 
    Ukkonen's algorithm builds the suffix tree by inserting the suffixes from the longest ($\mathcal{T}[1,n])$ to the shortest ($\mathcal{T}[n,n]$). Instead of inserting exactly one suffix per iteration like McCreight's algorithm \cite{McCreight76}, Ukkonen's algorithm inserts a variable number of suffixes per iteration.
    
    Consider the beginning of the $i$-th iteration, for $i=1,\dots,n$. The algorithm maintains a variable $j\in[n]$, indicating that we have already inserted suffixes $\mathcal{T}[1,n],\cdots\mathcal{T}[j-1,n]$ and we are at the locus of $\mathcal{T}[j,i-1]$ on the suffix tree. We have two cases. (i) If there already exists a locus corresponding to $\mathcal{T}[j,i]$, i.e., if there exists an out-edge labeled with $\mathcal{T}[i]$ from the current locus, we just walk down the tree accordingly taking child $\mathcal{T}[i]$. (ii) If there is no locus for $\mathcal{T}[j,i]$ in the tree, then we create a path labeled with $\mathcal{T}[i,n]$ (if $\mathcal{T}[j,i-1]$ ends in the middle of the label of a suffix tree edge, we also create an internal node and a suffix link to this new internal node if necessary), meaning that we are inserting the leaf that corresponds to $\mathcal{T}[j,n]$, then move to the locus for $\mathcal{T}[j+1,i-1]$ (following a suffix link), increment $j$ by $1$ and repeat (ii) until we fall into Case~(i) or until $j>i$. 

    When inserting a suffix $\mathcal{T}[j,n]$ during the $i$-th iteration, we are on the locus of $\mathcal{T}[j,i-1]$, meaning that there exists a suffix $\mathcal{T}[j',n]$ with $j'<j$ that is prefixed by $\mathcal{T}[j,i-1]$, but none of the suffixes $\mathcal{T}[j',n]$ with $j'<j$ is prefixed by $\mathcal{T}[j,i]$, which implies that $\LPF[j]=(i-1)-j+1=i-j$. Note that if more than one suffix is inserted in the same iteration $i$, they must be consecutive in terms of starting positions, i.e., those are $\mathcal{T}[j,n],\mathcal{T}[j+1,n],\cdots,\mathcal{T}[j+k,n]$ for some $k\in[n]$. Note that such suffixes form a decreasing run in $\LPF$ since $\LPF[j+k']=i-j-k'$ for $0\le k'\le k$. Moreover, observe that $(j+k')+\LPF[j+k']=i$ for $0\le k'\le k$, which implies they all those suffixes contribute to a single element $i\in\stposm$. 
    As a result, $\stposm$ can be interpreted as the set of iterations in Ukkonen's algorithm in which at least one leaf is inserted.

    \paragraph{Relation with PPM$^*$.}

    To conclude the section, we observe that $\stposm$ is tightly connected with the \emph{Prediction by Partial Matching} compression algorithm. This connection indicates why this STPD achieves very good compression in practice (see experimental results in Section \ref{sec:experiments}). 

    Consider how we decompose the suffix tree for $\mathcal T [1,n]$ when building $\stposm$.  If we choose a path from a node $u$ to a leaf labeled with SA entry $i$, then either
\begin{enumerate}
\item $u$ is the root, $i = 1$, the path's label is all of $\mathcal T$, and we add 1 to $\stposm$; or
\item $u$ is already in a path to a leaf labelled $h < i$.
\end{enumerate}
We focus on the second case.  Let $\ell > 0$ be the length of $u$'s path label,
\[\mathcal T [h,h + \ell - 1] = \mathcal T [i,i + \ell - 1]\,,\]
so we add $i + \ell$ to $\stposm$.

Notice $\mathcal T [i,i + \ell - 1]$ occurs starting at position $h < i$ but $\mathcal T [i,i + \ell]$ is the first occurrence of that substring.  Let $i' \leq i < i + \ell$ be the minimum value such that $\mathcal T [i',i + \ell - 1]$ occurs in $\mathcal T [1,i + \ell - 2]$; then $\mathcal T [i',i + \ell]$ does not occur in $\mathcal T [1,i + \ell - 1]$.  Setting $j = i + \ell$, we obtain the following corollary:

\begin{corollary}
\label{lem:subset}
\[\stposm
\subseteq \left\{\ j\ :\ \begin{array}{l}
	\mbox{either $j = 1$ or, for the minimum value $i' < j$} \\
	\mbox{such that $\mathcal T [i',j - 1]$ occurs in $\mathcal T [1,j - 2]$,} \\
	\mbox{$\mathcal T [i',j]$ does not occur in $\mathcal T [1,j - 1]$}
\end{array}\right\}\,.\]
\end{corollary}

\emph{Prediction by Partial Matching with unbounded context (PPM$^*$)} is a popular and effective compression algorithm that encodes each text's symbol $\mathcal T[j] = c$ according to the probability distribution of characters following the longest context $\mathcal T[i',j-1] = \alpha$ preceding it that already occurred before in the text followed by $c$. Here we consider the simple version of PPM$^*$ lengths~\cite{cleary1997unbounded} shown in Algorithm~\ref{algorithm: PPM}. When all previous occurrences of $\alpha$ are never followed by $c$ (that is, the model cannot assign a nonzero probability for character $c$ given context $\alpha$), the algorithm outputs a special \emph{escape} symbol $\bot$ followed by character $c$. 

Consider the set of positions at which we emit copies of the escape symbol $\bot$ when encoding $\mathcal T$ with Algorithm~\ref{algorithm: PPM}. By inspection of the algorithm, that set is the same as the one on the right side of the inequality in Lemma~\ref{lem:subset}. We conclude: 

\begin{lemma}
    $|\stposm|$ is upper-bounded by the number of escape symbols $\bot$ output by the PPM$^*$ algorithm described in~\cite{cleary1997unbounded} (Algorithm \ref{algorithm: PPM}).
\end{lemma}

\bigskip

\begin{algorithm}
\caption{A simple version of PPM$^*$}\label{algorithm: PPM}
Encode $\mathcal{T}[1]$ as $\bot \mathcal{T}[1]$\;

\For{$j =2\ldots n$}{
    $i' \gets \arg\min_{i'} \left\{i' \le j: \mathcal{T}[i',j-1] \text{ occurs in } \mathcal{T}[1,j-2] \right\}$\;
    \eIf{$\mathcal{T}[i',j]$ occurs in $\mathcal{T}[1,j-1]$}{
        Encode $\mathcal{T}[j]$ according to the distribution of characters
        immediately following occurrences of $\mathcal{T}[i',j-1]$\;
    }{
        Encode $\mathcal{T}[j]$ as $\bot \mathcal{T}[j]$\;
    }
}
\end{algorithm}